\documentclass[conference]{IEEEtran}

\usepackage[T5,T1]{fontenc}
\usepackage{xspace}
\usepackage{authblk}
\usepackage[frozencache=true,cachedir=.]{minted}
\setminted{fontsize=\footnotesize,
  escapeinside=||,
  mathescape=true}
\usepackage{microtype}                

\usepackage{float}                    
\usepackage{placeins}                 

\usepackage{hyperref}                 
\hypersetup{
    colorlinks,
    linkcolor={red!50!black},
    citecolor={blue!50!black},
    urlcolor={blue!80!black}
}

\usepackage{amsmath}                  
\usepackage{amsfonts}                 
\usepackage{amssymb}                  
\usepackage{amsthm}                   
\usepackage{mathtools}                

\usepackage{cleveref}

\newtheorem*{fact*}{Fact}                       
\newtheorem{definition}{Definition}[section]    
\newtheorem*{definition*}{Definition}           
  % ...
\newtheorem*{proposition*}{Proposition}         
\newtheorem{theorem}{Theorem}[section]          
\newtheorem*{theorem*}{Theorem}                 
\newtheorem{lemma}[theorem]{Lemma}              
\newtheorem*{lemma*}{Lemma}                     
        % ...
\newtheorem*{sublemma*}{Sublemma}               
\newtheorem{corollary}{Corollary}[theorem]      
\newtheorem*{corollary*}{Corollary}             

\usepackage[dvipsnames]{xcolor}       

\usepackage{stmaryrd}                 
\usepackage{pifont}                   
\usepackage{textcomp}                 

\usepackage{centernot}                

\usepackage{stackengine}              
\stackMath

\usepackage{nicefrac}                 

\usepackage{array}                    
\usepackage{tabularx}                 

\usepackage{framed}                   

\usepackage{lipsum}                   

\usepackage{enumitem}                 

\usepackage{mathpartir}               

\usepackage{galois}                   

\usepackage{tikz}

\usepackage{relsize}

\usepackage[normalem]{ulem}
\usepackage{comment}

\newcommand{\mtext}[1]{\ifmmode\operatorname{#1}\else\textnormal{#1}\fi}
\newcommand{\mtexttt}[1]{\ifmmode\operatorname{\mathtt{#1}}\else\textnormal{\texttt{#1}}\fi}
\newcommand{\mtextit}[1]{\ifmmode\operatorname{\mathit{#1}}\else\textnormal{\textit{#1}}\fi}
\newcommand{\mtextbf}[1]{\ifmmode\operatorname{\mathbf{#1}}\else\textnormal{\textbf{#1}}\fi}
\newcommand{\mtextsc}[1]{\ifmmode\operatorname{\textsc{\smaller #1}}\else\textnormal{\textsc{\smaller #1}}\fi}

\definecolor{cbsafeABright}{RGB}{8,72,145}
\definecolor{cbsafeADark}{RGB}{109,36,150}

\definecolor{cbsafeBBright}{RGB}{85,119,13}
\definecolor{cbsafeBDark}{RGB}{109,70,17}

\definecolor{cbsafeCBright}{RGB}{150,48,89}
\definecolor{cbsafeCDark}{RGB}{70,24,48}

\newcommand{\colorMATHA}{cbsafeABright}
\newcommand{\colorSYNTAXA}{cbsafeADark!80!black}

\newcommand{\colorTEXT}{black}

\newcommand{\colorMATH}{\colorMATHA}
\newcommand{\colorSYNTAX}{\colorSYNTAXA}

\renewcommand{\paragraph}[1]{\vspace{5pt}\noindent\textbf{#1}}

\definecolor{mygray}{gray}{0.97}
\BeforeBeginEnvironment{minted}{\vspace{-1em}}
\AfterEndEnvironment{minted}{\vspace{-1em}}
\setminted{bgcolor=white}

\newcommand{\dduo}{\textsc{DDuo}\xspace}
\newcommand{\wrapper}{\texttt{Sensitive}\xspace}

\newcommand{\jn}[1]{{\color{orange}[Joe] #1}}

\newcommand{\fuzz}{\textsc{Fuzz}\xspace}
\newcommand{\pinq}{\textsc{PINQ}\xspace}
\newcommand{\dfuzz}{\textsc{DFuzz}\xspace}
\newcommand{\fuzzi}{\textsc{Fuzzi}\xspace}

\newcommand{\apRHL}{\textsc{apRHL}\xspace}
\newcommand{\apRHLplus}{\textsc{apRHL$^+$}\xspace}
\newcommand{\spanapRHL}{\textsc{span-apRHL}\xspace}

\title{\dduo: General-Purpose Dynamic Analysis for Differential Privacy}

\author[1]{Chike Abuah}
\author[1]{Alex Silence}
\author[2]{David Darais}
\author[1]{Joe Near}

\affil[1]{Computer Science Department, University of Vermont}
\affil[2]{Galois, Inc.}

\begin{document}
\maketitle
\thispagestyle{plain}
\pagestyle{plain}

\begin{abstract}
Differential privacy enables general statistical analysis of data with formal guarantees
of privacy protection at the individual level. Tools that assist data analysts with utilizing
differential privacy have frequently taken the form
of programming languages and libraries. However, many existing programming languages
designed for compositional verification of differential privacy impose significant burden on the programmer (in the form of complex type annotations). Supplementary library support for privacy analysis built on top of existing general-purpose languages has been more usable, but incapable of
pervasive end-to-end enforcement of sensitivity analysis and privacy composition.

We introduce \dduo, a \emph{dynamic analysis} for enforcing differential privacy.
\dduo is usable by non-experts: its analysis is automatic and it requires no additional type annotations. \dduo can be implemented \emph{as a library} for existing programming languages; we present a reference implementation in Python which features moderate runtime overheads on realistic workloads. We include support for several data types, distance metrics and operations
which are commonly used in modern machine learning programs. We also provide initial support for
tracking the sensitivity of data transformations in popular Python libraries for data analysis.

We formalize the novel core of the \dduo system and prove it sound for sensitivity
analysis via a logical relation for metric preservation.
We also illustrate \dduo's usability and flexibility through various case studies which implement state-of-the-art machine learning algorithms.
\end{abstract}

\section{Introduction}

Differential privacy has achieved prominence over the past decade as a rigorous formal
foundation upon which diverse tools and mechanisms for performing private data analysis can be
built. The guarantee of differential privacy is that it protects privacy at the individual level: if the result of a differentially private query or operation on a dataset is publicly
released, any individual present in that dataset can claim \emph{plausible deniability}. This means that
any participating individual can deny the presence of their information in the dataset based on the query result, because differentially private queries introduce enough random noise to make the result indistinguishable from that of the same query run on a dataset which actually \emph{does not} contain the individual's information. Additionally, differential privacy guarantees are resilient against any form of \emph{linking attack} in the presence of \emph{auxiliary information} about individuals.

High profile tech companies such as Google have shown a commitment to differential privacy by developing projects such as RAPPOR \cite{rappor} as well as several open-source privacy-preserving technologies \cite{guevara_2019,guevara_2020,wilson2019}.
Facebook recently released an unprecedented social dataset, protected by differential privacy guarantees, which contains information regarding people who publicly shared and engaged with about 38 million unique URLs, as an effort to help researchers study social media's impact on democracy and the 2020 United States presidential election \cite{nayak_2020,kifer2020,king_persily_2020,evans646717,evans643747}. The US Census Bureau has also adopted differential privacy to safeguard the 2020 census results \cite{abowd2018}.

Both static and dynamic tools have been developed to help non-experts write differentially private programs.
Many of the static tools take the form of statically-typed programming languages, where
correct privacy analysis is built into the soundness of the type system. However, existing language-oriented tools for compositional verification of differential privacy impose significant burden on the programmer (in the form of additional type annotations) \cite{reed2010distance, gaboardi2013linear, near2019duet, de2019probabilistic, zhang2019fuzzi, Winograd-CortHR17, DBLP:journals/lmcs/BartheEHSS19, Barthe:POPL12,Barthe:TOPLAS:13, Barthe:LICS16, Sato:LICS19, DBLP:journals/pacmpl/AlbarghouthiH18, zhang2017lightdp, wang2019proving, bichsel2018dp, ding2018detecting, wang2020checkdp} (see Section~\ref{sec:related} for a longer discussion).

The best-known dynamic tool is \pinq \cite{mcsherry2009}, a dynamic analysis for sensitivity and privacy. It features an extensible system which allows non-experts
in differential privacy to execute SQL-like queries against relational databases. However, \pinq comes with
several restrictions that limit its applicability. For example, \pinq's expressiveness is limited to a subset of the SQL language
for relational databases. Methods in \pinq are assumed to be side-effect free, which is necessary
to preserve their privacy guarantee.

We introduce \dduo, a {dynamic analysis} for enforcing differential privacy. \dduo is usable by non-experts: its analysis is automatic and it requires no additional type annotations. \dduo can be implemented \emph{as a library} for existing programming languages; we present a reference implementation in Python.
Our goal in this work is to answer the following four questions, based on the limitations of \pinq:
\begin{itemize}[topsep=1mm,leftmargin=5mm]
  \item Can a PINQ-style dynamic analysis extend to base types in the
    programming language, to allow its use pervasively?
  \item Is the analysis sound in the presence of side effects?
  \item Can we use this style of analysis for complex algorithms like differentially private gradient descent?
  \item Can we extend the privacy analysis beyond pure $\epsilon$-differential privacy?
\end{itemize}

\noindent We answer all four questions in the affirmative, building on \pinq in the following ways:

\begin{itemize}[topsep=1mm,leftmargin=5mm]
  \item \dduo provides a dynamic analysis for base types in a general purpose language
 (Python). \dduo supports general language operations, such as mapping arbitrary
 functions over lists, and tracks the sensitivity (stability) and privacy throughout.
  \item Methods in \dduo are not required to be side-effect free and allow
  programmers to mutate references inside functions which manipulate sensitive values.
  \item \dduo supports for various notions of sensitivity and arbitrary distance metrics (including $L_1$ and $L_2$ distance).
  \item \dduo is capable of leveraging advanced privacy
  variants such as ({{\color{\colorMATH}\ensuremath{\mathit{\epsilon , \delta }}}}) and R\'enyi differential privacy.
\end{itemize}

\noindent Privacy analysis is reliant on \emph{sensitivity} analysis, which determines the scale of noise an analyst must add to values in order to achieve any level of privacy. Dynamic analysis for differential privacy is thus a dual challenge:

\paragraph{Dynamic sensitivity analysis.} Program sensitivity is a (hyper)property quantified over two runs of a program with related inputs (sources). A major challenge for dynamic sensitivity analysis is the ability to bound sensitivity, ensuring that the metric preservation property is satisfied, by only observing \emph{a single run} of the program. In addition, an analysis which is performed on a specific input to the program must generalize to future possible arbitrary inputs.

The key insight to our solution is attaching sensitivity environments and distance metric information to \emph{values} rather than variables. Our approach provides a sound upper bound on global sensitivity even in the presence of side effects, conditionals, and higher-order functions.
We present a proof using a \emph{step-indexed logical relation} which shows that our sensitivity analysis is sound.

\paragraph{Dynamic privacy analysis.} To implement a dynamic privacy analysis, we leverage prior work on privacy filters and odometers~\cite{RogersVRU16}. This work, originally designed for the adaptive choice of privacy parameters, can also be used as part of a dynamic analysis for privacy analysis. We view each application of a privacy mechanism (e.g. the Laplace mechanism) as a \emph{global privacy effect} on total privacy cost, and use privacy filters and odometers to track total privacy cost.

We implemented these features in a Python prototype of \dduo via object proxies and other pythonic idioms. We implement several case studies to showcase these features and demonstrate the usage of \dduo in practice. We also provide integrations with several popular Python libraries for data and privacy analysis.

\paragraph{Contributions.}
In summary, this paper makes the following contributions:

\begin{itemize}[label=\textbf{-},leftmargin=1em]\item  We introduce \dduo, a \emph{dynamic} analysis for enforcing differential privacy, and a reference implementation as a Python library \footnote{The (non-anonymized) implementation is available to reviewers on request}.
\item  We formalize a subset of \dduo in a core language model, and prove the soundness of \dduo's dynamic sensitivity analysis (as encoded in the model) using a step-indexed logical relation.
\item  We present several case studies demonstrating the use of \dduo to build practical, verified Python implementations of complex differentially private algorithms.
\end{itemize}

The rest of the paper is organized as follows. First we provide some background knowledge regarding the field of differential privacy (Section~\ref{sec:background}). We provide an overview of our work and the necessary tradeoffs (Section~\ref{sec:overview}). We illustrate the usefulness and power of \dduo through some worked examples (Section~\ref{sec:dduo-example}). We then discuss some of the nuances of dynamic sensitivity (Section~\ref{sec:sensitivity}) and dynamic privacy tracking (Section~\ref{sec:privacy}). We present the formalization of \dduo and prove the soundness of our sensitivity analysis (Section~\ref{sec:formalism}). We provide several case studies demonstrating the usefulness of \dduo in practice (Section~\ref{sec:case}). Finally we outline related work (Section~\ref{sec:related}) and conclude (Section~\ref{sec:conclusion}).

\section{Background}
\label{sec:background}

\paragraph{Differential Privacy.}
Differential privacy is a formal notion of privacy; certain algorithms (called \emph{mechanisms}) can be said to \emph{satisfy} differential privacy. Intuitively, the idea behind a differential privacy mechanism is that: given inputs which differ in the data of a single individual, the mechanism should return statistically indistinguishable answers.
This means that the data of any one individual should not have any significant effect on the outcome of the mechanism, effectively protecting privacy on the individual level. Formally, differential privacy is parameterized by the privacy parameters {{\color{\colorMATH}\ensuremath{\mathit{\epsilon , \delta }}}} which control the strength of the guarantee.

When we say "neighboring" inputs, this implies two inputs that differ in the information of a single individual. However, formally we can defer to some general \emph{distance metric} which may take several forms. We then say that according to the distance metric, the distance between two databases must have an upper bound of 1. Variation of the distance metric has led to several other useful, non-standard forms of differential privacy in the literature.

\begin{definition}[Differential privacy]
  Given a distance metric {{\color{\colorMATH}\ensuremath{\mathit{d_{A} \in  A \times  A \rightarrow  {\mathbb{R}}}}}}, a randomized algorithm (or \emph{mechanism}) {{\color{\colorMATH}\ensuremath{\mathit{{\mathcal{M}} \in  A \rightarrow  B }}}} satisfies ({{\color{\colorMATH}\ensuremath{\mathit{\epsilon , \delta }}}})-differential privacy if for all {{\color{\colorMATH}\ensuremath{\mathit{x, x^{\prime} \in  A }}}} such that {{\color{\colorMATH}\ensuremath{\mathit{d_{A}(x, x^{\prime}) \leq  1}}}} and all possible sets {{\color{\colorMATH}\ensuremath{\mathit{S \subseteq  B }}}} of outcomes, {{\color{\colorMATH}\ensuremath{\mathit{ {\text{Pr}}[{\mathcal{M}}(x) \in  S] \leq  e^{\epsilon } {\text{Pr}}[{\mathcal{M}}(x^{\prime}) \in  S] + \delta  }}}}.
\end{definition}

Differential privacy is \emph{compositional}: running two mechanisms {{\color{\colorMATH}\ensuremath{\mathit{{\mathcal{M}}_{1}}}}} and {{\color{\colorMATH}\ensuremath{\mathit{{\mathcal{M}}_{2}}}}} with privacy costs of {{\color{\colorMATH}\ensuremath{\mathit{(\epsilon _{1}, \delta _{1})}}}} and {{\color{\colorMATH}\ensuremath{\mathit{(\epsilon _{2}, \delta _{2})}}}} respectively has a total privacy cost of {{\color{\colorMATH}\ensuremath{\mathit{(\epsilon _{1}+\epsilon _{2}, \delta _{1}+\delta _{2})}}}}.
\emph{Advanced composition}~\cite{dwork2014algorithmic} improves on this composition bound for iterative algorithms; several variants of differential privacy (e.g. R\'enyi differential privacy~\cite{Mironov17} and zero-concentrated differential privacy~\cite{bun2016concentrated}) have been developed that improve the bound even further. Importantly, sequential composition theorems for differential privacy do not necessarily allow the privacy parameters to be chosen \emph{adaptively}, which presents a special challenge in our setting---we discuss this issue in Section~\ref{sec:privacy}.


\section{Overview of \dduo}
\label{sec:overview}

\dduo is a \emph{dynamic} analysis for enforcing differential privacy. Our approach does not require static analysis of programs, and allows \dduo to be implemented as a library for programming languages like Python. \dduo's dynamic analysis has complete access to run-time information, so it does not require the programmer to write any additional type annotations---in many cases, \dduo can verify differential privacy for essentially {unmodified} Python programs (see the case studies in Section~\ref{sec:case}). As a Python library, \dduo is easily integrated with popular libraries like Pandas and NumPy.

\paragraph{Threat model.}
We assume an ``honest but fallible'' programmer---that is, the programmer \emph{intends} to produce a differentially private program, but may unintentionally introduce bugs. We assume that the programmer is \emph{not} intentionally attempting to subvert \dduo's enforcement approach. Our reference implementation is embedded in Python, an inherently dynamic language with run-time features like reflection. In this setting, a malicious programmer or privacy-violating third-party libraries can bypass our dynamic monitor and extract sensitive information directly.
Data-independent exceptions can be safely used in our system, however our model must
explicitly avoid data-dependent exceptions such as division-by-zero errors. Terminated programs can be rerun safely (while consuming the privacy budget) because our analysis is independent of any sensitive information (our metatheory implies that sensitivity of a value is itself not sensitive).
We also do not address side-channels, including execution time. Like existing enforcement approaches (PINQ, OpenDP, Diffprivlib), \dduo is intended as a tool to help well-intentioned programmers produce correct differentially private algorithms.

\paragraph{Soundness of the analysis.}
We formalize our dynamic sensitivity analysis and prove its soundness in Section~\ref{sec:formalism}. Our formalization includes the most challenging features of the dynamic setting---conditionals and side effects---and provides evidence that our Python implementation will be effective in catching privacy bugs in real programs. \dduo relies on existing work on privacy filters and odometers (discussed in Section~\ref{sec:privacy}), whose soundness has been previously established, for tracking privacy cost.

\section{\dduo by Example}
\label{sec:dduo-example}

This section introduces the \dduo system via examples written using
our reference Python implementation.

\paragraph{Data Sources.}
Data sources are wrappers around sensitive data that enable tracking of privacy information in the \dduo python library.
Each data source is associated with an identifying string, such as the name of the input file the data was read from.
Data sources can be created manually by attaching an identifying string (such as a filename) to a raw value (such as a vector).
Or, data sources be created automatically upon loading data through \dduo's custom-wrapped third party APIs, such as pandas. Note that our API can be easily modified to account for initial sensitivities greater than 1 when users have multiple datapoints in the input data.

\vspace{1em}
\begin{minted}{python}
from dduo import pandas as pd
df = pd.read_csv("data.csv")
df
\end{minted}
\begin{minted}[frame=lines,bgcolor=mygray]{text}
Sensitive(<'DataFrame'>, |$\{ \textit{data.csv} \mapsto  1\} $|, |$L_{\infty }$|)
\end{minted}

A \wrapper value is returned. \wrapper values represent sensitive information that cannot be viewed by the analyst. When a \wrapper value is printed out, the analyst sees (1) the type of the value, (2) its \emph{sensitivity environment}, and (3) its \emph{distance metric}. The latter two components are described next. The analyst is \emph{prevented} from viewing the value itself.



\paragraph{Sensitivity \& distance metrics.}
Function sensitivity is a scalar value which represents how much a change in a function's input will change the function's output.
For example, the binary addition function {{\color{\colorMATH}\ensuremath{\mathit{f(x,y) = x + y}}}} is 1-sensitive in both {{\color{\colorMATH}\ensuremath{\mathit{x}}}} and {{\color{\colorMATH}\ensuremath{\mathit{y}}}}, because changing either input by {{\color{\colorMATH}\ensuremath{\mathit{n}}}} will change the sum by {{\color{\colorMATH}\ensuremath{\mathit{n}}}}. The function {{\color{\colorMATH}\ensuremath{\mathit{f(x) = x + x}}}}, on the other hand, is 2-sensitive in its argument {{\color{\colorMATH}\ensuremath{\mathit{x}}}}, because changing {{\color{\colorMATH}\ensuremath{\mathit{x}}}} by {{\color{\colorMATH}\ensuremath{\mathit{n}}}} changes the function's output by {{\color{\colorMATH}\ensuremath{\mathit{2n}}}}.
Sensitivity is key to differential privacy because it is directly proportional to the amount of noise we must add to the output of a function to make it private.


\dduo tracks the sensitivity of a value to changes in the program's
inputs using a \emph{sensitivity environment} mapping input data
sources to sensitivities. Our example program returned a \wrapper
value with a sensitivity environment of {{\color{\colorMATH}\ensuremath{\mathit{\{ {\textit{data.csv}} \mapsto  1\} }}}}, indicating
that the underlying value is 1-sensitive in the data contained in
{{\color{\colorMATH}\ensuremath{\mathit{{\textit{data.csv}}}}}}.
The \dduo library tracks and updates the sensitivity environments of
\wrapper objects as operations are applied to them. For example,
adding a constant value to the elements of the DataFrame results in no
change to the sensitivity environment.



\vspace{5pt}
\begin{minted}{python}
df + 5 # no change to sensitivity environment
\end{minted}

\vspace{-1em}
\begin{minted}[frame=lines,bgcolor=mygray]{text}
Sensitive(<'DataFrame'>, |$\{ \textit{data.csv} \mapsto  1\} $|, |$L_{\infty }$|)
\end{minted}

\noindent Adding the DataFrame to \emph{itself} doubles the
sensitivity, in the same way as the function {{\color{\colorMATH}\ensuremath{\mathit{f(x) = x + x}}}}.

\vspace{5pt}
\begin{minted}{python}
df + df  # doubles the sensitivity
\end{minted}

\vspace{-1em}
\begin{minted}[frame=lines,bgcolor=mygray]{text}
Sensitive(<'DataFrame'>, |$\{ \textit{data.csv} \mapsto  2\} $|, |$L_{\infty }$|)
\end{minted}

\noindent Finally, multiplying the DataFrame by a constant scales the
sensitivity, and multiplying the DataFrame by \emph{itself} results in
\emph{infinite} sensitivity.

\vspace{5pt}
\begin{minted}{python}
( df * 5, df * df)
\end{minted}

\vspace{-1em}
\begin{minted}[frame=lines,bgcolor=mygray]{text}
( Sensitive(<'DataFrame'>, |$\{ \textit{data.csv} \mapsto  5\} $|,  |$L_{\infty }$|),
  Sensitive(<'DataFrame'>, |$\{ \textit{data.csv} \mapsto  \infty \} $|, |$L_{\infty }$|) )
\end{minted}



The \emph{distance metric} component of a \wrapper value describes
\emph{how} to measure sensitivity. For simple numeric functions like
{{\color{\colorMATH}\ensuremath{\mathit{f(x) = x + x}}}}, the distance between two possible inputs {{\color{\colorMATH}\ensuremath{\mathit{x}}}} and {{\color{\colorMATH}\ensuremath{\mathit{x^{\prime}}}}}
is simply {{\color{\colorMATH}\ensuremath{\mathit{|x - x^{\prime}|}}}} (this is called the \emph{cartesian metric}). For
more complicated data structures (e.g. DataFrames), calculating the
distance between two values is more involved. The {{\color{\colorMATH}\ensuremath{\mathit{L_{\infty }}}}} metric used
in our example calculates the distance between two DataFrames by
measuring how many rows are different (this is one standard way of
defining ``neighboring databases'' in differential privacy). \dduo's
handling of distance metrics is detailed in
Section~\ref{sec:distance-metrics}.


\paragraph{Privacy.}
\dduo also tracks the \emph{privacy} of computations. To achieve
differential privacy, programs add noise to sensitive values. The
Laplace mechanism described earlier is one basic mechanism for
achieving differential privacy by adding noise drawn from the Laplace
distribution (\dduo provides a number of basic mechanisms, including
the Gaussian mechanism). The following expression counts the number of
rows in our example DataFrame and uses the Laplace mechanism to
achieve {{\color{\colorMATH}\ensuremath{\mathit{\epsilon }}}}-differential privacy, for {{\color{\colorMATH}\ensuremath{\mathit{\epsilon  = 1.0}}}}.

\vspace{1em}
\begin{minted}{python}
dduo.laplace(df.shape[0], |$\epsilon $|=1.0)
\end{minted}

\vspace{-1em}
\begin{minted}[frame=lines,bgcolor=mygray]{text}
9.963971319623278
\end{minted}

\noindent The result is a \emph{regular Python value}---the analyst is
free to view it, write it to a file, or do further computation on
it. Once the correct amount of noise has been added, the principle of
\emph{post-processing} applies, and so \dduo no longer needs to track
the sensitivity or privacy cost of operations on the value.

When the Laplace mechanism is used multiple times, their privacy costs
\emph{compose} (i.e. the {{\color{\colorMATH}\ensuremath{\mathit{\epsilon }}}}s ``add up'' as described earlier). \dduo
tracks \emph{total} privacy cost using objects called \emph{privacy
  odometers}~\cite{RogersVRU16}. The analyst can interact with a
privacy odometer object to learn the total privacy cost of a complex
computation.

\vspace{1em}
\begin{minted}{python}
with dduo.EpsOdometer() as odo:
  _ = dduo.laplace(df.shape[0], |$\epsilon $| = 1.0)
  _ = dduo.laplace(df.shape[0], |$\epsilon $| = 1.0)
  print(odo)
\end{minted}

\vspace{-1em}
\begin{minted}[frame=lines,bgcolor=mygray]{text}
Odometer_|$\epsilon $|(|$\{ \textit{data.csv} \mapsto  2.0\} $|)
\end{minted}

\noindent Printing the odometer's value allows the analyst to view the
privacy cost of the program with respect to each of the data sources
used in the computation. In this example, two differentially private
approximations of the number of rows in the dataframe {{\color{\colorMATH}\ensuremath{\mathit{{\textit{df}}}}}} are
computed, each with a privacy cost of {{\color{\colorMATH}\ensuremath{\mathit{\epsilon  = 1.0}}}}. The total privacy
cost of running the program is therefore {{\color{\colorMATH}\ensuremath{\mathit{2\mathord{\cdotp }\epsilon  = 2.0}}}}.

\dduo also allows the analyst to place upper bounds on total privacy
cost (i.e. a privacy \emph{budget}) using privacy
\emph{filters}~\cite{RogersVRU16}. Privacy odometers and filters are
discussed in detail in Section~\ref{sec:privacy}.

\section{Dynamic Sensitivity Tracking}
\label{sec:sensitivity}

\dduo implements a \emph{dynamic sensitivity analysis} by wrapping
values in \wrapper objects and calculating sensitivities
as operations are performed on these objects. Type systems for
sensitivity~\cite{reed2010distance,gaboardi2013linear} construct a sensitivity environment for each
program expression; in the static analysis setting, a sensitivity
environment records the expression's sensitivity with respect to each
of the variables currently in scope.

\dduo attaches sensitivity environments to \emph{values} at runtime:
each \wrapper object holds both a value and its sensitivity
environment. As described earlier, \dduo's sensitivity environments
record a value's sensitivity with respect to each of the program's
data sources.
Formally, the sensitivity of a single-argument function {{\color{\colorMATH}\ensuremath{\mathit{f}}}} in its input is defined as:
\begingroup\color{\colorMATH}\begin{gather*}{\text{sens}}(f) \triangleq  {\text{argmax}}_{x,y}\Big(\frac{d(f(x), f(y))}{d(x,y)}\Big) \end{gather*}\endgroup
\noindent Where {{\color{\colorMATH}\ensuremath{\mathit{d}}}} is a \emph{distance metric} over the values {{\color{\colorMATH}\ensuremath{\mathit{x}}}}
and {{\color{\colorMATH}\ensuremath{\mathit{y}}}} could take (distance metrics are discussed in
Section~\ref{sec:distance-metrics}). Thus, a sensitivity environment
{{\color{\colorMATH}\ensuremath{\mathit{\{ a \mapsto  1\} }}}} means that if the value of the program input {{\color{\colorMATH}\ensuremath{\mathit{a}}}} changes by
{{\color{\colorMATH}\ensuremath{\mathit{n}}}}, then the value of {{\color{\colorMATH}\ensuremath{\mathit{f(a)}}}} will change by at most {{\color{\colorMATH}\ensuremath{\mathit{n}}}}.

\subsection{Bounding the Sensitivity of Operations}

Operations on \wrapper objects are defined to perform the same
operation on the underlying values, and \emph{also} construct a new
sensitivity environment for the operation's result. For example,
\dduo's \texttt{\_\_add\_\_} operation sums both the underlying values
\emph{and} their sensitivity environments:
\begin{minted}{python}
def __add__(self, other):
    assert self.metric == other.metric
    return dduo.Sensitive(self.value + other.value,
                            self.senv + other.senv,
                            self.metric)
\end{minted}

\noindent The sum of two sensitivity environments is defined as the
element-wise sum of their items. For example:
\begingroup\color{\colorMATH}\begin{gather*} \{ a \mapsto  2, b \mapsto  1 \}  + \{  b \mapsto  3, c \mapsto  5 \}  = \{  a \mapsto  2, b \mapsto  4, c \mapsto  5 \}  \end{gather*}\endgroup
The \dduo library provides sensitivity-aware versions of Python's
basic numeric operations (formalized in
Section~\ref{sec:formalism}). We have also defined
sensitivity-aware versions of commonly-used library functions,
including the Pandas functions used in Section~\ref{sec:dduo-example},
and subsets of NumPy and Scikit-learn.

\subsection{Distance Metrics}
\label{sec:distance-metrics}

At the core of the concept of sensitivity is the notion of distance:
how far apart we consider two information sources to be from each
other. For scalar values, the following two distance metrics are often
used:

\begin{itemize}
\item Cartesian (absolute difference) metric: {{\color{\colorMATH}\ensuremath{\mathit{d(x, y) = |x - y|}}}}
\item Discrete metric: {{\color{\colorMATH}\ensuremath{\mathit{d(x, y) = 0\hspace*{0.33em}{\textit{if}}\hspace*{0.33em}x = y; 1\hspace*{0.33em}{\textit{otherwise}}}}}}
\end{itemize}

For more complex structures---like lists and dataframes---we can use
distance metrics on \emph{vectors}. Two commonly-used metrics for
vectors {{\color{\colorMATH}\ensuremath{\mathit{x}}}} and {{\color{\colorMATH}\ensuremath{\mathit{y}}}} of equal length are:

\begin{itemize}
\item {{\color{\colorMATH}\ensuremath{\mathit{L_{1}(d_{i})}}}} metric: {{\color{\colorMATH}\ensuremath{\mathit{d(x, y) = \sum \limits_{x_{i}, y_{i} \in  x, y} d_{i}(x_{i}, y_{i})}}}}
\item {{\color{\colorMATH}\ensuremath{\mathit{L_{2}(d_{i})}}}} metric: {{\color{\colorMATH}\ensuremath{\mathit{d(x, y) = \sqrt {\sum \limits_{x_{i}, y_{i} \in  x, y} d_{i}(x_{i}, y_{i})^{2}}}}}}
\end{itemize}

\noindent Both metrics are parameterized by {{\color{\colorMATH}\ensuremath{\mathit{d_{i}}}}}, a metric for the
vector's elements. In addition to these two, we use the shorthand
{{\color{\colorMATH}\ensuremath{\mathit{L_{\infty }}}}} to mean {{\color{\colorMATH}\ensuremath{\mathit{L_{1}(d)}}}}, where {{\color{\colorMATH}\ensuremath{\mathit{d}}}} is the cartesian metric defined
above. The {{\color{\colorMATH}\ensuremath{\mathit{L_{\infty }}}}} metric works for any space with equality
(e.g. strings), and measures the \emph{number of elements where {{\color{\colorMATH}\ensuremath{\mathit{x}}}}
  and {{\color{\colorMATH}\ensuremath{\mathit{y}}}} differ}.

The definition of differential privacy is parameterized by a distance
metric that is intended to capture the idea of two inputs that
\emph{differ in one individual's data}. Database-oriented algorithms
typically assume that each individual contributes exactly one row to
the database, and use the {{\color{\colorMATH}\ensuremath{\mathit{L_{\infty }}}}} metric to define neighboring
databases (as we did in Section~\ref{sec:dduo-example}).

Distance metrics can be manipulated manually through operations such as clipping, a technique commonly employed in differentially private machine learning.
\dduo tracks distance metrics for \wrapper information, which can allow for automatic conservation of the
privacy budget while providing more accurate query analysis.




Lists and arrays are compared by one of the {{\color{\colorMATH}\ensuremath{\mathit{L_{1}}}}}, {{\color{\colorMATH}\ensuremath{\mathit{L_{2}}}}}, or {{\color{\colorMATH}\ensuremath{\mathit{L_{\infty }}}}} distance metrics. The choice of distance metric is important when defining sensitivity and thus privacy. For example, the Laplace mechanism can only be used with the {{\color{\colorMATH}\ensuremath{\mathit{L_{1}}}}} metric, while the Gaussian mechanism can be used with either {{\color{\colorMATH}\ensuremath{\mathit{L_{1}}}}} or {{\color{\colorMATH}\ensuremath{\mathit{L_{2}}}}}.

\subsection{Conditionals \& Side Effects}

Conditionals and other branching structures are challenging for any
sensitivity analysis, but they present a particular challenge for our
dynamic analysis. Consider the following conditional:

\vspace{1em}
\begin{minted}[escapeinside=||,mathescape=true]{python}
if df.shape[0] == 10:
  return df.shape[0]
else:
  return df.shape[0] * 10000
\end{minted}

Here, the two branches have \emph{different} sensitivities (the
\emph{else} branch is 10,000 times more sensitive in its data sources
than the \emph{then} branch). Static sensitivity analyses handle this
situation by taking the maximum of the two branches' sensitivities
(i.e. they assume the worst-case branch is executed), but this
approach is not possible in our dynamic analysis.

In addition, special care must be taken when a sensitive value appears
in the guard position (as in our example). Static analyses typically
scale the branches' sensitivity by the sensitivity of the guard; in
practice, this approach results in \emph{infinite sensitivity for
  conditionals with a sensitive guard}.

To retain soundness in our dynamic analysis, \dduo requires that
\emph{conditional guards contain no sensitive values}. A run-time
error is thrown if \dduo finds a sensitive value in the guard position
(as in our example above). Disallowing sensitive guards makes it
possible to ignore branches that are not executed: the guard's value
remains the same under neighboring program inputs, so the program
follows the same branch for neighboring executions. This approach does
not limit the set of useful programs we can write, since conditionals
with sensitive guards yield infinite sensitivities even under a
precise static analysis.

Since \dduo attaches sensitivity environments to \emph{values}
(instead of variables), the use of side effects does not affect the
soundness of the analysis. When a program variable is updated to
reference a new value, that value's sensitivity environment remains
attached. \dduo handles many common side-effect-based patterns used in
Python this way; for example, \dduo correctly infers that the
following program results in the variable \texttt{total} holding a
value that is 20 times more sensitive than \texttt{df.shape[0]}.

\vspace{1em}
\begin{minted}[escapeinside=||,mathescape=true]{python}
total = 0
for i in range(20):
  total = total + df.shape[0]
\end{minted}

For side effects, our dynamic analysis is more capable than type-based
static analysis, due to the additional challenges arising in the
static setting (e.g. aliasing). We have formalized the way \dduo
handles side effects and conditionals, and proved the soundness of our
sensitivity analysis; our formalization appears in
Section~\ref{sec:formalism}.






\section{Dynamic Privacy Tracking}
\label{sec:privacy}

\dduo tracks privacy cost \emph{dynamically}, at runtime. Dynamic
privacy tracking is challenging because the dynamic analysis has no
visibility into code that is \emph{not executed}. For example,
consider the following conditional:

\begin{minted}{python}
if dduo.gauss(|$\epsilon $|=1.0, |$\delta $|=1e-5, x) > 5:
  print(dduo.gauss(|$\epsilon $|=1.0, |$\delta $|=1e-5, y))
else:
  print(dduo.gauss(|$\epsilon $|=100000000000.0, |$\delta $|=1e-5, y))
\end{minted}

\noindent The executed branch of this conditional depends on the
result of the first call to \mintinline{python}{dduo.gauss}, which is
non-deterministic. The two branches use different privacy parameters
for the remaining calls to \mintinline{python}{dduo.gauss}; in other
words, the privacy parameter for the second use of the Gaussian
mechanism is chosen \emph{adaptively}, based on the results of the
first use. Sequential composition theorems for differential
privacy~\cite{dwork2014algorithmic} are typically stated in terms of
\emph{fixed} (i.e. non-adaptive) privacy parameters, and do not apply
if the privacy parameters are chosen adaptively.

A static analysis of this program will consider \emph{both} branches,
and most analyses will produce an upper bound on the program's privacy
cost by combining the two (i.e. taking the maximum of the two {{\color{\colorMATH}\ensuremath{\mathit{\epsilon }}}}
values). This approach avoids the issue of adaptively-chosen privacy
parameters.

A dynamic analysis, by contrast, \emph{cannot} consider both branches,
and must bound privacy cost by analyzing \emph{only} the branch that
is executed. Sequential composition does not apply directly when
privacy parameters are chosen adaptively, so ignoring the non-executed
branch in a dynamic analysis of privacy would be \emph{unsound}.

\subsection{Privacy Filters \& Odometers}

Privacy \emph{filters} and \emph{odometers} were originally developed
by Rogers et al.~\cite{RogersVRU16} specifically to address the
setting in which privacy parameters are selected
adaptively. Winograd-Cort et al.~\cite{Winograd-CortHR17} used privacy
filters and odometers as part of the Adaptive Fuzz framework, which
integrates both dynamic analysis (for composing privacy mechanisms)
and static analysis (for bounding the cost of individual
mechanisms). Recently, Feldman and Zrnic~\cite{feldman2020individual}
developed filters and odometers for R\'enyi differential
privacy~\cite{Mironov17}.

\emph{Privacy odometers} can be used to obtain a running upper bound
on total privacy cost at any point in the sequence of adaptive
mechanisms, and to obtain an overall total at the end of the
sequence. A function {{\color{\colorMATH}\ensuremath{\mathit{{{\color{\colorSYNTAX}\texttt{COMP}}}_{\delta _{g}}\mathrel{:} {\mathbb{R}}_{\geq 0}^{2k} \rightarrow  {\mathbb{R}} \cup  \{ \infty \} }}}} is called a
\emph{valid privacy odometer}~\cite{RogersVRU16} for a sequence of
mechanisms {{\color{\colorMATH}\ensuremath{\mathit{{\mathcal{M}}_{1}, \ldots , {\mathcal{M}}_{k}}}}} if for all (adaptively-chosen) settings of
{{\color{\colorMATH}\ensuremath{\mathit{(\epsilon _{1}, \delta _{1}), \ldots , (\epsilon _{k}, \delta _{k})}}}} for the individual mechanisms in the sequence,
their composition satisfies {{\color{\colorMATH}\ensuremath{\mathit{({{\color{\colorSYNTAX}\texttt{COMP}}}_{\delta _{g}}(\mathord{\cdotp }), \delta _{g})}}}}-differential
privacy. In other words, {{\color{\colorMATH}\ensuremath{\mathit{{{\color{\colorSYNTAX}\texttt{COMP}}}_{\delta _{g}}(\mathord{\cdotp })}}}} returns a value for {{\color{\colorMATH}\ensuremath{\mathit{\epsilon }}}}
that upper-bounds the privacy cost of the adaptive sequence of
mechanisms.
A valid privacy odometer for sequential composition in
{{\color{\colorMATH}\ensuremath{\mathit{(\epsilon , \delta )}}}}-differential privacy can be defined as follows (Rogers et
al.~\cite{RogersVRU16}, Theorem 3.6):
\begingroup\color{\colorMATH}\begin{gather*} {{\color{\colorSYNTAX}\texttt{COMP}}}_{\delta _{g}}(\epsilon _{1}, \delta _{1}, \ldots , \epsilon _{k}, \delta _{k}) =
    \begin{cases}
      \infty  & {\textit{if}}\hspace*{0.33em} \sum \limits_{i=1}^{k} \delta _{i} > \delta _{g} \\
      \sum \limits_{i=1}^{k} \epsilon _{i} & {\textit{otherwise}} \\
    \end{cases}
\end{gather*}\endgroup

\emph{Privacy filters} allow the analyst to place an upper bound
{{\color{\colorMATH}\ensuremath{\mathit{(\epsilon _{g}, \delta _{g})}}}} on the desired privacy cost, and halt the computation
immediately if the bound is violated. A function {{\color{\colorMATH}\ensuremath{\mathit{{{\color{\colorSYNTAX}\texttt{COMP}}}_{\epsilon _{g}, \delta _{g}}\mathrel{:}
{\mathbb{R}}_{\geq 0}^{2k} \rightarrow  \{ {{\color{\colorSYNTAX}\texttt{HALT}}}, {{\color{\colorSYNTAX}\texttt{CONT}}}\} }}}} is called a \emph{valid privacy
  filter}~\cite{RogersVRU16} for a sequence of mechanisms {{\color{\colorMATH}\ensuremath{\mathit{{\mathcal{M}}_{1}, \ldots , {\mathcal{M}}_{k}}}}}
if for all (adaptively-chosen) settings of {{\color{\colorMATH}\ensuremath{\mathit{(\epsilon _{1}, \delta _{1}), \ldots , (\epsilon _{k}, \delta _{k})}}}} for
the individual mechanisms in the sequence, {{\color{\colorMATH}\ensuremath{\mathit{{{\color{\colorSYNTAX}\texttt{COMP}}}_{\epsilon _{g}, \delta _{g}}(\epsilon _{1}, \delta _{1},
\ldots , \epsilon _{k}, \delta _{k})}}}} outputs {{\color{\colorMATH}\ensuremath{\mathit{{{\color{\colorSYNTAX}\texttt{CONT}}}}}}} only if the sequence satisfies {{\color{\colorMATH}\ensuremath{\mathit{(\epsilon _{g},
\delta _{g})}}}}-differential privacy (otherwise, it outputs {{\color{\colorMATH}\ensuremath{\mathit{{{\color{\colorSYNTAX}\texttt{HALT}}}}}}} for the
first mechanism in the sequence that violates the privacy cost bound).
A valid privacy filter for sequential composition in {{\color{\colorMATH}\ensuremath{\mathit{(\epsilon ,
\delta )}}}}-differential privacy can be defined as follows (Rogers et
al.~\cite{RogersVRU16}, Theorem 3.6):
\begingroup\color{\colorMATH}\begin{gather*} {{\color{\colorSYNTAX}\texttt{COMP}}}_{\epsilon _{g}, \delta _{g}}(\epsilon _{1}, \delta _{1}, \ldots , \epsilon _{k}, \delta _{k}) =  \hspace*{1.00em} \hspace*{1.00em} \hspace*{1.00em}\hspace*{1.00em} \hspace*{1.00em} \hspace*{1.00em}\hspace*{1.00em} \hspace*{1.00em} \hspace*{1.00em}\\
    \begin{cases}
      {{\color{\colorSYNTAX}\texttt{HALT}}} & {\textit{if}}\hspace*{0.33em} \sum \limits_{i=1}^{k} \delta _{i} > \delta _{g}\hspace*{0.33em}{\textit{or}}\hspace*{0.33em}\sum \limits_{i=1}^{k} \epsilon _{i} > \epsilon _{g} \\
      {{\color{\colorSYNTAX}\texttt{CONT}}} & {\textit{otherwise}} \\
    \end{cases}
\end{gather*}\endgroup

It is clear from these definitions that the odometer and filter for
sequential composition under {{\color{\colorMATH}\ensuremath{\mathit{(\epsilon , \delta )}}}}-differential privacy yield the
same bounds on privacy loss as the standard theorem for sequential
composition~\cite{dwork2014algorithmic} (i.e. there is no ``cost'' to
picking the privacy parameters adaptively).

Rogers et al.~\cite{RogersVRU16} also define filters and odometers for
\emph{advanced composition} under {{\color{\colorMATH}\ensuremath{\mathit{(\epsilon , \delta )}}}}-differential privacy
(\cite{RogersVRU16}, \S 5 and \S 6); in this case, there \emph{is} a
cost. In exchange for the ability to set privacy parameters
adaptively, filters and odometers for advanced composition have
slightly worse constants than the standard advanced composition
theorem~\cite{dwork2014algorithmic} (but are asymptotically the same).


\subsection{Filters \& Odometers in \dduo}

\dduo's API allows the programmer to explicitly create privacy
odometers and filters, and make them active for a specific part of the
program (using Python's \emph{with} syntax). When an odometer is
active, it records a running total of the total privacy cost, and it
can be queried to return this information to the programmer.

\vspace{1em}
\begin{minted}{python}
with dduo.EdOdometer(delta = 10e-6) as odo:
  _ = dduo.gauss(df.shape[0], |$\epsilon $| = 1.0, |$\delta $| = 10e-6)
  _ = dduo.gauss(df.shape[0], |$\epsilon $| = 1.0, |$\delta $| = 10e-6)
  print(odo)
\end{minted}

\vspace{-1em}
\begin{minted}[frame=lines,bgcolor=mygray]{text}
Odometer_|$(\epsilon ,\delta )$|(|$\{ \textit{data.csv} \mapsto  (2.0, 20^{-6})\} $|)
\end{minted}

\noindent When a filter is active, it tracks the privacy cost for
individual mechanisms, and halts the program if the filter's upper
bound on privacy cost is violated.

\begin{minted}{python}
with dduo.EdFilter(|$\epsilon $| = 1.0, |$\delta $| = 10e-6) as odo:
  print('1:', dduo.gauss(df.shape[0], |$\epsilon $|=1.0, |$\delta $|=10e-6))
  print('2:', dduo.gauss(df.shape[0], |$\epsilon $|=1.0, |$\delta $|=10e-6))
\end{minted}

\vspace{-1em}
\begin{minted}[frame=lines,bgcolor=mygray]{text}
1: 10.5627
Traceback (most recent call last):
  ...
  dduo.PrivacyFilterException
\end{minted}

In addition to odometers and filters for sequential composition under
{{\color{\colorMATH}\ensuremath{\mathit{(\epsilon , \delta )}}}}-differential privacy (such as \texttt{EdFilter} and
\texttt{EdOdometer}), \dduo provides odometers and filters for
advanced composition (\texttt{AdvEdFilter} and
\texttt{AdvEdOdometer}) and R\'enyi differential privacy
(\texttt{RenyiFilter} and \texttt{RenyiOdometer}, which follow the results of Feldman and Zrnic~\cite{feldman2020individual}).

\subsection{Loops and Composition}

Iterative algorithms can be built in \dduo using Python's standard
looping constructs, and \dduo's privacy odometers and filters take
care of ensuring the correct form of composition. Parallel composition
is also available---via functional mapping. Advanced composition can be
achieved via special advanced composition filters and odometers
exposed in the \dduo API. For example, the following simple loop runs
the Laplace mechanism 20 times, and its total privacy cost is
reflected by the odometer:

\vspace{1em}
\begin{minted}{python}
with dduo.EpsOdometer() as odo:
  for i in range(20):
    dduo.laplace(df.shape[0], |$\epsilon $| = 1.0)
  print(odo)
\end{minted}

\vspace{-1em}
\begin{minted}[frame=lines,bgcolor=mygray]{text}
Odometer_|$\epsilon $|(|$\{ \textit{data.csv} \mapsto  20.0\} $|)
\end{minted}

\noindent To use advanced composition instead of sequential
composition, we simply replace the odometer with a different one:

\begin{minted}{python}
with dduo.AdvEdOdometer() as odo:
  for i in range(20):
    dduo.gauss(df.shape[0], |$\epsilon $| = 0.01, |$\delta $| = 0.001)
\end{minted}



\subsection{Mixing Variants of Differential Privacy}

The \dduo library includes support for pure {{\color{\colorMATH}\ensuremath{\mathit{\epsilon }}}}-differential privacy, {{\color{\colorMATH}\ensuremath{\mathit{(\epsilon , \delta )}}}}-differential privacy, and R\'enyi differential privacy (RDP). Programs may use all three variants together, convert between them, and compose mechanisms from each.

We demonstrate execution of a query while switched to the R\'enyi differential privacy variant using pythonic "with" syntax blocks. For programs that make extensive use of composition, this approach yields significant improvements in privacy cost. For example, the following program uses the Gaussian mechanism 200 times, using R\'enyi differential privacy for sequential composition; the total privacy cost is automatically converted into an {{\color{\colorMATH}\ensuremath{\mathit{(\epsilon , \delta )}}}}-differential privacy cost after the loop finishes.

\vspace{1em}
\begin{minted}[escapeinside=||,mathescape=true]{python}
with dduo.RenyiDP(1e-5):
  for x in range(200):
    noisy_count = dduo.renyi_gauss(|$\alpha $| = 10,
      |$\epsilon $|=0.2, df.shape[0])
dduo.print_privacy_cost()
\end{minted}

\vspace{-1em}
\begin{minted}[frame=lines,bgcolor=mygray]{text}
Odometer_|$(\epsilon ,\delta )$|(|$\{ \textit{data.csv} \mapsto  (41.28, 1e-05)\} $|)
\end{minted}

\section{Formal Description of Sensitivity Analysis}
\label{sec:formalism}

In \dduo we implement a novel dynamic analysis for {\textit{function
sensitivity}}, which is a relational (hyper)property quantified over
two runs of the program with arbitrary but related inputs.
In particular, our analysis computes function sensitivity---a two-run
property---after only observing {\textit{one}} execution of the program.
Only observing one execution poses challenges to the design of the
analysis, and significant challenges to the proof, all of which we
overcome.
To overcome this challenge in the design of the analysis, we first
disallow branching control flow which depends on any sensitive
inputs; this ensures that any two runs of the program being
considered for the purposes of privacy will take the same branch
observed by the dynamic analysis.
Second, we disallow sensitive input-dependent arguments to the
``scalar'' side of multiplication; this ensures that the dynamic
analysis' use of that argument in analysis results is identical for
any two runs of the program being considered for the purposes of
privacy.
Our dynamic analysis for function sensitivity is {\textit{sound}}---meaning that
the true sensitivity of a program is guaranteed to be equal or less
than the sensitivity reported by \dduo's dynamic monitor---and we
support this claim with a detailed proof.


\paragraph{Formalism Approach.}
We formalize the correctness of our dynamic analysis for function
sensitivity using a {\textit{step-indexed big-step semantics}} to describe the
dynamic analysis, a {\textit{step-indexed logical relation}} to describe the
meaning of function sensitivity, and a proof by induction and case
analysis on program syntax to show that dynamic analysis results
soundly predict function sensitivity.
A {\textit{step-indexed}} relation is a relation {{\color{\colorMATH}\ensuremath{\mathit{{\mathcal{R}} \in  A \rightarrow  B \rightarrow  {\text{prop}}}}}} whose
definition is stratified by a natural number index {{\color{\colorMATH}\ensuremath{\mathit{n}}}}, so for each
level {{\color{\colorMATH}\ensuremath{\mathit{n}}}} there is a new relation {{\color{\colorMATH}\ensuremath{\mathit{{\mathcal{R}}_{n}}}}}.
Typically, the relation {{\color{\colorMATH}\ensuremath{\mathit{{\mathcal{R}}_{0}}}}} is defined {{\color{\colorMATH}\ensuremath{\mathit{{\mathcal{R}}_{0}(x,y) \triangleq  {\text{true}}}}}}, and the
final relation of interest is {{\color{\colorMATH}\ensuremath{\mathit{\hat {\mathcal{R}} \triangleq  \bigcap _{n}{\mathcal{R}}_{n}}}}}, i.e., {{\color{\colorMATH}\ensuremath{\mathit{\hat {\mathcal{R}}(x,y) \iff    \forall n.\hspace*{0.33em}
{\mathcal{R}}_{n}(x,y)}}}}.
Step-indexing is typically used---as we do in our formalism---when the
definition of a relation would be not well founded in its absence.
The most common reason a relation definition might be not
well-founded is the use of self-reference without any decreasing
measure.
When a decreasing measure exists, self-reference leads to
well-founded recursion, however when a decreasing measure does not
exist, self-reference is not well-founded.
When using step-indexing, self-reference is allowed in the definition
of {{\color{\colorMATH}\ensuremath{\mathit{{\mathcal{R}}_{n}}}}}, but only for the relation at strictly lower levels, so
{{\color{\colorMATH}\ensuremath{\mathit{{\mathcal{R}}_{n^{\prime}}}}}} when {{\color{\colorMATH}\ensuremath{\mathit{n^{\prime}<n}}}}; this is well-founded because the index {{\color{\colorMATH}\ensuremath{\mathit{n}}}}
becomes a decreasing measure for the self-reference.
In this way, step-indexing enables self-reference without any
existing decreasing measure by introducing a new decreasing measure,
and maintains well-foundedness of the relation definition.

A {\textit{logical}} relation is one where the definition of relation on
function values (or types) is extensional, essentially saying ``when
given related inputs, the function produces related outputs''.
This definition is self-referrential and not well-founded, and among
common reasons to introduce step-indexing in programming language
proofs.
As the relation {{\color{\colorMATH}\ensuremath{\mathit{{\mathcal{R}}}}}} is stratified with a step-index to {{\color{\colorMATH}\ensuremath{\mathit{{\mathcal{R}}_{n}}}}}, so must
the definition of the semantics, so for a big-step relation {{\color{\colorMATH}\ensuremath{\mathit{e \Downarrow  v}}}}
(relating an expression {{\color{\colorMATH}\ensuremath{\mathit{e}}}} to its final value {{\color{\colorMATH}\ensuremath{\mathit{v}}}} after evaluation)
we stratify as {{\color{\colorMATH}\ensuremath{\mathit{e \Downarrow _{n} v}}}}.
Also, because the definition of a logical relation {\textit{decrements}} the
step-index for the case of function values, we {\textit{increment}} the
step-index in the semantic case for function application.
These techniques are standard from prior work~\cite{ahmed2006},
and we merely summarize the key ideas here to give background to our
reader.

\paragraph{Formal Definition of Dynamic Analysis.}
We model language features for arithmetic operations ({{\color{\colorMATH}\ensuremath{\mathit{e \odot  e}}}}),
conditionals ({{\color{\colorMATH}\ensuremath{\mathit{{{\color{\colorSYNTAX}\texttt{if0}}}(e)\{ e\} \{ e\} }}}}), pairs ({{\color{\colorMATH}\ensuremath{\mathit{\langle e,e\rangle }}}} and {{\color{\colorMATH}\ensuremath{\mathit{\pi _{i}(e)}}}}),
functions ({{\color{\colorMATH}\ensuremath{\mathit{\lambda x.\hspace*{0.33em}e}}}} and {{\color{\colorMATH}\ensuremath{\mathit{e(e)}}}}) and references ({{\color{\colorMATH}\ensuremath{\mathit{{{\color{\colorSYNTAX}\texttt{ref}}}(e)}}}}, {{\color{\colorMATH}\ensuremath{\mathit{{!}e}}}} and
{{\color{\colorMATH}\ensuremath{\mathit{e\leftarrow e}}}}); the full language is shown in Figure~\ref{fig:syntax}.
There is one base value: {{\color{\colorMATH}\ensuremath{\mathit{r@_{m}^{\Sigma }}}}} for a real number result {{\color{\colorMATH}\ensuremath{\mathit{r}}}} tagged
with dynamic analysis information {{\color{\colorMATH}\ensuremath{\mathit{\Sigma }}}}---the sensitivity analysis for
the expression which evaluated to {{\color{\colorMATH}\ensuremath{\mathit{r}}}}---and {{\color{\colorMATH}\ensuremath{\mathit{m}}}}---the metric associated
with the resulting value {{\color{\colorMATH}\ensuremath{\mathit{r}}}}. The sensitivity analysis {{\color{\colorMATH}\ensuremath{\mathit{\Sigma }}}}---also
called a {\textit{sensitivity environment}}---is a map from sensitive sources {{\color{\colorMATH}\ensuremath{\mathit{o
\in  {\text{source}}}}}} to how sensitive the result is {w.r.t.} that source. Our
formalism includes two base metrics {{\color{\colorMATH}\ensuremath{\mathit{m \in  {\text{metric}}}}}}: {{\color{\colorMATH}\ensuremath{\mathit{{{\color{\colorSYNTAX}\texttt{diff}}}}}}} and
{{\color{\colorMATH}\ensuremath{\mathit{{{\color{\colorSYNTAX}\texttt{disc}}}}}}} for absolute difference ({{\color{\colorMATH}\ensuremath{\mathit{|x - y|}}}}) and discrete distance
({{\color{\colorMATH}\ensuremath{\mathit{0}}}} if {{\color{\colorMATH}\ensuremath{\mathit{x = y}}}} and {{\color{\colorMATH}\ensuremath{\mathit{1}}}} otherwise) respectively---and two derived
metrics: {{\color{\colorMATH}\ensuremath{\mathit{\top }}}} and {{\color{\colorMATH}\ensuremath{\mathit{\bot }}}} for the smallest metric larger than each base
metric and largest metric smaller than each base metric,
respectively. Each metric is commonly used when implementing
differentially private algorithms.
Pair values ({{\color{\colorMATH}\ensuremath{\mathit{\langle v,v\rangle }}}}), closure values ({{\color{\colorMATH}\ensuremath{\mathit{\langle \lambda x.\hspace*{0.33em}e\mathrel{|}\rho \rangle }}}}) and reference
values ({{\color{\colorMATH}\ensuremath{\mathit{\ell }}}}) do not contain dynamic analysis information.

Our dynamic analysis is described formally as a big-step relation {{\color{\colorMATH}\ensuremath{\mathit{\rho
\vdash  \mathrlap{\raisebox{-5pt}{\color{black!50}$\llcorner$}}\mathrlap{\raisebox{6pt}{\color{black!50}$\ulcorner$}}\hspace{2pt}\vphantom{\underset x{\overset XX}}\sigma ,e\hspace{2pt}\mathllap{\raisebox{-5pt}{\color{black!50}$\lrcorner$}}\mathllap{\raisebox{6pt}{\color{black!50}$\urcorner$}}\vphantom{\underset x{\overset XX}} \Downarrow _{n} \mathrlap{\raisebox{-5pt}{\color{black!50}$\llcorner$}}\mathrlap{\raisebox{6pt}{\color{black!50}$\ulcorner$}}\hspace{2pt}\vphantom{\underset x{\overset XX}}\sigma ,v\hspace{2pt}\mathllap{\raisebox{-5pt}{\color{black!50}$\lrcorner$}}\mathllap{\raisebox{6pt}{\color{black!50}$\urcorner$}}\vphantom{\underset x{\overset XX}}}}}} where {{\color{\colorMATH}\ensuremath{\mathit{\rho  \in  {\text{var}} \rightharpoonup  {\text{value}}}}}} is the lexical
environment mapping lexical variables to values, {{\color{\colorMATH}\ensuremath{\mathit{\sigma  \in  {\text{loc}} \rightharpoonup
{\text{value}}}}}} is the dynamic environment (i.e., the heap, or store)
mapping dynamically allocated references to values, {{\color{\colorMATH}\ensuremath{\mathit{e}}}} is the
expression being executed, and {{\color{\colorMATH}\ensuremath{\mathit{v}}}} is the resulting runtime value
which also includes dynamic analysis information. We write gray box
corners around the ``input'' configuration {{\color{\colorMATH}\ensuremath{\mathit{\mathrlap{\raisebox{-5pt}{\color{black!50}$\llcorner$}}\mathrlap{\raisebox{6pt}{\color{black!50}$\ulcorner$}}\hspace{2pt}\vphantom{\underset x{\overset XX}}\sigma ,e\hspace{2pt}\mathllap{\raisebox{-5pt}{\color{black!50}$\lrcorner$}}\mathllap{\raisebox{6pt}{\color{black!50}$\urcorner$}}\vphantom{\underset x{\overset XX}}}}}} and the ``output''
configuration {{\color{\colorMATH}\ensuremath{\mathit{\mathrlap{\raisebox{-5pt}{\color{black!50}$\llcorner$}}\mathrlap{\raisebox{6pt}{\color{black!50}$\ulcorner$}}\hspace{2pt}\vphantom{\underset x{\overset XX}}\sigma ,v\hspace{2pt}\mathllap{\raisebox{-5pt}{\color{black!50}$\lrcorner$}}\mathllap{\raisebox{6pt}{\color{black!50}$\urcorner$}}\vphantom{\underset x{\overset XX}}}}}} to aid readability. The index {{\color{\colorMATH}\ensuremath{\mathit{n}}}} is for
step-indexing, and tracks the number of function applications which
occurred in the process of evaluation. We show the full definition of
the dynamic analysis in Figure~\ref{fig:semantics}.

Consider the following example:
\begingroup\color{\colorMATH}\begin{gather*} \{ x\mapsto 21@_{{{\color{\colorSYNTAX}\texttt{diff}}}}^{\{ o\mapsto 1\} }\}  \vdash  \varnothing ,(x + x) \Downarrow _{0} 42@_{{{\color{\colorSYNTAX}\texttt{diff}}}}^{\{ o\mapsto 2\} } \end{gather*}\endgroup
This relation corresponds to a scenario where the program to evaluate
and analyze is {{\color{\colorMATH}\ensuremath{\mathit{x + x}}}}, the variable {{\color{\colorMATH}\ensuremath{\mathit{x}}}} represents a sensitive
source value {{\color{\colorMATH}\ensuremath{\mathit{o}}}}, we want to track sensitivity {w.r.t.} the absolute
difference metric, and the initial value for {{\color{\colorMATH}\ensuremath{\mathit{x}}}} is {{\color{\colorMATH}\ensuremath{\mathit{21}}}}. This
information is encoded in an initial environment {{\color{\colorMATH}\ensuremath{\mathit{\rho  =
\{ x\mapsto 21@_{{{\color{\colorSYNTAX}\texttt{diff}}}}^{\{ o\mapsto 1\} }}}}}. The result value is {{\color{\colorMATH}\ensuremath{\mathit{42}}}}, and the resulting
analysis reports that {{\color{\colorMATH}\ensuremath{\mathit{e}}}} is {{\color{\colorMATH}\ensuremath{\mathit{2}}}}-sensitive in the source {{\color{\colorMATH}\ensuremath{\mathit{o}}}} {w.r.t.}
the absolute difference metric. This analysis information is encoded
in the return value {{\color{\colorMATH}\ensuremath{\mathit{42@_{{{\color{\colorSYNTAX}\texttt{diff}}}}^{\{ o\mapsto 2\} }}}}}. Because no function
applications occur during evaluation, the step index {{\color{\colorMATH}\ensuremath{\mathit{n}}}} is 0.

\begin{figure*}
\begin{framed}
\vspace{-2ex}
\begingroup\color{\colorMATH}\begin{gather*}
\begin{tabularx}{\linewidth}{>{\centering\arraybackslash\(}X<{\)}}\hfill\hspace{0pt} b \in  {\mathbb{B}} \hfill\hspace{0pt} n \in  {\mathbb{N}} \hfill\hspace{0pt} i \in  {\mathbb{Z}} \hfill\hspace{0pt} r \in  {\mathbb{R}} \hfill\hspace{0pt} x \in  {\text{var}} \hfill\hspace{0pt}
\cr
\cr \hfill\hspace{0pt}
  \begin{array}{rclcl@{\hspace*{1.00em}}l
  } o   &{}\in {}& {\text{source}}    &{} {}&                  & {{\color{\colorTEXT}\textnormal{sensitive sources}}}
  \cr  \ell    &{}\in {}& {\text{loc}}       &{} {}&                  & {{\color{\colorTEXT}\textnormal{reference locations}}}
  \cr  e   &{}\in {}& {\text{expr}}      &{}\mathrel{\Coloneqq }{}& x                & {{\color{\colorTEXT}\textnormal{variables}}}
  \cr      &{} {}&             &{}\mathrel{|}{}& r                & {{\color{\colorTEXT}\textnormal{real numbers}}}
  \cr      &{} {}&             &{}\mathrel{|}{}& e \odot  e            & {{\color{\colorTEXT}\textnormal{arith. operations}}}
  \cr      &{} {}&             &{}\mathrel{|}{}& {{\color{\colorSYNTAX}\texttt{if0}}}(e)\{ e\} \{ e\}    & {{\color{\colorTEXT}\textnormal{cond. branching}}}
  \cr      &{} {}&             &{}\mathrel{|}{}& \langle e,e\rangle             & {{\color{\colorTEXT}\textnormal{pair creation}}}
  \cr      &{} {}&             &{}\mathrel{|}{}& \pi _{i}(e)            & {{\color{\colorTEXT}\textnormal{pair access}}}
  \cr      &{} {}&             &{}\mathrel{|}{}& \lambda x.\hspace*{0.33em}e             & {{\color{\colorTEXT}\textnormal{function creation}}}
  \cr      &{} {}&             &{}\mathrel{|}{}& e(e)             & {{\color{\colorTEXT}\textnormal{function application}}}
  \cr      &{} {}&             &{}\mathrel{|}{}& {{\color{\colorSYNTAX}\texttt{ref}}}(e)         & {{\color{\colorTEXT}\textnormal{reference creation}}}
  \cr      &{} {}&             &{}\mathrel{|}{}& {!}e               & {{\color{\colorTEXT}\textnormal{reference read}}}
  \cr      &{} {}&             &{}\mathrel{|}{}& e \leftarrow  e            & {{\color{\colorTEXT}\textnormal{reference write}}}
  \end{array}
  \hfill\hspace{0pt}
  \begin{array}{rclcl@{\hspace*{1.00em}}l
  } q   &{}\in {}& \hat {\mathbb{R}}         &{}\mathrel{\Coloneqq }{}& r \mathrel{|} \infty             & {{\color{\colorTEXT}\textnormal{ext. reals}}}
  \cr  {\odot } &{}\in {}& {\text{binop}}     &{}\mathrel{\Coloneqq }{}& {+} \mathrel{|} {\ltimes } \mathrel{|} {\rtimes }  & {{\color{\colorTEXT}\textnormal{operations}}}
  \cr  m   &{}\in {}& {\text{metric}}    &{}\mathrel{\Coloneqq }{}& {{\color{\colorSYNTAX}\texttt{diff}}}           & {{\color{\colorTEXT}\textnormal{absolute difference}}}
  \cr      &{} {}&             &{}\mathrel{|}{}& {{\color{\colorSYNTAX}\texttt{disc}}}           & {{\color{\colorTEXT}\textnormal{discrete}}}
  \cr      &{} {}&             &{}\mathrel{|}{}& \bot                 & {{\color{\colorTEXT}\textnormal{bot metric}}}
  \cr      &{} {}&             &{}\mathrel{|}{}& \top                 & {{\color{\colorTEXT}\textnormal{top metric}}}
  \cr  v   &{}\in {}& {\text{value}}     &{}\mathrel{\Coloneqq }{}& r@_{m}^{\Sigma }           & {{\color{\colorTEXT}\textnormal{tagged base value}}}
  \cr      &{} {}&             &{}\mathrel{|}{}& \langle v,v\rangle             & {{\color{\colorTEXT}\textnormal{pair}}}
  \cr      &{} {}&             &{}\mathrel{|}{}& \langle \lambda x.\hspace*{0.33em}e\mathrel{|}\rho \rangle          & {{\color{\colorTEXT}\textnormal{function (closure)}}}
  \cr      &{} {}&             &{}\mathrel{|}{}& \ell                 & {{\color{\colorTEXT}\textnormal{location (pointer)}}}
  \cr  \rho    &{}\in {}& {\text{env}}   &{}\triangleq {}& {\text{var}} \rightharpoonup  {\text{value}}      & {{\color{\colorTEXT}\textnormal{value environment}}}
  \cr  \sigma    &{}\in {}& {\text{store}} &{}\triangleq {}& {\text{loc}} \rightharpoonup  {\text{value}}      & {{\color{\colorTEXT}\textnormal{mutable store}}}
  \cr  \Sigma    &{}\in {}& {\text{senv}}  &{}\triangleq {}& {\text{source}} \rightharpoonup  \hat {\mathbb{R}}       & {{\color{\colorTEXT}\textnormal{sens. environment}}}
  \end{array}
  \hfill\hspace{0pt}
\cr
\cr \hline
\cr \hfill\hspace{0pt}
  \begin{array}[t]{rclrl
  } \rho _{1},\sigma _{1},e_{1} &{}\sim _{0}^{\Sigma }    {}& \rho _{2},\sigma _{2},e_{2} &{}\overset \vartriangle \iff {}& {\text{true}}
  \cr  \rho _{1},\sigma _{1},e_{1} &{}\sim _{n+1}^{\Sigma }{}& \rho _{2},\sigma _{2},e_{2} &{}\overset \vartriangle \iff {}& \forall  n_{1} \leq  n, n_{2}, \sigma _{1}^{\prime},\sigma _{2}^{\prime},v_{1}, v_{2}.\hspace*{0.33em}
  \cr           &{}         {}&          &{}   {}& \rho _{1} \vdash  \mathrlap{\raisebox{-5pt}{\color{black!50}$\llcorner$}}\mathrlap{\raisebox{6pt}{\color{black!50}$\ulcorner$}}\hspace{2pt}\vphantom{\underset x{\overset XX}}\sigma _{1},e_{1}\hspace{2pt}\mathllap{\raisebox{-5pt}{\color{black!50}$\lrcorner$}}\mathllap{\raisebox{6pt}{\color{black!50}$\urcorner$}}\vphantom{\underset x{\overset XX}} \Downarrow _{n_{1}} \mathrlap{\raisebox{-5pt}{\color{black!50}$\llcorner$}}\mathrlap{\raisebox{6pt}{\color{black!50}$\ulcorner$}}\hspace{2pt}\vphantom{\underset x{\overset XX}}\sigma _{1}^{\prime},v_{1}\hspace{2pt}\mathllap{\raisebox{-5pt}{\color{black!50}$\lrcorner$}}\mathllap{\raisebox{6pt}{\color{black!50}$\urcorner$}}\vphantom{\underset x{\overset XX}} \hspace*{0.33em}\wedge \hspace*{0.33em} \rho _{2} \vdash  \mathrlap{\raisebox{-5pt}{\color{black!50}$\llcorner$}}\mathrlap{\raisebox{6pt}{\color{black!50}$\ulcorner$}}\hspace{2pt}\vphantom{\underset x{\overset XX}}\sigma _{2},e_{2}\hspace{2pt}\mathllap{\raisebox{-5pt}{\color{black!50}$\lrcorner$}}\mathllap{\raisebox{6pt}{\color{black!50}$\urcorner$}}\vphantom{\underset x{\overset XX}} \Downarrow _{n_{2}} \mathrlap{\raisebox{-5pt}{\color{black!50}$\llcorner$}}\mathrlap{\raisebox{6pt}{\color{black!50}$\ulcorner$}}\hspace{2pt}\vphantom{\underset x{\overset XX}}\sigma _{2}^{\prime},v_{2}\hspace{2pt}\mathllap{\raisebox{-5pt}{\color{black!50}$\lrcorner$}}\mathllap{\raisebox{6pt}{\color{black!50}$\urcorner$}}\vphantom{\underset x{\overset XX}}
  \cr           &{}         {}&          &{}  \Rightarrow {}& n_{1} = n_{2} \hspace*{0.33em}\wedge \hspace*{0.33em} \sigma _{1}^{\prime} \sim ^{\Sigma }_{n-n_{1}} \sigma _{2}^{\prime} \hspace*{0.33em}\wedge \hspace*{0.33em}  v_{1} \sim ^{\Sigma }_{n-n_{1}} v_{2}
  \end{array}
  \hfill\hspace{0pt}
  \hfill\hspace{0pt}
  \mathllap{\begingroup\color{\colorTEXT}\boxed{\begingroup\color{\colorMATH} \rho ,\sigma ,e \sim _{n}^{\Sigma } \rho ,\sigma ,e \endgroup}\endgroup}
\cr
\cr \hfill\hspace{0pt}
  \begin{array}[t]{rcl
  } r_{1} \sim _{{{\color{\colorSYNTAX}\texttt{diff}}}}^{r} r_{2}            &{}\overset \vartriangle \iff {}& |r_{1} - r_{2}| \leq  r
  \cr  r_{1} \sim _{{{\color{\colorSYNTAX}\texttt{disc}}}}^{r} r_{2}            &{}\overset \vartriangle \iff {}& \left\{ \begin{array}{l@{\hspace*{1.00em}}c@{\hspace*{1.00em}}l
                                         } 0 \leq  r &{}{\textit{if}}{}& r_{1} = r_{2}
                                         \cr  1 \leq  r &{}{\textit{if}}{}& r_{1} \neq  r_{2}
                                         \end{array}\right.
  \end{array}
  \hfill\hspace{0pt}
  \begin{array}[t]{rcl
  } r_{1} \sim _{\bot }^{r} r_{2}                 &{}\overset \vartriangle \iff {}& r_{1} \sim _{{{\color{\colorSYNTAX}\texttt{diff}}}}^{r} r_{2} \hspace*{0.33em}\wedge \hspace*{0.33em} r_{1} \sim _{{{\color{\colorSYNTAX}\texttt{disc}}}}^{r} r_{2}
  \cr  r_{1} \sim _{\top }^{r} r_{2}                 &{}\overset \vartriangle \iff {}& r_{1} \sim _{{{\color{\colorSYNTAX}\texttt{diff}}}}^{r} r_{2} \hspace*{0.33em}\vee \hspace*{0.33em} r_{1} \sim _{{{\color{\colorSYNTAX}\texttt{disc}}}}^{r} r_{2}
  \end{array}
  \hfill\hspace{0pt}
  \hfill\hspace{0pt}
  \mathllap{\begingroup\color{\colorTEXT}\boxed{\begingroup\color{\colorMATH} r \sim _{m}^{r} r \endgroup}\endgroup}
\cr
\cr \hfill\hspace{0pt}
  \begin{array}[t]{rrl
  } r_{1}@^{\Sigma _{1}}_{m_{1}} \sim _{n}^{\Sigma } r_{2}@^{\Sigma _{2}}_{m_{2}} &{}\overset \vartriangle \iff {}& \Sigma _{1} = \Sigma _{2} \hspace*{0.33em}\wedge \hspace*{0.33em} m_{1} = m_{2} \hspace*{0.33em}\wedge \hspace*{0.33em} r_{1} \sim ^{\Sigma \mathord{\cdotp }\Sigma _{1}}_{m_{1}} r_{2}
  \cr  \langle v_{1 1},v_{1 2}\rangle  \sim _{n}^{\Sigma } \langle v_{2 1},v_{2 2}\rangle      &{}\overset \vartriangle \iff {}& v_{1 1} \sim _{n}^{\Sigma } v_{2 1} \hspace*{0.33em}\wedge \hspace*{0.33em} v_{2 1} \sim _{n}^{\Sigma } v_{2 2}
  \cr  \langle \lambda x.\hspace*{0.33em}e_{1}\mathrel{|}\rho _{1}\rangle  \sim _{n}^{\Sigma } \langle \lambda x.\hspace*{0.33em}e_{2}\mathrel{|}\rho _{2}\rangle    &{}\overset \vartriangle \iff {}& \forall  n^{\prime}\leq n,v_{1},v_{2},\sigma _{1},\sigma _{2}.\hspace*{0.33em} \sigma _{1} \sim _{n^{\prime}}^{\Sigma } \sigma _{2} \hspace*{0.33em}\wedge \hspace*{0.33em} v_{1} \sim _{n^{\prime}}^{\Sigma } v_{2}
  \cr                                &{}  \Rightarrow {}& \sigma _{1},\{ x\mapsto v_{1}\} \uplus \rho _{1},e_{1} \sim _{n^{\prime}}^{\Sigma } \sigma _{2},\{ x\mapsto v_{2}\} \uplus \rho _{2},e_{2}
  \cr  \ell _{1} \sim _{n}^{\Sigma } \ell _{2}                   &{}\overset \vartriangle \iff {}& \ell _{1} = \ell _{2}
  \end{array}
  \hfill\hspace{0pt}
  \mathllap{\begingroup\color{\colorTEXT}\boxed{\begingroup\color{\colorMATH} v \sim _{n}^{\Sigma } v \endgroup}\endgroup}
\cr
\cr \hfill\hspace{0pt}
  \begin{array}[t]{rcl
  } \rho _{1} \sim _{n}^{\Sigma } \rho _{2} &{}\overset \vartriangle \iff {}& \forall  x \in  ({\text{dom}}(\rho _{1}) \cup  {\text{dom}}(\rho _{2})).\hspace*{0.33em} \rho _{1}(x) \sim _{n}^{\Sigma } \rho _{2}(x)
  \cr  \sigma _{1} \sim _{n}^{\Sigma } \sigma _{2} &{}\overset \vartriangle \iff {}& \forall  \ell  \in  ({\text{dom}}(\sigma _{1}) \cup  {\text{dom}}(\sigma _{2})).\hspace*{0.33em} \sigma _{1}(\ell ) \sim _{n}^{\Sigma } \sigma _{2}(\ell )
  \end{array}
  \hfill\hspace{0pt}
  \mathllap{\begingroup\color{\colorTEXT}\boxed{\begingroup\color{\colorMATH}\begin{array}{r} \rho  \sim _{n}^{\Sigma } \rho  \cr  \sigma  \sim _{n}^{\Sigma } \sigma  \end{array}\endgroup}\endgroup}
\end{tabularx}
\end{gather*}\endgroup
\vspace{-2ex}
\end{framed}
\vspace{-2ex}
\caption{Formal Syntax \& Step-indexed Logical Relation.}
\label{fig:syntax}\label{fig:logical-relations}
\end{figure*}

\paragraph{Formal Definition of Function Sensitivity.}
Function sensitivity is encoded through multiple relation
definitions:
\begin{enumerate}\item  {{\color{\colorMATH}\ensuremath{\mathit{\rho _{1},\sigma _{1},e_{1} \sim _{n}^{\Sigma } \rho _{2},\sigma _{2},e_{2}}}}} holds when the input triples {{\color{\colorMATH}\ensuremath{\mathit{\rho _{1},\sigma _{1},e_{1}}}}}
   and {{\color{\colorMATH}\ensuremath{\mathit{\rho _{2},\sigma _{2},e_{2}}}}} evaluate to output stores and values which are
   related by {{\color{\colorMATH}\ensuremath{\mathit{\Sigma }}}}. Note this definition decrements the step-index
   {{\color{\colorMATH}\ensuremath{\mathit{n}}}}, and is the constant relation when {{\color{\colorMATH}\ensuremath{\mathit{n=0}}}}.
\item  {{\color{\colorMATH}\ensuremath{\mathit{r_{1} \sim _{m}^{r} r_{2}}}}} holds when the difference between real numbers {{\color{\colorMATH}\ensuremath{\mathit{r_{1}}}}}
   and {{\color{\colorMATH}\ensuremath{\mathit{r_{2}}}}} {w.r.t.} metric {{\color{\colorMATH}\ensuremath{\mathit{m}}}} is less than {{\color{\colorMATH}\ensuremath{\mathit{r}}}}.
\item  {{\color{\colorMATH}\ensuremath{\mathit{v_{1} \sim _{n}^{\Sigma } v_{2}}}}} holds when values {{\color{\colorMATH}\ensuremath{\mathit{v_{1}}}}} and {{\color{\colorMATH}\ensuremath{\mathit{v_{2}}}}} are related for
   initial distance {{\color{\colorMATH}\ensuremath{\mathit{\Sigma }}}} and step-index {{\color{\colorMATH}\ensuremath{\mathit{n}}}}. The definition is by case
   analysis on the syntactic category for values, such as:
   \begin{enumerate}\item  The relation on base values {{\color{\colorMATH}\ensuremath{\mathit{r_{1}@^{\Sigma _{1}}_{m_{1}} \sim _{n}^{\Sigma } r_{2}@^{\Sigma _{2}}_{m_{2}}}}}}
      holds when {{\color{\colorMATH}\ensuremath{\mathit{\Sigma _{1}}}}}, {{\color{\colorMATH}\ensuremath{\mathit{\Sigma _{2}}}}}, {{\color{\colorMATH}\ensuremath{\mathit{m_{1}}}}} and {{\color{\colorMATH}\ensuremath{\mathit{m_{2}}}}} are pairwise equal, and
      when {{\color{\colorMATH}\ensuremath{\mathit{r_{1}}}}} and {{\color{\colorMATH}\ensuremath{\mathit{r_{2}}}}} are related by {{\color{\colorMATH}\ensuremath{\mathit{\Sigma \mathord{\cdotp }\Sigma _{1}}}}}, where {{\color{\colorMATH}\ensuremath{\mathit{\Sigma }}}} is the
      initial distances between each input source {{\color{\colorMATH}\ensuremath{\mathit{o}}}}, and {{\color{\colorMATH}\ensuremath{\mathit{\Sigma _{1}}}}} is
      how much {{\color{\colorMATH}\ensuremath{\mathit{r_{1}}}}} and {{\color{\colorMATH}\ensuremath{\mathit{r_{2}}}}} are allowed to differ as a linear
      function of input distances {{\color{\colorMATH}\ensuremath{\mathit{\Sigma }}}}, and where this function is
      applied via vector dot product {{\color{\colorMATH}\ensuremath{\mathit{\mathord{\cdotp }}}}}.
   \item  The relation on pair values {{\color{\colorMATH}\ensuremath{\mathit{\langle v_{1 1},v_{1 2}\rangle  \sim _{m}^{\Sigma } \langle v_{2 1},v_{2 2}\rangle }}}} holds
      when each element of the pair are pairwise related.
   \item  The relation on function values {{\color{\colorMATH}\ensuremath{\mathit{\langle \lambda x.\hspace*{0.33em}e_{1}\mathrel{|}\rho _{1}\rangle  \sim _{n}^{\Sigma } \langle \lambda x.\hspace*{0.33em}e_{2}\mathrel{|}\rho _{2}\rangle }}}}
      holds when each closure returns related output configurations
      when evaluated with related inputs.
   \item  The relation on locations {{\color{\colorMATH}\ensuremath{\mathit{\ell _{1} \sim _{n}^{\Sigma } \ell _{2}}}}} holds when the two
      locations are equal.
   \end{enumerate}
\item  {{\color{\colorMATH}\ensuremath{\mathit{\rho _{1} \sim _{n}^{\Sigma } \rho _{2}}}}} holds when lexical environments {{\color{\colorMATH}\ensuremath{\mathit{\rho _{1}}}}} and {{\color{\colorMATH}\ensuremath{\mathit{\rho _{2}}}}} map
   all variables to related values.
\item  {{\color{\colorMATH}\ensuremath{\mathit{\sigma _{1} \sim _{n}^{\Sigma } \sigma _{2}}}}} holds when dynamic environments {{\color{\colorMATH}\ensuremath{\mathit{\sigma _{1}}}}} and {{\color{\colorMATH}\ensuremath{\mathit{\sigma _{2}}}}} map
   all locations to related values.
\end{enumerate}
Note that the definitions of {{\color{\colorMATH}\ensuremath{\mathit{\rho _{1},\sigma _{1},e_{1} \sim _{n}^{\Sigma } \rho _{2},\sigma _{2},e_{2}}}}} and {{\color{\colorMATH}\ensuremath{\mathit{v \sim _{n}^{\Sigma }
v}}}} are mutually recursive, but are well founded due to the decrement
of the step index in the former relation. We show the full definition
of these relations in Figure~\ref{fig:logical-relations}.

\begin{figure*}
\begin{framed}
\vspace{-2ex}
\begingroup\color{\colorMATH}\begin{gather*}\begin{tabularx}{\linewidth}{>{\centering\arraybackslash\(}X<{\)}}\hfill\hspace{0pt}
    \begin{array}{rcl
    } {}\rceil \underline{\hspace{0.66em}}\lceil {}^{r} &{}\in {}& \hat {\mathbb{R}} \rightarrow  \hat {\mathbb{R}}
    \cr  {}\rceil \underline{\hspace{0.66em}}\lceil {}^{r} &{}\in {}& {\text{senv}} \rightarrow  {\text{sens}}
    \end{array}
    \hfill\hspace{0pt}
    \begin{array}{rcl
    } {}\rceil r^{\prime}\lceil {}^{r} &{}\triangleq {}& \left\{ \begin{array}{l@{\hspace*{1.00em}}c@{\hspace*{1.00em}}l
                  } 0 &{}{\textit{if}}{}& r^{\prime} = 0
                  \cr  r &{}{\textit{if}}{}& r^{\prime} \neq  0
                  \end{array}\right.
    \cr  {}\rceil \Sigma \lceil {}^{r}(o) &{}\triangleq {}& {}\rceil \Sigma (o)\lceil {}^{r}
    \end{array}
    \hfill\hspace{0pt}
    \begin{array}{rcl
    } {\text{alloc}}({\mathcal{L}}) \notin  {\mathcal{L}} \in  \wp ({\text{loc}})
    \end{array}
    \hfill\hspace{0pt}
    {\textbf{Z}} = \{ o \mapsto  0\}
    \hfill\hspace{0pt}
  \end{tabularx}
\end{gather*}\endgroup
\vspace{-4ex}
\begingroup\color{\colorMATH}\begin{gather*}\begin{tabularx}{\linewidth}{>{\centering\arraybackslash\(}X<{\)}}\hfill\hspace{0pt}\begingroup\color{\colorTEXT}\boxed{\begingroup\color{\colorMATH} \rho  \vdash  \mathrlap{\raisebox{-5pt}{\color{black!50}$\llcorner$}}\mathrlap{\raisebox{6pt}{\color{black!50}$\ulcorner$}}\hspace{2pt}\vphantom{\underset x{\overset XX}}\sigma ,e\hspace{2pt}\mathllap{\raisebox{-5pt}{\color{black!50}$\lrcorner$}}\mathllap{\raisebox{6pt}{\color{black!50}$\urcorner$}}\vphantom{\underset x{\overset XX}} \Downarrow _{n} \mathrlap{\raisebox{-5pt}{\color{black!50}$\llcorner$}}\mathrlap{\raisebox{6pt}{\color{black!50}$\ulcorner$}}\hspace{2pt}\vphantom{\underset x{\overset XX}}\sigma ,v\hspace{2pt}\mathllap{\raisebox{-5pt}{\color{black!50}$\lrcorner$}}\mathllap{\raisebox{6pt}{\color{black!50}$\urcorner$}}\vphantom{\underset x{\overset XX}} \endgroup}\endgroup\end{tabularx}\end{gather*}\endgroup
\vspace{-2ex}
\begingroup
\def\MathparLineskip{\lineskip=6pt}
\begingroup\color{\colorMATH}\begin{mathpar} \inferrule*[lab={\textsc{ Var}}
   ]{ }{
      \rho  \vdash  \mathrlap{\raisebox{-5pt}{\color{black!50}$\llcorner$}}\mathrlap{\raisebox{6pt}{\color{black!50}$\ulcorner$}}\hspace{2pt}\vphantom{\underset x{\overset XX}}\sigma ,x\hspace{2pt}\mathllap{\raisebox{-5pt}{\color{black!50}$\lrcorner$}}\mathllap{\raisebox{6pt}{\color{black!50}$\urcorner$}}\vphantom{\underset x{\overset XX}} \Downarrow _{0} \mathrlap{\raisebox{-5pt}{\color{black!50}$\llcorner$}}\mathrlap{\raisebox{6pt}{\color{black!50}$\ulcorner$}}\hspace{2pt}\vphantom{\underset x{\overset XX}}\sigma ,\rho (x)\hspace{2pt}\mathllap{\raisebox{-5pt}{\color{black!50}$\lrcorner$}}\mathllap{\raisebox{6pt}{\color{black!50}$\urcorner$}}\vphantom{\underset x{\overset XX}}
   }
\and \inferrule*[lab={\textsc{ Real}}
   ]{ }{
      \rho  \vdash  \mathrlap{\raisebox{-5pt}{\color{black!50}$\llcorner$}}\mathrlap{\raisebox{6pt}{\color{black!50}$\ulcorner$}}\hspace{2pt}\vphantom{\underset x{\overset XX}}\sigma ,r\hspace{2pt}\mathllap{\raisebox{-5pt}{\color{black!50}$\lrcorner$}}\mathllap{\raisebox{6pt}{\color{black!50}$\urcorner$}}\vphantom{\underset x{\overset XX}} \Downarrow _{0} \mathrlap{\raisebox{-5pt}{\color{black!50}$\llcorner$}}\mathrlap{\raisebox{6pt}{\color{black!50}$\ulcorner$}}\hspace{2pt}\vphantom{\underset x{\overset XX}}\sigma ,r@_{{{\color{\colorSYNTAX}\texttt{disc}}}}^{{\textbf{Z}}}\hspace{2pt}\mathllap{\raisebox{-5pt}{\color{black!50}$\lrcorner$}}\mathllap{\raisebox{6pt}{\color{black!50}$\urcorner$}}\vphantom{\underset x{\overset XX}}
   }
\and \inferrule*[lab={\textsc{ Fun}}
   ]{ }{
      \rho  \vdash  \mathrlap{\raisebox{-5pt}{\color{black!50}$\llcorner$}}\mathrlap{\raisebox{6pt}{\color{black!50}$\ulcorner$}}\hspace{2pt}\vphantom{\underset x{\overset XX}}\sigma ,\lambda x.\hspace*{0.33em} e\hspace{2pt}\mathllap{\raisebox{-5pt}{\color{black!50}$\lrcorner$}}\mathllap{\raisebox{6pt}{\color{black!50}$\urcorner$}}\vphantom{\underset x{\overset XX}} \Downarrow _{0} \mathrlap{\raisebox{-5pt}{\color{black!50}$\llcorner$}}\mathrlap{\raisebox{6pt}{\color{black!50}$\ulcorner$}}\hspace{2pt}\vphantom{\underset x{\overset XX}}\sigma ,\langle \lambda x.\hspace*{0.33em}e\mathrel{|}\rho \rangle \hspace{2pt}\mathllap{\raisebox{-5pt}{\color{black!50}$\lrcorner$}}\mathllap{\raisebox{6pt}{\color{black!50}$\urcorner$}}\vphantom{\underset x{\overset XX}}
   }
\and \inferrule*[lab={\textsc{ Plus}}
   ]{{\begin{array}{rclclclclcl
      } \rho  &{}\vdash {}& \mathrlap{\raisebox{-5pt}{\color{black!50}$\llcorner$}}\mathrlap{\raisebox{6pt}{\color{black!50}$\ulcorner$}}\hspace{2pt}\vphantom{\underset x{\overset XX}}\sigma  &{},{}&e_{1}\hspace{2pt}\mathllap{\raisebox{-5pt}{\color{black!50}$\lrcorner$}}\mathllap{\raisebox{6pt}{\color{black!50}$\urcorner$}}\vphantom{\underset x{\overset XX}} &{}\Downarrow _{n_{1}}{}& \mathrlap{\raisebox{-5pt}{\color{black!50}$\llcorner$}}\mathrlap{\raisebox{6pt}{\color{black!50}$\ulcorner$}}\hspace{2pt}\vphantom{\underset x{\overset XX}}\sigma ^{\prime}&{},{}&r_{1}&{}@{}&^{\Sigma _{1}}_{m_{1}}\hspace{2pt}\mathllap{\raisebox{-5pt}{\color{black!50}$\lrcorner$}}\mathllap{\raisebox{6pt}{\color{black!50}$\urcorner$}}\vphantom{\underset x{\overset XX}}
      \cr  \rho  &{}\vdash {}& \mathrlap{\raisebox{-5pt}{\color{black!50}$\llcorner$}}\mathrlap{\raisebox{6pt}{\color{black!50}$\ulcorner$}}\hspace{2pt}\vphantom{\underset x{\overset XX}}\sigma ^{\prime}&{},{}&e_{2}\hspace{2pt}\mathllap{\raisebox{-5pt}{\color{black!50}$\lrcorner$}}\mathllap{\raisebox{6pt}{\color{black!50}$\urcorner$}}\vphantom{\underset x{\overset XX}} &{}\Downarrow _{n_{2}}{}& \mathrlap{\raisebox{-5pt}{\color{black!50}$\llcorner$}}\mathrlap{\raisebox{6pt}{\color{black!50}$\ulcorner$}}\hspace{2pt}\vphantom{\underset x{\overset XX}}\sigma ^{\prime \prime}&{},{}&r_{2}&{}@{}&^{\Sigma _{2}}_{m_{2}}\hspace{2pt}\mathllap{\raisebox{-5pt}{\color{black!50}$\lrcorner$}}\mathllap{\raisebox{6pt}{\color{black!50}$\urcorner$}}\vphantom{\underset x{\overset XX}}
      \end{array}}
   \\ }{
      \rho  \vdash  \mathrlap{\raisebox{-5pt}{\color{black!50}$\llcorner$}}\mathrlap{\raisebox{6pt}{\color{black!50}$\ulcorner$}}\hspace{2pt}\vphantom{\underset x{\overset XX}}\sigma ,e_{1} + e_{2}\hspace{2pt}\mathllap{\raisebox{-5pt}{\color{black!50}$\lrcorner$}}\mathllap{\raisebox{6pt}{\color{black!50}$\urcorner$}}\vphantom{\underset x{\overset XX}} \Downarrow _{n_{1}+n_{2}} \mathrlap{\raisebox{-5pt}{\color{black!50}$\llcorner$}}\mathrlap{\raisebox{6pt}{\color{black!50}$\ulcorner$}}\hspace{2pt}\vphantom{\underset x{\overset XX}}\sigma ^{\prime \prime},(r_{1}+r_{2})@^{\Sigma _{1}+\Sigma _{2}}_{m_{1}\sqcup m_{2}}\hspace{2pt}\mathllap{\raisebox{-5pt}{\color{black!50}$\lrcorner$}}\mathllap{\raisebox{6pt}{\color{black!50}$\urcorner$}}\vphantom{\underset x{\overset XX}}
   }
\and \inferrule*[lab={\textsc{ Times-L}}
   ]{{\begin{array}{rclclclclcl
      } \rho  &{}\vdash {}& \mathrlap{\raisebox{-5pt}{\color{black!50}$\llcorner$}}\mathrlap{\raisebox{6pt}{\color{black!50}$\ulcorner$}}\hspace{2pt}\vphantom{\underset x{\overset XX}}\sigma  &{},{}&e_{1}\hspace{2pt}\mathllap{\raisebox{-5pt}{\color{black!50}$\lrcorner$}}\mathllap{\raisebox{6pt}{\color{black!50}$\urcorner$}}\vphantom{\underset x{\overset XX}} &{}\Downarrow _{n_{1}}{}& \mathrlap{\raisebox{-5pt}{\color{black!50}$\llcorner$}}\mathrlap{\raisebox{6pt}{\color{black!50}$\ulcorner$}}\hspace{2pt}\vphantom{\underset x{\overset XX}}\sigma ^{\prime}&{},{}&r_{1}&{}@{}&^{{\textbf{Z}}}_{m_{1}}\hspace{2pt}\mathllap{\raisebox{-5pt}{\color{black!50}$\lrcorner$}}\mathllap{\raisebox{6pt}{\color{black!50}$\urcorner$}}\vphantom{\underset x{\overset XX}}
      \cr  \rho  &{}\vdash {}& \mathrlap{\raisebox{-5pt}{\color{black!50}$\llcorner$}}\mathrlap{\raisebox{6pt}{\color{black!50}$\ulcorner$}}\hspace{2pt}\vphantom{\underset x{\overset XX}}\sigma ^{\prime}&{},{}&e_{2}\hspace{2pt}\mathllap{\raisebox{-5pt}{\color{black!50}$\lrcorner$}}\mathllap{\raisebox{6pt}{\color{black!50}$\urcorner$}}\vphantom{\underset x{\overset XX}} &{}\Downarrow _{n_{2}}{}& \mathrlap{\raisebox{-5pt}{\color{black!50}$\llcorner$}}\mathrlap{\raisebox{6pt}{\color{black!50}$\ulcorner$}}\hspace{2pt}\vphantom{\underset x{\overset XX}}\sigma ^{\prime \prime}&{},{}&r_{2}&{}@{}&^{\Sigma _{2}}_{m_{2}}\hspace{2pt}\mathllap{\raisebox{-5pt}{\color{black!50}$\lrcorner$}}\mathllap{\raisebox{6pt}{\color{black!50}$\urcorner$}}\vphantom{\underset x{\overset XX}}
      \end{array}}
   \\ }{
      \rho  \vdash  \mathrlap{\raisebox{-5pt}{\color{black!50}$\llcorner$}}\mathrlap{\raisebox{6pt}{\color{black!50}$\ulcorner$}}\hspace{2pt}\vphantom{\underset x{\overset XX}}\sigma ,e_{1} \ltimes  e_{2}\hspace{2pt}\mathllap{\raisebox{-5pt}{\color{black!50}$\lrcorner$}}\mathllap{\raisebox{6pt}{\color{black!50}$\urcorner$}}\vphantom{\underset x{\overset XX}} \Downarrow _{n_{1}+n_{2}} \mathrlap{\raisebox{-5pt}{\color{black!50}$\llcorner$}}\mathrlap{\raisebox{6pt}{\color{black!50}$\ulcorner$}}\hspace{2pt}\vphantom{\underset x{\overset XX}}\sigma ^{\prime \prime},(r_{1}\times r_{2})@^{r_{1}\Sigma _{2}}_{m_{2}}\hspace{2pt}\mathllap{\raisebox{-5pt}{\color{black!50}$\lrcorner$}}\mathllap{\raisebox{6pt}{\color{black!50}$\urcorner$}}\vphantom{\underset x{\overset XX}}
   }
\and \inferrule*[lab={\textsc{ Times-R}}
   ]{{\begin{array}{rclclclclcl
      } \rho  &{}\vdash {}& \mathrlap{\raisebox{-5pt}{\color{black!50}$\llcorner$}}\mathrlap{\raisebox{6pt}{\color{black!50}$\ulcorner$}}\hspace{2pt}\vphantom{\underset x{\overset XX}}\sigma  &{},{}&e_{1}\hspace{2pt}\mathllap{\raisebox{-5pt}{\color{black!50}$\lrcorner$}}\mathllap{\raisebox{6pt}{\color{black!50}$\urcorner$}}\vphantom{\underset x{\overset XX}} &{}\Downarrow _{n_{1}}{}& \mathrlap{\raisebox{-5pt}{\color{black!50}$\llcorner$}}\mathrlap{\raisebox{6pt}{\color{black!50}$\ulcorner$}}\hspace{2pt}\vphantom{\underset x{\overset XX}}\sigma ^{\prime}&{},{}&r_{1}&{}@{}&^{\Sigma _{1}}_{m_{1}}\hspace{2pt}\mathllap{\raisebox{-5pt}{\color{black!50}$\lrcorner$}}\mathllap{\raisebox{6pt}{\color{black!50}$\urcorner$}}\vphantom{\underset x{\overset XX}}
      \cr  \rho  &{}\vdash {}& \mathrlap{\raisebox{-5pt}{\color{black!50}$\llcorner$}}\mathrlap{\raisebox{6pt}{\color{black!50}$\ulcorner$}}\hspace{2pt}\vphantom{\underset x{\overset XX}}\sigma ^{\prime}&{},{}&e_{2}\hspace{2pt}\mathllap{\raisebox{-5pt}{\color{black!50}$\lrcorner$}}\mathllap{\raisebox{6pt}{\color{black!50}$\urcorner$}}\vphantom{\underset x{\overset XX}} &{}\Downarrow _{n_{2}}{}& \mathrlap{\raisebox{-5pt}{\color{black!50}$\llcorner$}}\mathrlap{\raisebox{6pt}{\color{black!50}$\ulcorner$}}\hspace{2pt}\vphantom{\underset x{\overset XX}}\sigma ^{\prime \prime}&{},{}&r_{2}&{}@{}&^{{\textbf{Z}}}_{m_{2}}\hspace{2pt}\mathllap{\raisebox{-5pt}{\color{black!50}$\lrcorner$}}\mathllap{\raisebox{6pt}{\color{black!50}$\urcorner$}}\vphantom{\underset x{\overset XX}}
      \end{array}}
   \\ }{
      \rho  \vdash  \mathrlap{\raisebox{-5pt}{\color{black!50}$\llcorner$}}\mathrlap{\raisebox{6pt}{\color{black!50}$\ulcorner$}}\hspace{2pt}\vphantom{\underset x{\overset XX}}\sigma ,e_{1} \rtimes  e_{2}\hspace{2pt}\mathllap{\raisebox{-5pt}{\color{black!50}$\lrcorner$}}\mathllap{\raisebox{6pt}{\color{black!50}$\urcorner$}}\vphantom{\underset x{\overset XX}} \Downarrow _{n_{1}+n_{2}} \mathrlap{\raisebox{-5pt}{\color{black!50}$\llcorner$}}\mathrlap{\raisebox{6pt}{\color{black!50}$\ulcorner$}}\hspace{2pt}\vphantom{\underset x{\overset XX}}\sigma ^{\prime \prime},(r_{1}\times r_{2})@^{r_{2}\Sigma _{1}}_{m_{1}}\hspace{2pt}\mathllap{\raisebox{-5pt}{\color{black!50}$\lrcorner$}}\mathllap{\raisebox{6pt}{\color{black!50}$\urcorner$}}\vphantom{\underset x{\overset XX}}
   }
\and \inferrule*[lab={\textsc{ IfZ-T}}
   ]{{\begin{array}[b]{rclclclclcl
      } \rho  &{}\vdash {}& \mathrlap{\raisebox{-5pt}{\color{black!50}$\llcorner$}}\mathrlap{\raisebox{6pt}{\color{black!50}$\ulcorner$}}\hspace{2pt}\vphantom{\underset x{\overset XX}}\sigma  &{},{}&e_{1}\hspace{2pt}\mathllap{\raisebox{-5pt}{\color{black!50}$\lrcorner$}}\mathllap{\raisebox{6pt}{\color{black!50}$\urcorner$}}\vphantom{\underset x{\overset XX}} &{}\Downarrow _{n_{1}}{}& \mathrlap{\raisebox{-5pt}{\color{black!50}$\llcorner$}}\mathrlap{\raisebox{6pt}{\color{black!50}$\ulcorner$}}\hspace{2pt}\vphantom{\underset x{\overset XX}}\sigma ^{\prime}&{},{}&r_{1}&{}@{}&^{{\textbf{Z}}}_{m_{1}}\hspace{2pt}\mathllap{\raisebox{-5pt}{\color{black!50}$\lrcorner$}}\mathllap{\raisebox{6pt}{\color{black!50}$\urcorner$}}\vphantom{\underset x{\overset XX}}
      \cr  \rho  &{}\vdash {}& \mathrlap{\raisebox{-5pt}{\color{black!50}$\llcorner$}}\mathrlap{\raisebox{6pt}{\color{black!50}$\ulcorner$}}\hspace{2pt}\vphantom{\underset x{\overset XX}}\sigma ^{\prime}&{},{}&e_{2}\hspace{2pt}\mathllap{\raisebox{-5pt}{\color{black!50}$\lrcorner$}}\mathllap{\raisebox{6pt}{\color{black!50}$\urcorner$}}\vphantom{\underset x{\overset XX}} &{}\Downarrow _{n_{2}}{}& \mathrlap{\raisebox{-5pt}{\color{black!50}$\llcorner$}}\mathrlap{\raisebox{6pt}{\color{black!50}$\ulcorner$}}\hspace{2pt}\vphantom{\underset x{\overset XX}}\sigma ^{\prime \prime}&{},{}&v_{2}\hspace{2pt}\mathllap{\raisebox{-5pt}{\color{black!50}$\lrcorner$}}\mathllap{\raisebox{6pt}{\color{black!50}$\urcorner$}}\vphantom{\underset x{\overset XX}}
      \end{array}}
   \\ r_{1} = 0
      }{
      \rho  \vdash  \mathrlap{\raisebox{-5pt}{\color{black!50}$\llcorner$}}\mathrlap{\raisebox{6pt}{\color{black!50}$\ulcorner$}}\hspace{2pt}\vphantom{\underset x{\overset XX}}\sigma ,{{\color{\colorSYNTAX}\texttt{if0}}}(e_{1})\{ e_{2}\} \{ e_{3}\} \hspace{2pt}\mathllap{\raisebox{-5pt}{\color{black!50}$\lrcorner$}}\mathllap{\raisebox{6pt}{\color{black!50}$\urcorner$}}\vphantom{\underset x{\overset XX}} \Downarrow _{n_{1}+n_{2}} \mathrlap{\raisebox{-5pt}{\color{black!50}$\llcorner$}}\mathrlap{\raisebox{6pt}{\color{black!50}$\ulcorner$}}\hspace{2pt}\vphantom{\underset x{\overset XX}}\sigma ^{\prime \prime},v_{2}\hspace{2pt}\mathllap{\raisebox{-5pt}{\color{black!50}$\lrcorner$}}\mathllap{\raisebox{6pt}{\color{black!50}$\urcorner$}}\vphantom{\underset x{\overset XX}}
   }
\and \inferrule*[lab={\textsc{ IfZ-F}}
   ]{{\begin{array}[b]{rclclclclcl
      } \rho  &{}\vdash {}& \mathrlap{\raisebox{-5pt}{\color{black!50}$\llcorner$}}\mathrlap{\raisebox{6pt}{\color{black!50}$\ulcorner$}}\hspace{2pt}\vphantom{\underset x{\overset XX}}\sigma  &{},{}&e_{1}\hspace{2pt}\mathllap{\raisebox{-5pt}{\color{black!50}$\lrcorner$}}\mathllap{\raisebox{6pt}{\color{black!50}$\urcorner$}}\vphantom{\underset x{\overset XX}} &{}\Downarrow _{n_{1}}{}& \mathrlap{\raisebox{-5pt}{\color{black!50}$\llcorner$}}\mathrlap{\raisebox{6pt}{\color{black!50}$\ulcorner$}}\hspace{2pt}\vphantom{\underset x{\overset XX}}\sigma ^{\prime}&{},{}&r_{1}&{}@{}&^{{\textbf{Z}}}_{m_{1}}\hspace{2pt}\mathllap{\raisebox{-5pt}{\color{black!50}$\lrcorner$}}\mathllap{\raisebox{6pt}{\color{black!50}$\urcorner$}}\vphantom{\underset x{\overset XX}}
      \cr  \rho  &{}\vdash {}& \mathrlap{\raisebox{-5pt}{\color{black!50}$\llcorner$}}\mathrlap{\raisebox{6pt}{\color{black!50}$\ulcorner$}}\hspace{2pt}\vphantom{\underset x{\overset XX}}\sigma ^{\prime}&{},{}&e_{3}\hspace{2pt}\mathllap{\raisebox{-5pt}{\color{black!50}$\lrcorner$}}\mathllap{\raisebox{6pt}{\color{black!50}$\urcorner$}}\vphantom{\underset x{\overset XX}} &{}\Downarrow _{n_{2}}{}& \mathrlap{\raisebox{-5pt}{\color{black!50}$\llcorner$}}\mathrlap{\raisebox{6pt}{\color{black!50}$\ulcorner$}}\hspace{2pt}\vphantom{\underset x{\overset XX}}\sigma ^{\prime \prime}&{},{}&v_{3}\hspace{2pt}\mathllap{\raisebox{-5pt}{\color{black!50}$\lrcorner$}}\mathllap{\raisebox{6pt}{\color{black!50}$\urcorner$}}\vphantom{\underset x{\overset XX}}
      \end{array}}
   \\ r_{1} \neq  0
      }{
      \rho  \vdash  \mathrlap{\raisebox{-5pt}{\color{black!50}$\llcorner$}}\mathrlap{\raisebox{6pt}{\color{black!50}$\ulcorner$}}\hspace{2pt}\vphantom{\underset x{\overset XX}}\sigma ,{{\color{\colorSYNTAX}\texttt{if0}}}(e_{1})\{ e_{2}\} \{ e_{3}\} \hspace{2pt}\mathllap{\raisebox{-5pt}{\color{black!50}$\lrcorner$}}\mathllap{\raisebox{6pt}{\color{black!50}$\urcorner$}}\vphantom{\underset x{\overset XX}} \Downarrow _{n_{1}+n_{2}} \mathrlap{\raisebox{-5pt}{\color{black!50}$\llcorner$}}\mathrlap{\raisebox{6pt}{\color{black!50}$\ulcorner$}}\hspace{2pt}\vphantom{\underset x{\overset XX}}\sigma ^{\prime \prime},v_{3}\hspace{2pt}\mathllap{\raisebox{-5pt}{\color{black!50}$\lrcorner$}}\mathllap{\raisebox{6pt}{\color{black!50}$\urcorner$}}\vphantom{\underset x{\overset XX}}
   }
\and \inferrule*[lab={\textsc{ Pair}}
   ]{{\begin{array}[b]{rclclclcl
      } \rho  &{}\vdash {}& \mathrlap{\raisebox{-5pt}{\color{black!50}$\llcorner$}}\mathrlap{\raisebox{6pt}{\color{black!50}$\ulcorner$}}\hspace{2pt}\vphantom{\underset x{\overset XX}}\sigma  &{},{}&e_{1}\hspace{2pt}\mathllap{\raisebox{-5pt}{\color{black!50}$\lrcorner$}}\mathllap{\raisebox{6pt}{\color{black!50}$\urcorner$}}\vphantom{\underset x{\overset XX}} &{}\Downarrow _{n_{1}}{}& \mathrlap{\raisebox{-5pt}{\color{black!50}$\llcorner$}}\mathrlap{\raisebox{6pt}{\color{black!50}$\ulcorner$}}\hspace{2pt}\vphantom{\underset x{\overset XX}}\sigma ^{\prime}&{},{}&v_{1}\hspace{2pt}\mathllap{\raisebox{-5pt}{\color{black!50}$\lrcorner$}}\mathllap{\raisebox{6pt}{\color{black!50}$\urcorner$}}\vphantom{\underset x{\overset XX}}
      \cr  \rho  &{}\vdash {}& \mathrlap{\raisebox{-5pt}{\color{black!50}$\llcorner$}}\mathrlap{\raisebox{6pt}{\color{black!50}$\ulcorner$}}\hspace{2pt}\vphantom{\underset x{\overset XX}}\sigma ^{\prime}&{},{}&e_{2}\hspace{2pt}\mathllap{\raisebox{-5pt}{\color{black!50}$\lrcorner$}}\mathllap{\raisebox{6pt}{\color{black!50}$\urcorner$}}\vphantom{\underset x{\overset XX}} &{}\Downarrow _{n_{2}}{}& \mathrlap{\raisebox{-5pt}{\color{black!50}$\llcorner$}}\mathrlap{\raisebox{6pt}{\color{black!50}$\ulcorner$}}\hspace{2pt}\vphantom{\underset x{\overset XX}}\sigma ^{\prime \prime}&{},{}&v_{2}\hspace{2pt}\mathllap{\raisebox{-5pt}{\color{black!50}$\lrcorner$}}\mathllap{\raisebox{6pt}{\color{black!50}$\urcorner$}}\vphantom{\underset x{\overset XX}}
      \end{array}}
      }{
      \rho  \vdash  \mathrlap{\raisebox{-5pt}{\color{black!50}$\llcorner$}}\mathrlap{\raisebox{6pt}{\color{black!50}$\ulcorner$}}\hspace{2pt}\vphantom{\underset x{\overset XX}}\sigma ,\langle e_{1},e_{2}\rangle \hspace{2pt}\mathllap{\raisebox{-5pt}{\color{black!50}$\lrcorner$}}\mathllap{\raisebox{6pt}{\color{black!50}$\urcorner$}}\vphantom{\underset x{\overset XX}} \Downarrow _{n_{1}+n_{2}} \mathrlap{\raisebox{-5pt}{\color{black!50}$\llcorner$}}\mathrlap{\raisebox{6pt}{\color{black!50}$\ulcorner$}}\hspace{2pt}\vphantom{\underset x{\overset XX}}\sigma ^{\prime \prime},\langle v_{1},v_{2}\rangle \hspace{2pt}\mathllap{\raisebox{-5pt}{\color{black!50}$\lrcorner$}}\mathllap{\raisebox{6pt}{\color{black!50}$\urcorner$}}\vphantom{\underset x{\overset XX}}
   }
\and \inferrule*[lab={\textsc{ Proj}}
   ]{ \rho  \vdash  \mathrlap{\raisebox{-5pt}{\color{black!50}$\llcorner$}}\mathrlap{\raisebox{6pt}{\color{black!50}$\ulcorner$}}\hspace{2pt}\vphantom{\underset x{\overset XX}}\sigma ,e\hspace{2pt}\mathllap{\raisebox{-5pt}{\color{black!50}$\lrcorner$}}\mathllap{\raisebox{6pt}{\color{black!50}$\urcorner$}}\vphantom{\underset x{\overset XX}} \Downarrow _{n} \mathrlap{\raisebox{-5pt}{\color{black!50}$\llcorner$}}\mathrlap{\raisebox{6pt}{\color{black!50}$\ulcorner$}}\hspace{2pt}\vphantom{\underset x{\overset XX}}\sigma ^{\prime},\langle v_{1},v_{2}\rangle \hspace{2pt}\mathllap{\raisebox{-5pt}{\color{black!50}$\lrcorner$}}\mathllap{\raisebox{6pt}{\color{black!50}$\urcorner$}}\vphantom{\underset x{\overset XX}}
      }{
      \rho  \vdash  \mathrlap{\raisebox{-5pt}{\color{black!50}$\llcorner$}}\mathrlap{\raisebox{6pt}{\color{black!50}$\ulcorner$}}\hspace{2pt}\vphantom{\underset x{\overset XX}}\sigma ,\pi _{n^{\prime}}(e)\hspace{2pt}\mathllap{\raisebox{-5pt}{\color{black!50}$\lrcorner$}}\mathllap{\raisebox{6pt}{\color{black!50}$\urcorner$}}\vphantom{\underset x{\overset XX}} \Downarrow _{n} \mathrlap{\raisebox{-5pt}{\color{black!50}$\llcorner$}}\mathrlap{\raisebox{6pt}{\color{black!50}$\ulcorner$}}\hspace{2pt}\vphantom{\underset x{\overset XX}}\sigma ^{\prime},v_{n^{\prime}}\hspace{2pt}\mathllap{\raisebox{-5pt}{\color{black!50}$\lrcorner$}}\mathllap{\raisebox{6pt}{\color{black!50}$\urcorner$}}\vphantom{\underset x{\overset XX}}
   }
\and \inferrule*[lab={\textsc{ Ref}}
   ]{ \rho  \vdash  \mathrlap{\raisebox{-5pt}{\color{black!50}$\llcorner$}}\mathrlap{\raisebox{6pt}{\color{black!50}$\ulcorner$}}\hspace{2pt}\vphantom{\underset x{\overset XX}}\sigma ,e\hspace{2pt}\mathllap{\raisebox{-5pt}{\color{black!50}$\lrcorner$}}\mathllap{\raisebox{6pt}{\color{black!50}$\urcorner$}}\vphantom{\underset x{\overset XX}} \Downarrow _{n} \mathrlap{\raisebox{-5pt}{\color{black!50}$\llcorner$}}\mathrlap{\raisebox{6pt}{\color{black!50}$\ulcorner$}}\hspace{2pt}\vphantom{\underset x{\overset XX}}\sigma ^{\prime},v\hspace{2pt}\mathllap{\raisebox{-5pt}{\color{black!50}$\lrcorner$}}\mathllap{\raisebox{6pt}{\color{black!50}$\urcorner$}}\vphantom{\underset x{\overset XX}}
   \\\\ \ell  = {\text{alloc}}({\text{dom}}(\sigma ))
      }{
      \rho  \vdash  \mathrlap{\raisebox{-5pt}{\color{black!50}$\llcorner$}}\mathrlap{\raisebox{6pt}{\color{black!50}$\ulcorner$}}\hspace{2pt}\vphantom{\underset x{\overset XX}}\sigma ,{{\color{\colorSYNTAX}\texttt{ref}}}(e)\hspace{2pt}\mathllap{\raisebox{-5pt}{\color{black!50}$\lrcorner$}}\mathllap{\raisebox{6pt}{\color{black!50}$\urcorner$}}\vphantom{\underset x{\overset XX}} \Downarrow _{n} \mathrlap{\raisebox{-5pt}{\color{black!50}$\llcorner$}}\mathrlap{\raisebox{6pt}{\color{black!50}$\ulcorner$}}\hspace{2pt}\vphantom{\underset x{\overset XX}}\{ \ell \mapsto v\} \uplus \sigma ^{\prime},\ell \hspace{2pt}\mathllap{\raisebox{-5pt}{\color{black!50}$\lrcorner$}}\mathllap{\raisebox{6pt}{\color{black!50}$\urcorner$}}\vphantom{\underset x{\overset XX}}
   }
\and \inferrule*[lab={\textsc{ Read}}
   ]{ \rho  \vdash  \mathrlap{\raisebox{-5pt}{\color{black!50}$\llcorner$}}\mathrlap{\raisebox{6pt}{\color{black!50}$\ulcorner$}}\hspace{2pt}\vphantom{\underset x{\overset XX}}\sigma ,e\hspace{2pt}\mathllap{\raisebox{-5pt}{\color{black!50}$\lrcorner$}}\mathllap{\raisebox{6pt}{\color{black!50}$\urcorner$}}\vphantom{\underset x{\overset XX}} \Downarrow _{n} \mathrlap{\raisebox{-5pt}{\color{black!50}$\llcorner$}}\mathrlap{\raisebox{6pt}{\color{black!50}$\ulcorner$}}\hspace{2pt}\vphantom{\underset x{\overset XX}}\sigma ^{\prime},\ell \hspace{2pt}\mathllap{\raisebox{-5pt}{\color{black!50}$\lrcorner$}}\mathllap{\raisebox{6pt}{\color{black!50}$\urcorner$}}\vphantom{\underset x{\overset XX}}
      }{
      \rho  \vdash  \mathrlap{\raisebox{-5pt}{\color{black!50}$\llcorner$}}\mathrlap{\raisebox{6pt}{\color{black!50}$\ulcorner$}}\hspace{2pt}\vphantom{\underset x{\overset XX}}\sigma ,{!}e\hspace{2pt}\mathllap{\raisebox{-5pt}{\color{black!50}$\lrcorner$}}\mathllap{\raisebox{6pt}{\color{black!50}$\urcorner$}}\vphantom{\underset x{\overset XX}} \Downarrow _{n} \mathrlap{\raisebox{-5pt}{\color{black!50}$\llcorner$}}\mathrlap{\raisebox{6pt}{\color{black!50}$\ulcorner$}}\hspace{2pt}\vphantom{\underset x{\overset XX}}\sigma ^{\prime},\sigma ^{\prime}(\ell )\hspace{2pt}\mathllap{\raisebox{-5pt}{\color{black!50}$\lrcorner$}}\mathllap{\raisebox{6pt}{\color{black!50}$\urcorner$}}\vphantom{\underset x{\overset XX}}
   }
\and \inferrule*[lab={\textsc{ Write}}
   ]{{\begin{array}[b]{rclclclclcl
      } \rho  &{}\vdash {}& \mathrlap{\raisebox{-5pt}{\color{black!50}$\llcorner$}}\mathrlap{\raisebox{6pt}{\color{black!50}$\ulcorner$}}\hspace{2pt}\vphantom{\underset x{\overset XX}}\sigma  &{},{}&e_{1}\hspace{2pt}\mathllap{\raisebox{-5pt}{\color{black!50}$\lrcorner$}}\mathllap{\raisebox{6pt}{\color{black!50}$\urcorner$}}\vphantom{\underset x{\overset XX}} \Downarrow _{n_{1}} \mathrlap{\raisebox{-5pt}{\color{black!50}$\llcorner$}}\mathrlap{\raisebox{6pt}{\color{black!50}$\ulcorner$}}\hspace{2pt}\vphantom{\underset x{\overset XX}}\sigma ^{\prime}&{},{}&\ell \hspace{2pt}\mathllap{\raisebox{-5pt}{\color{black!50}$\lrcorner$}}\mathllap{\raisebox{6pt}{\color{black!50}$\urcorner$}}\vphantom{\underset x{\overset XX}}
      \cr  \rho  &{}\vdash {}& \mathrlap{\raisebox{-5pt}{\color{black!50}$\llcorner$}}\mathrlap{\raisebox{6pt}{\color{black!50}$\ulcorner$}}\hspace{2pt}\vphantom{\underset x{\overset XX}}\sigma ^{\prime}&{},{}&e_{2}\hspace{2pt}\mathllap{\raisebox{-5pt}{\color{black!50}$\lrcorner$}}\mathllap{\raisebox{6pt}{\color{black!50}$\urcorner$}}\vphantom{\underset x{\overset XX}} \Downarrow _{n_{2}} \mathrlap{\raisebox{-5pt}{\color{black!50}$\llcorner$}}\mathrlap{\raisebox{6pt}{\color{black!50}$\ulcorner$}}\hspace{2pt}\vphantom{\underset x{\overset XX}}\sigma ^{\prime \prime}&{},{}&v\hspace{2pt}\mathllap{\raisebox{-5pt}{\color{black!50}$\lrcorner$}}\mathllap{\raisebox{6pt}{\color{black!50}$\urcorner$}}\vphantom{\underset x{\overset XX}}
      \end{array}}
      }{
      \rho  \vdash  \mathrlap{\raisebox{-5pt}{\color{black!50}$\llcorner$}}\mathrlap{\raisebox{6pt}{\color{black!50}$\ulcorner$}}\hspace{2pt}\vphantom{\underset x{\overset XX}}\sigma ,e_{1} \leftarrow  e_{2}\hspace{2pt}\mathllap{\raisebox{-5pt}{\color{black!50}$\lrcorner$}}\mathllap{\raisebox{6pt}{\color{black!50}$\urcorner$}}\vphantom{\underset x{\overset XX}} \Downarrow _{n} \mathrlap{\raisebox{-5pt}{\color{black!50}$\llcorner$}}\mathrlap{\raisebox{6pt}{\color{black!50}$\ulcorner$}}\hspace{2pt}\vphantom{\underset x{\overset XX}}\sigma ^{\prime \prime}[\ell \mapsto v],v\hspace{2pt}\mathllap{\raisebox{-5pt}{\color{black!50}$\lrcorner$}}\mathllap{\raisebox{6pt}{\color{black!50}$\urcorner$}}\vphantom{\underset x{\overset XX}}
   }
\and \inferrule*[lab={\textsc{ App}}
   ]{       \rho  \vdash  \mathrlap{\raisebox{-5pt}{\color{black!50}$\llcorner$}}\mathrlap{\raisebox{6pt}{\color{black!50}$\ulcorner$}}\hspace{2pt}\vphantom{\underset x{\overset XX}}\sigma  ,e_{1}\hspace{2pt}\mathllap{\raisebox{-5pt}{\color{black!50}$\lrcorner$}}\mathllap{\raisebox{6pt}{\color{black!50}$\urcorner$}}\vphantom{\underset x{\overset XX}} \Downarrow _{n_{1}} \mathrlap{\raisebox{-5pt}{\color{black!50}$\llcorner$}}\mathrlap{\raisebox{6pt}{\color{black!50}$\ulcorner$}}\hspace{2pt}\vphantom{\underset x{\overset XX}}\sigma ^{\prime},\langle \lambda x.\hspace*{0.33em}e\mathrel{|}\rho \rangle \hspace{2pt}\mathllap{\raisebox{-5pt}{\color{black!50}$\lrcorner$}}\mathllap{\raisebox{6pt}{\color{black!50}$\urcorner$}}\vphantom{\underset x{\overset XX}}
   \\       \rho  \vdash  \mathrlap{\raisebox{-5pt}{\color{black!50}$\llcorner$}}\mathrlap{\raisebox{6pt}{\color{black!50}$\ulcorner$}}\hspace{2pt}\vphantom{\underset x{\overset XX}}\sigma ^{\prime},e_{2}\hspace{2pt}\mathllap{\raisebox{-5pt}{\color{black!50}$\lrcorner$}}\mathllap{\raisebox{6pt}{\color{black!50}$\urcorner$}}\vphantom{\underset x{\overset XX}} \Downarrow _{n_{2}} \mathrlap{\raisebox{-5pt}{\color{black!50}$\llcorner$}}\mathrlap{\raisebox{6pt}{\color{black!50}$\ulcorner$}}\hspace{2pt}\vphantom{\underset x{\overset XX}}\sigma ^{\prime \prime},v\hspace{2pt}\mathllap{\raisebox{-5pt}{\color{black!50}$\lrcorner$}}\mathllap{\raisebox{6pt}{\color{black!50}$\urcorner$}}\vphantom{\underset x{\overset XX}}
   \\ \{ x\mapsto v\} \uplus \rho  \vdash  \mathrlap{\raisebox{-5pt}{\color{black!50}$\llcorner$}}\mathrlap{\raisebox{6pt}{\color{black!50}$\ulcorner$}}\hspace{2pt}\vphantom{\underset x{\overset XX}}\sigma ^{\prime \prime},e\hspace{2pt}\mathllap{\raisebox{-5pt}{\color{black!50}$\lrcorner$}}\mathllap{\raisebox{6pt}{\color{black!50}$\urcorner$}}\vphantom{\underset x{\overset XX}}  \Downarrow _{n_{3}} \mathrlap{\raisebox{-5pt}{\color{black!50}$\llcorner$}}\mathrlap{\raisebox{6pt}{\color{black!50}$\ulcorner$}}\hspace{2pt}\vphantom{\underset x{\overset XX}}\sigma ^{\prime \prime \prime},v^{\prime}\hspace{2pt}\mathllap{\raisebox{-5pt}{\color{black!50}$\lrcorner$}}\mathllap{\raisebox{6pt}{\color{black!50}$\urcorner$}}\vphantom{\underset x{\overset XX}}
      }{
      \rho  \vdash  \mathrlap{\raisebox{-5pt}{\color{black!50}$\llcorner$}}\mathrlap{\raisebox{6pt}{\color{black!50}$\ulcorner$}}\hspace{2pt}\vphantom{\underset x{\overset XX}}\sigma ,e_{1}(e_{2})\hspace{2pt}\mathllap{\raisebox{-5pt}{\color{black!50}$\lrcorner$}}\mathllap{\raisebox{6pt}{\color{black!50}$\urcorner$}}\vphantom{\underset x{\overset XX}} \Downarrow _{n_{1}+n_{2}+n_{3}+1} \mathrlap{\raisebox{-5pt}{\color{black!50}$\llcorner$}}\mathrlap{\raisebox{6pt}{\color{black!50}$\ulcorner$}}\hspace{2pt}\vphantom{\underset x{\overset XX}}\sigma ^{\prime \prime \prime},v^{\prime}\hspace{2pt}\mathllap{\raisebox{-5pt}{\color{black!50}$\lrcorner$}}\mathllap{\raisebox{6pt}{\color{black!50}$\urcorner$}}\vphantom{\underset x{\overset XX}}
   }
\end{mathpar}\endgroup
\endgroup
\vspace{-4ex}
\end{framed}
\vspace{-2ex}
\caption{Formal Big-step, Step-indexed Semantics and Metafunctions.}
\label{fig:semantics}
\end{figure*}

The function sensitivity of an expression is encoded first as a
statement about expressions respecting relatedness, that is, returning
related outputs when given related inputs, i.e., (assuming no use of
the store) if {{\color{\colorMATH}\ensuremath{\mathit{\rho _{1} \sim _{n}^{\Sigma } \rho _{2}}}}} and {{\color{\colorMATH}\ensuremath{\mathit{\rho _{1} \vdash  \varnothing ,e \Downarrow _{n_{1}} \varnothing ,v_{1}}}}} and {{\color{\colorMATH}\ensuremath{\mathit{\rho _{2} \vdash  \varnothing e
\Downarrow _{n_{2}} \varnothing ,v_{2}}}}} then {{\color{\colorMATH}\ensuremath{\mathit{n_{1} = n_{2}}}}} and {{\color{\colorMATH}\ensuremath{\mathit{v_{1} \sim ^{\Sigma }_{n-n_{1}} v_{2}}}}}. When instantiated
to base types, we have: if {{\color{\colorMATH}\ensuremath{\mathit{\rho _{1} \sim _{n}^{\Sigma } \rho _{2}}}}} and {{\color{\colorMATH}\ensuremath{\mathit{\rho _{1} \vdash  \varnothing ,e \Downarrow _{n_{1}}
\varnothing ,r_{1}@^{\Sigma _{1}}_{m_{1}}}}}} and {{\color{\colorMATH}\ensuremath{\mathit{\rho _{2} \vdash  \varnothing e \Downarrow _{n_{2}} \varnothing ,r_{2}@^{\Sigma _{2}}_{m_{2}}}}}} then {{\color{\colorMATH}\ensuremath{\mathit{n_{1}=n_{2}}}}},
{{\color{\colorMATH}\ensuremath{\mathit{\Sigma _{1}=\Sigma _{2}}}}}, {{\color{\colorMATH}\ensuremath{\mathit{m_{1}=m_{2}}}}} and {{\color{\colorMATH}\ensuremath{\mathit{r_{1} \sim ^{\Sigma \mathord{\cdotp }\Sigma _{1}}_{m_{1}} r_{2}}}}}. The fully general form of
this property is called {\textit{metric preservation}}, which is the main
property we prove in our formal development.

\paragraph{Metric Preservation.}
Metric preservation states that when given related initial
configurations and evaluation outputs, then those outputs are
related. Outputs include result values, as well as dynamic analysis
results, and the relationship that holds demonstrates the soundness
of the analysis results.

\begin{theorem}[Metric Preservation]\ \\
  \begin{itemize}[label={},leftmargin=0pt]\item  \begin{tabular}{r@{\hspace*{1.00em}}l
     } If:      & {{\color{\colorMATH}\ensuremath{\mathit{\rho _{1} \sim _{n}^{\Sigma } \rho _{2}}}}}
     \cr  And:     & {{\color{\colorMATH}\ensuremath{\mathit{\sigma _{1} \sim _{n}^{\Sigma } \sigma _{2}}}}}
     \cr  Then:    & {{\color{\colorMATH}\ensuremath{\mathit{\rho _{1},\sigma _{1},e \sim _{n}^{\Sigma } \rho _{2},\sigma _{2},e}}}}
     \end{tabular}
  \item
  \item  That is, either {{\color{\colorMATH}\ensuremath{\mathit{n = 0}}}} or {{\color{\colorMATH}\ensuremath{\mathit{n = n^{\prime}+1}}}} and\ldots
  \item
  \item  \begin{tabular}{r@{\hspace*{1.00em}}l@{\hspace*{1.00em}}c
     } If:      & {{\color{\colorMATH}\ensuremath{\mathit{n_{1} \leq  n^{\prime}}}}}
     \cr  And:     & {{\color{\colorMATH}\ensuremath{\mathit{\rho _{1} \vdash  \mathrlap{\raisebox{-5pt}{\color{black!50}$\llcorner$}}\mathrlap{\raisebox{6pt}{\color{black!50}$\ulcorner$}}\hspace{2pt}\vphantom{\underset x{\overset XX}}\sigma _{1},e\hspace{2pt}\mathllap{\raisebox{-5pt}{\color{black!50}$\lrcorner$}}\mathllap{\raisebox{6pt}{\color{black!50}$\urcorner$}}\vphantom{\underset x{\overset XX}} \Downarrow _{n_{1}} \mathrlap{\raisebox{-5pt}{\color{black!50}$\llcorner$}}\mathrlap{\raisebox{6pt}{\color{black!50}$\ulcorner$}}\hspace{2pt}\vphantom{\underset x{\overset XX}}\sigma _{1}^{\prime},v_{1}\hspace{2pt}\mathllap{\raisebox{-5pt}{\color{black!50}$\lrcorner$}}\mathllap{\raisebox{6pt}{\color{black!50}$\urcorner$}}\vphantom{\underset x{\overset XX}}}}}}
     \cr  And:     & {{\color{\colorMATH}\ensuremath{\mathit{\rho _{2} \vdash  \mathrlap{\raisebox{-5pt}{\color{black!50}$\llcorner$}}\mathrlap{\raisebox{6pt}{\color{black!50}$\ulcorner$}}\hspace{2pt}\vphantom{\underset x{\overset XX}}\sigma _{2},e\hspace{2pt}\mathllap{\raisebox{-5pt}{\color{black!50}$\lrcorner$}}\mathllap{\raisebox{6pt}{\color{black!50}$\urcorner$}}\vphantom{\underset x{\overset XX}} \Downarrow _{n_{2}} \mathrlap{\raisebox{-5pt}{\color{black!50}$\llcorner$}}\mathrlap{\raisebox{6pt}{\color{black!50}$\ulcorner$}}\hspace{2pt}\vphantom{\underset x{\overset XX}}\sigma _{2}^{\prime},v_{2}\hspace{2pt}\mathllap{\raisebox{-5pt}{\color{black!50}$\lrcorner$}}\mathllap{\raisebox{6pt}{\color{black!50}$\urcorner$}}\vphantom{\underset x{\overset XX}}}}}}
     \cr  Then:    & {{\color{\colorMATH}\ensuremath{\mathit{n_{1} = n_{2}}}}}
     \cr  And:     & {{\color{\colorMATH}\ensuremath{\mathit{\sigma _{1}^{\prime} \sim ^{\Sigma }_{n^{\prime}-n_{1}} \sigma _{2}^{\prime}}}}}
     \cr  And:     & {{\color{\colorMATH}\ensuremath{\mathit{v_{1} \sim ^{\Sigma }_{n^{\prime}-n_{1}} v_{2}}}}}
     \end{tabular}
  \end{itemize}
\end{theorem}
\begin{proof}
  See detailed proof in the appendix.
\end{proof}

\paragraph{Instantiating Metric Preservation.}
Metric preservation is not enough on its own to demonstrate sound
dynamic analysis of function sensitivity. Suppose we execute the
dynamic analysis on program {{\color{\colorMATH}\ensuremath{\mathit{e}}}} with initial environment {{\color{\colorMATH}\ensuremath{\mathit{\rho }}}},
yielding a final store {{\color{\colorMATH}\ensuremath{\mathit{\sigma }}}}, base value {{\color{\colorMATH}\ensuremath{\mathit{r}}}}, sensitivity
environment {{\color{\colorMATH}\ensuremath{\mathit{\Sigma }}}}, metric {{\color{\colorMATH}\ensuremath{\mathit{m}}}} and step-index {{\color{\colorMATH}\ensuremath{\mathit{n}}}} as a result:
\begingroup\color{\colorMATH}\begin{gather*} {{\color{\colorMATH}\ensuremath{\mathit{\rho  \vdash  \varnothing ,e \Downarrow _{n} \sigma ,r@_{m}^{\Sigma }}}}} \end{gather*}\endgroup
To know the sensitivity of {{\color{\colorMATH}\ensuremath{\mathit{e}}}} is to know a bound on two {\textit{arbitrary}}
runs of {{\color{\colorMATH}\ensuremath{\mathit{e}}}}, that is, using two arbitrary environments {{\color{\colorMATH}\ensuremath{\mathit{\rho _{1}}}}} and {{\color{\colorMATH}\ensuremath{\mathit{\rho _{2}}}}}.
Does {{\color{\colorMATH}\ensuremath{\mathit{\Sigma }}}} tell us this? Remarkably, it does, with one small condition:
{{\color{\colorMATH}\ensuremath{\mathit{\rho _{1}}}}} and {{\color{\colorMATH}\ensuremath{\mathit{\rho _{2}}}}} must agree with {{\color{\colorMATH}\ensuremath{\mathit{\rho }}}} on all non-sensitive values. This
is not actually limiting: a non-sensitive value is essentially
auxiliary information; they are constants and fixed for the purposes
of sensitivity and privacy.

We can encode the relationship that environments {{\color{\colorMATH}\ensuremath{\mathit{\rho }}}} and {{\color{\colorMATH}\ensuremath{\mathit{\rho _{1}}}}} agree
on all non-sensitive values as {{\color{\colorMATH}\ensuremath{\mathit{\rho  \sim ^{\Sigma ^{\prime}} \rho _{1}}}}} for any {{\color{\colorMATH}\ensuremath{\mathit{\Sigma ^{\prime}}}}}, and we
allow for environments {{\color{\colorMATH}\ensuremath{\mathit{\rho }}}} and {{\color{\colorMATH}\ensuremath{\mathit{\rho _{1}}}}} to differ on any sensitive value
while agreeing on non-sensitive values as {{\color{\colorMATH}\ensuremath{\mathit{\rho  \sim ^{\{ o\mapsto \infty \} } \rho _{1}}}}}.
Under such an assumption, {{\color{\colorMATH}\ensuremath{\mathit{\Sigma }}}} and {{\color{\colorMATH}\ensuremath{\mathit{m}}}} are sound dynamic analysis
results for two arbitrary runs of {{\color{\colorMATH}\ensuremath{\mathit{e}}}}, i.e., under environments {{\color{\colorMATH}\ensuremath{\mathit{\rho _{1}}}}}
and {{\color{\colorMATH}\ensuremath{\mathit{\rho _{2}}}}}, so long as one of those environments agrees with {{\color{\colorMATH}\ensuremath{\mathit{\rho }}}}---the
environment used to compute the dynamic analysis. We encode this
property formally as the following corollary to metric preservation:
\begin{corollary}[Sound Dynamic Analysis for Sensitivity]\ \\
   \begin{tabular}{r@{\hspace*{1.00em}}l@{\hspace*{1.00em}}l
   } If:   & {{\color{\colorMATH}\ensuremath{\mathit{n_{1} < n}}}}, {{\color{\colorMATH}\ensuremath{\mathit{n_{2} < n}}}} and {{\color{\colorMATH}\ensuremath{\mathit{n_{3} < n}}}}     & {\textit{(H1)}}
   \cr  And:  & {{\color{\colorMATH}\ensuremath{\mathit{\rho  \sim _{n}^{\{ o\mapsto \infty \} } \rho _{1}}}}}                    & {\textit{(H2)}}
   \cr  And:  & {{\color{\colorMATH}\ensuremath{\mathit{\rho  \vdash  \mathrlap{\raisebox{-5pt}{\color{black!50}$\llcorner$}}\mathrlap{\raisebox{6pt}{\color{black!50}$\ulcorner$}}\hspace{2pt}\vphantom{\underset x{\overset XX}}\varnothing ,e\hspace{2pt}\mathllap{\raisebox{-5pt}{\color{black!50}$\lrcorner$}}\mathllap{\raisebox{6pt}{\color{black!50}$\urcorner$}}\vphantom{\underset x{\overset XX}} \Downarrow _{n_{1}} \mathrlap{\raisebox{-5pt}{\color{black!50}$\llcorner$}}\mathrlap{\raisebox{6pt}{\color{black!50}$\ulcorner$}}\hspace{2pt}\vphantom{\underset x{\overset XX}}\sigma ,r@_{m}^{\Sigma }\hspace{2pt}\mathllap{\raisebox{-5pt}{\color{black!50}$\lrcorner$}}\mathllap{\raisebox{6pt}{\color{black!50}$\urcorner$}}\vphantom{\underset x{\overset XX}}}}}}        & {\textit{(H3)}}
   \cr  And:  & {{\color{\colorMATH}\ensuremath{\mathit{\rho _{1} \sim ^{\Sigma ^{\prime}}_{n} \rho _{2}}}}}                      & {\textit{(H4)}}
   \cr  And:  & {{\color{\colorMATH}\ensuremath{\mathit{\rho _{1} \vdash  \mathrlap{\raisebox{-5pt}{\color{black!50}$\llcorner$}}\mathrlap{\raisebox{6pt}{\color{black!50}$\ulcorner$}}\hspace{2pt}\vphantom{\underset x{\overset XX}}\varnothing ,e\hspace{2pt}\mathllap{\raisebox{-5pt}{\color{black!50}$\lrcorner$}}\mathllap{\raisebox{6pt}{\color{black!50}$\urcorner$}}\vphantom{\underset x{\overset XX}} \Downarrow _{n_{2}} \mathrlap{\raisebox{-5pt}{\color{black!50}$\llcorner$}}\mathrlap{\raisebox{6pt}{\color{black!50}$\ulcorner$}}\hspace{2pt}\vphantom{\underset x{\overset XX}}\sigma _{1},r_{1}@^{\Sigma _{1}}_{m_{1}}\hspace{2pt}\mathllap{\raisebox{-5pt}{\color{black!50}$\lrcorner$}}\mathllap{\raisebox{6pt}{\color{black!50}$\urcorner$}}\vphantom{\underset x{\overset XX}}}}}} & {\textit{(H5)}}
   \cr  And:  & {{\color{\colorMATH}\ensuremath{\mathit{\rho _{2} \vdash  \mathrlap{\raisebox{-5pt}{\color{black!50}$\llcorner$}}\mathrlap{\raisebox{6pt}{\color{black!50}$\ulcorner$}}\hspace{2pt}\vphantom{\underset x{\overset XX}}\varnothing ,e\hspace{2pt}\mathllap{\raisebox{-5pt}{\color{black!50}$\lrcorner$}}\mathllap{\raisebox{6pt}{\color{black!50}$\urcorner$}}\vphantom{\underset x{\overset XX}} \Downarrow _{n_{3}} \mathrlap{\raisebox{-5pt}{\color{black!50}$\llcorner$}}\mathrlap{\raisebox{6pt}{\color{black!50}$\ulcorner$}}\hspace{2pt}\vphantom{\underset x{\overset XX}}\sigma _{2},r_{2}@^{\Sigma _{2}}_{m_{2}}\hspace{2pt}\mathllap{\raisebox{-5pt}{\color{black!50}$\lrcorner$}}\mathllap{\raisebox{6pt}{\color{black!50}$\urcorner$}}\vphantom{\underset x{\overset XX}}}}}} & {\textit{(H6)}}
   \cr  Then: & {{\color{\colorMATH}\ensuremath{\mathit{r_{1} \sim ^{\Sigma ^{\prime}\mathord{\cdotp }\Sigma }_{m_{1}} r_{2}}}}}                 & {\textit{(C1)}}
   \end{tabular}
\end{corollary}
\begin{proof}\ \\
  By \nameref{thm:metric-preservation}, {\textit{(H2)}}, {\textit{(H1)}}, {\textit{(H3)}} and
  {\textit{(H5)}} we have {{\color{\colorMATH}\ensuremath{\mathit{\Sigma _{1} = \Sigma }}}} and {{\color{\colorMATH}\ensuremath{\mathit{m_{1} = m}}}}. By
  \nameref{thm:metric-preservation}, {\textit{(H4)}}, {\textit{(H1)}}, {\textit{(H5)}} and
  {\textit{(H6)}} we have proved the goal {\textit{(C1)}}.
\end{proof}
Note that the final results are related using {{\color{\colorMATH}\ensuremath{\mathit{\Sigma }}}}---the analysis result
derived from an execution under {{\color{\colorMATH}\ensuremath{\mathit{\rho }}}}---while {{\color{\colorMATH}\ensuremath{\mathit{r_{1}}}}} and {{\color{\colorMATH}\ensuremath{\mathit{r_{2}}}}} are derived
from executions under unrelated (modulo auxiliary information)
environments {{\color{\colorMATH}\ensuremath{\mathit{\rho _{1}}}}} and {{\color{\colorMATH}\ensuremath{\mathit{\rho _{2}}}}}.

In simpler terms, this corollary shows that even though the dynamic
analysis only sees one particular execution of the program, it is
accurate in describing the sensitivity of the program---even though
the notion of sensitivity considers two arbitrary runs of the
program, including those whose inputs differ entirely from those
used in the dynamic analysis.

\section{Implementation \& Case Studies}
\label{sec:case}

\begin{figure*}
  \begingroup\color{\colorMATH}\begin{gather*} \begin{array}{rllll@{\hspace*{1.00em}}r@{\hspace*{1.00em}}rcl
     } {\textbf{Name}} &{}{}& {\textbf{Type}} &{}{}& & & {\textbf{Conditions}}
     \cr  \hline
        {\text{laplace}} &{}\mathrel{:}{}& (\epsilon  \mathrel{:} {\mathbb{R}}, {\textit{value}} \mathrel{:} {\mathbb{R}})             &{}\rightarrow {}& {\mathbb{R}}                     & {\textit{where}} & {\text{priv}}({\text{laplace}}(\epsilon ,{\textit{value}})) &{}\triangleq {}& \epsilon
     \cr  {\text{gauss}} &{}\mathrel{:}{}& (\epsilon  \mathrel{:} {\mathbb{R}}, \delta  \mathrel{:} {\mathbb{R}}, {\textit{value}} \mathrel{:} {\mathbb{R}})             &{}\rightarrow {}& {\mathbb{R}}                     & {\textit{where}} & {\text{priv}}({\text{gauss}}(\epsilon ,\delta ,{\textit{value}})) &{}\triangleq {}& (\epsilon ,\delta )
     \cr  {\text{ed\_odo}}    &{}\mathrel{:}{}& (f \mathrel{:} A \rightarrow  B , {\textit{in}} \mathrel{:} A )         &{}\rightarrow {}& ({\textit{out}} \mathrel{:} B , (\epsilon ,\delta ) \mathrel{:} ({\mathbb{R}},{\mathbb{R}}))  & {\textit{where}} & {\text{priv}}({\text{ed\_odo}}(f,{\textit{in}}))    &{}\triangleq {}& {\text{priv}}(f({\textit{in}}))
     \cr  {\text{renyi\_odo}}    &{}\mathrel{:}{}& (f \mathrel{:} A \rightarrow  B , {\textit{in}} \mathrel{:} A )         &{}\rightarrow {}& ({\textit{out}} \mathrel{:} B ,  (\alpha ,\epsilon ) \mathrel{:} ({\mathbb{R}},{\mathbb{R}})) & {\textit{where}} & {\text{priv}}({\text{renyi\_odo}}(f,{\textit{in}}))    &{}\triangleq {}& {\text{priv}}(f({\textit{in}}))
     \cr  {\text{ed\_filter}}    &{}\mathrel{:}{}& (f \mathrel{:} A \rightarrow  B , {\textit{in}} \mathrel{:} A , (\epsilon ,\delta ) \mathrel{:} ({\mathbb{R}},{\mathbb{R}}))         &{}\rightarrow {}& ({\textit{out}} \mathrel{:} B ) & {\textit{where}} & (\epsilon ,\delta )    &{}\geq {}& {\text{priv}}(f({\textit{in}}))
     \cr  {\text{renyi\_filter}}    &{}\mathrel{:}{}& (f \mathrel{:} A \rightarrow  B , {\textit{in}} \mathrel{:} A , (\alpha ,\epsilon ) \mathrel{:} ({\mathbb{R}},{\mathbb{R}}))         &{}\rightarrow {}& ({\textit{out}} \mathrel{:} B ) & {\textit{where}} & (\alpha ,\epsilon )    &{}\geq {}& {\text{priv}}(f({\textit{in}}))
     \cr  {\text{conv\_renyi}}    &{}\mathrel{:}{}& (f \mathrel{:} A \rightarrow  B , {\textit{in}} \mathrel{:} A , \delta  \mathrel{:} {\mathbb{R}})         &{}\rightarrow {}& ({\textit{out}} \mathrel{:} B ) & {\textit{where}} & {\text{priv}}(f(...))    &{}={}& (\alpha ,\epsilon )
     \cr            &{} {}&                                  &{} {}&                       & {\textit{and}}   & {\text{conv}}(\alpha , \epsilon , \delta )                               &{}={}& (\epsilon ,\delta )
     \cr  {\text{svt}}    &{}\mathrel{:}{}& (\epsilon  \mathrel{:} {\mathbb{R}}, qs \mathrel{:} [A \rightarrow  B], {\textit{data}} \mathrel{:} [A] , {\textit{t}} \mathrel{:} {\mathbb{R}})         &{}\rightarrow {}& {\mathbb{N}} & {\textit{where}} & {\text{for}}\hspace*{0.33em}{\text{q}}\hspace*{0.33em}{\text{in}}\hspace*{0.33em}{\text{qs}}, sens(q)    &{}={}& 1
     \cr            &{} {}&                                  &{} {}&                       & {\textit{and}}   & {\text{priv}}(svt(\epsilon  , qs , {\textit{data}} , {\textit{t}}))                               &{}\triangleq {}& \epsilon
     \cr  {\text{exp}}    &{}\mathrel{:}{}& (\epsilon  \mathrel{:} {\mathbb{R}}, q \mathrel{:} A \rightarrow  B , {\textit{data}} \mathrel{:} [A] )         &{}\rightarrow {}& {\mathbb{N}} & {\textit{where}} & {\text{priv}}(exp(\epsilon  , q , {\textit{data}}))    &{}\triangleq {}& \epsilon
     \cr  {\text{map}}     &{}\mathrel{:}{}& (f \mathrel{:} A \rightarrow  B , {\textit{in}} \mathrel{:} [A] ) &{}\rightarrow {}& [B]             & {\textit{where}} & {\text{sens}}({\text{map}}(f, \underline{\hspace{0.66em}}))      &{}\triangleq {}& {\text{sens}}(f(\underline{\hspace{0.66em}}))
     \cr  \hline
     \end{array}
  \end{gather*}\endgroup
{\footnotesize $priv$: denotes the privacy leakage of a program given by dynamic analysis; $sens$: denotes the sensitivity of a program given by dynamic analysis; $conv$: represents the conversion equation from renyi to approximate differential privacy }
\newline
{\footnotesize Types are written as follows: the $\rightarrow $ symbol is used to seperate the domain and range of a function, either of which may be given as an atomic type such as a natural number (${\mathbb{N}}$), or as a tuple which is a comma-seperated list of types surrounded by parentheses, or as a symbol ($A$) indicating parametric polymorphism (generics). In some cases, types may also be accompanied with a placeholder name ($\epsilon  \mathrel{:} {\mathbb{R}}$) for further qualification in the $where$ clause.}
\caption{Core API Methods}
\end{figure*}

We have developed a reference implementation of \dduo as a Python
library, using the approaches described in
Sections~\ref{sec:dduo-example}, \ref{sec:sensitivity}, and
\ref{sec:privacy}.

A major goal in the design of \dduo is seamless integration with other
libraries.
Our reference implementation provides initial support for NumPy,
Pandas, and Sklearn.
\dduo provides hooks for tracking both sensitivity and privacy, to
simplify integrating with additional libraries.

We present the case studies which focus on
demonstrating \dduo's (1) similarity to regular Python code, (2)
applicability to complex algorithms, (3) easy integration with
existing libraries. In this section we focus on the Noisy Gradient Descent case study,
other case studies have been moved to the appendix due to space requirements.
We introduce new
\emph{adaptive} variants of algorithms that stop early when
possible to conserve privacy budget. These variants cannot be verified
by prior work using purely static analyses, because their privacy
parameters are chosen adaptively.

\paragraph{Run-time overhead.}
Run-time overhead is a key concern in \dduo's instrumentation for dynamic analysis. Fortunately, experiments on our case studies suggest that the overhead of \dduo's analysis is generally low. Table~\ref{tbl:case_studies} presents the run-time performance overhead of \dduo's analysis as a percentage \emph{increase} of total runtime. The worst overhead time observed in our case studies was less than $60\%$.

In certain rare cases, \dduo's overhead can be much higher. For example, mapping the function \texttt{lambda x: x + 1} over a list of 1 million numbers takes 160x longer under \dduo than in standard Python. The overhead in this case comes from a combination of factors: first, \dduo's \texttt{map} function, itself implemented in Python, is much slower than Python's built-in \texttt{map} operator; second, \dduo's \texttt{map} function requires the creation of a new \wrapper object for each element of the list---a slow operation in Python.

Fortunately, the same strategies for producing high-performance Python code \emph{without} privacy also help reduce \dduo's overhead. Python's performance characteristics have prompted the development of higher-performance libraries like NumPy and Pandas, which essentially provide data-processing combinators that programmers compose. By providing sensitivity annotations for these libraries, we can re-use these high-performance implementations and avoid creating extra objects. As a result, none of our case studies demonstrates the worst-case performance overhead described above.

\begin{table}
\centering
\begin{tabular}{l l l l}
\textbf{Technique} & \textbf{Ref.} & \textbf{Libraries Used} & \textbf{Overhead} \\
\hline
Noisy Gradient Descent & \cite{BST}  & NumPy & $6.42\%$ \\
Multiplicative Weights (MWEM) & \cite{hardt2012simple}  & Pandas & $14.90\%$ \\
Private Naive Bayes Classification & \cite{vaidya2013}  & DiffPrivLib & $12.44\%$ \\
Private Logistic Regression & \cite{Chaudhuri2008}  & DiffPrivLib & $56.33\%$ \\
\end{tabular}
\smallskip
\caption{List of case studies included with the \dduo
  implementation.}
\vspace*{-5mm}
\label{tbl:case_studies}
\end{table}

\paragraph{Case study: gradient descent with NumPy.}
Our first case study (Figure~\ref{fig:gd}) is a simple machine learning algorithm based on \cite{BST} implemented directly with \dduo-instrumented NumPy primitives. The remaining case studies appear in the Appendix.

Given a dataset $X$ which is a list of feature vectors representing training examples, and a vector $y$ which classifies each element of $X$ in a finite set, gradient descent is the process of computing a model (a linear set of weights) which most accurately classifies a new, never seen before training example, based on our pre-existing evidence represented by the model.

Gradient descent works by first specifying a loss function that computes the effectiveness of a model in classifying a given dataset according to its known labels. The algorithm then iteratively computes a model that minimizes the loss function, by calculating the gradient of the loss function and moving the model in the opposite direction of the gradient.

One method of ensuring privacy in gradient descent involves adding noise to the gradient calculation, which is the only part of the process that is exposed to the private training data. In order to add noise to the gradient, it is convenient to bound its sensitivity via clipping to some L2 norm. In this example, clipping occurs in the \texttt{gradient\_sum} function before summation.

The original implementation of this algorithm~\cite{BST} was based on the advanced composition theorem. Advanced composition improves on sequential composition by providing much tighter privacy bounds over several iterations, but requires the analyst to fix the number of iterations up front, regardless of how
many iterations the gradient descent algorithm actually takes to converge to minimal error.

We present a modified version based on \emph{adaptive} Renyi differential privacy which provides not only a tighter analysis of the privacy leakage over several iterations, but also
allows the analyst to halt computation adaptively (conserving the remaining privacy budget) once a certain level of model accuracy has been reached, or loss has been minimized. We obscure the accuracy calculation
because it is a computation on the sensitive input training dataset in this case.

\begin{figure}

\begin{minted}{python}
def dp_gradient_descent(iterations, alpha, eps):
  eps_i = eps/iterations
  theta = np.zeros(X.shape[1])
  with dduo.RenyiFilter(alpha,eps_max):
    with dduo.RenyiOdometer((alpha, eps)) as odo:
      noisy_count = dduo.renyi_gauss(|$\alpha $|=alpha,
        |$\epsilon $|=eps,X.shape[0])
        priv_acc = 0
        acc_diff = 1
      while acc_diff > 0.05:
        grad_sum = gradient_sum(theta,
          X_train, y_train, sensitivity)
        noisy_grad_sum = dduo.gauss_vec(grad_sum,
          |$\alpha $|=alpha,|$\epsilon $|=eps_i)
        noisy_avg_grad = noisy_grad_sum/noisy_count
        theta = np.subtract(theta, noisy_avg_grad)
        priv_acc_curr = dduo.renyi_gauss(alpha,
          eps_acc, accuracy(theta))
        acc_diff =  priv_acc_curr - priv_acc
        priv_acc = priv_acc_curr
      print(odo)
    return theta

theta = dp_gradient_descent(iterations,
  |$\alpha $|=alpha, |$\epsilon $|=epsilon)
acc = dduo.renyi_gauss(alpha,
  eps_acc, accuracy(theta))
print(f"final accuracy: {acc}")
\end{minted}

\vspace{-1em}
\begin{minted}[frame=lines,bgcolor=mygray]{text}
Odometer_|$(\alpha ,\epsilon )$|(|$\{ \textit{data.csv} \mapsto  (10.0, 2.40)\} $|)
final accuracy: 0.753
\end{minted}
\vspace{-2ex}
\caption{Gradient Descent with NumPy}
\label{fig:gd}
\end{figure}

\section{Related Work}
\label{sec:related}

\paragraph{Dynamic Enforcement of Differential Privacy.}
The first approach for dynamic enforcement of differential privacy was \pinq \cite{mcsherry2009}.
Since then several works have been based on \pinq, such as Featherweight PINQ \cite{ebadi2015featherweight} which models \pinq formally and proves that any programs which use its simplified \pinq API are differentially private. ProPer \cite{ebadi2015} is a system (based on \pinq) designed to maintain a privacy budget for each individual in a database system, and operates by silently dropping records from queries when their privacy budget is exceeded. UniTrax \cite{munz2018} follows up on ProPer: this system allows per-user budgets but gets around the issue of silently dropping records by tracking queries against an abstract database as opposed to the actual database records.
These approaches are limited to an embedded DSL for expressing relational database queries, and do not support general purpose programming.

A number of programming frameworks for differential privacy have been
developed as libraries for existing programming languages. {\sc
  DPella}~\cite{lobo2020programming} is a Haskell library that
provides static bounds on the accuracy of differentially private
programs. Diffprivlib~\cite{holohan2019diffprivlib} (for Python) and
Google's library~\cite{wilson2020differentially} (for several
languages) provide differentially private algorithms, but do not track
sensitivity or privacy as these algorithms are composed.
{{\color{\colorMATH}\ensuremath{\mathit{\epsilon }}}}ktelo~\cite{zhang2018ektelo} executes programmer-specified
\emph{plans} that encode differentially private algorithms using
framework-supplied building blocks.

\paragraph{Dynamic Information Flow Control.}
Our approach to dynamic enforcement of differential privacy can be seen as similar
to work on dynamic information flow control (IFC) and taint analysis \cite{Austin2009EfficientPI}. The sensitivities that we attach to values are comparable to IFC labels. However, dynamic IFC typically allows the programmer to branch on sensitive information and handles implicit flows dynamically. \dduo prevents branching on sensitive information, similar to an approach taken for preventing side-channels in the "constant time" programming discipline for cryptographic code \cite{barthe2019}.
Another connection with this line of work is that our use of the logical relations proof technique for a differential privacy theorem is similar to the usage of this technique for noninterference theorems \cite{sabelfield1999,heintze1998,abadi1999,bowman2015,algehed2019,gregersen2021}.

\paragraph{Static Verification of Differential Privacy.}
The first approach for static verification of differential privacy was
\fuzz~\cite{reed2010distance}, which used \emph{linear types} to
analyze sensitivity. \dfuzz~\cite{gaboardi2013linear} adds dependent
types, and later work~\cite{de2019probabilistic, near2019duet} extends
the approach to {{\color{\colorMATH}\ensuremath{\mathit{(\epsilon , \delta )}}}}-differential privacy. \fuzzi~\cite{zhang2019fuzzi}
integrates a \fuzz-like type system with an expressive program logic.
\textit{Adaptive Fuzz}~\cite{Winograd-CortHR17} combines a \fuzz-style
\emph{static} type system for sensitivity analysis with a \emph{dynamic}
privacy analysis using privacy filters and odometers; \emph{Adaptive
  Fuzz} is most similar to our approach, but uses a static sensitivity
analysis. All of these approaches require additional type annotations.

A second category is based on \emph{approximate
  couplings}~\cite{DBLP:journals/lmcs/BartheEHSS19}. The
\apRHL~\cite{Barthe:POPL12,Barthe:TOPLAS:13},
\apRHLplus~\cite{Barthe:LICS16}, and \spanapRHL~\cite{Sato:LICS19}
relational logics are extremely expressive, but
less amenable to automation. Albarghouthi and
Hsu~\cite{DBLP:journals/pacmpl/AlbarghouthiH18} use an alternative
approach based on constraint solving to synthesize approximate
couplings.
A third approach is based on \emph{randomness alignments};
LightDP~\cite{zhang2017lightdp} and ShadowDP~\cite{wang2019proving}
take this approach. Randomness alignments are effective for verifying
low-level mechanisms like the sparse vector technique. The latter two
categories are generally restricted to first order imperative
programs.

\paragraph{Dynamic Testing for Differential Privacy.}
A recent line of work~\cite{bichsel2018dp, ding2018detecting,
  wang2020checkdp, wilson2020differentially} has resulted in
approaches for \emph{testing} differentially private programs. These
approaches generate a series of neighboring inputs, run the program
many times on the neighboring inputs, and raise an alarm if a
counterexample is found. These approaches do not require type
annotations, but do require running the program many times.





\section{Conclusion}
\label{sec:conclusion}
We have presented \dduo, a dynamic analysis that supports general-purpose differentially private programming with an emphasis on machine learning.
We have formalized the sensitivity analysis of \dduo and proven its soundness using a step-indexed logical relation. Our case studies demonstrate the utility of \dduo by enforcing adaptive variants of several differentially private state-of-the-art machine learning algorithms from the ground up, while integrating with some of Python's most popular libraries for data analysis.

\clearpage

\bibliographystyle{plain}
\bibliography{refs,gdp}

\appendices

\section{Lemmas,  Theorems \& Proofs}

\begin{lemma}[Plus Respects]\label{thm:plus-respects}
  If {{\color{\colorMATH}\ensuremath{\mathit{r_{1} \sim _{m}^{r} r_{2}}}}}
  Then {{\color{\colorMATH}\ensuremath{\mathit{(r_{1} + r_{3}) \sim _{m}^{r} (r_{2} + r_{3})}}}}.
\end{lemma}
\begin{proof}
  By unfolding definitions and simple arithmetic.
\end{proof}

\begin{lemma}[Times Respects]\label{thm:times-respects}
  If {{\color{\colorMATH}\ensuremath{\mathit{r_{1} \sim _{m}^{r} r_{2}}}}}
  then {{\color{\colorMATH}\ensuremath{\mathit{(r_{1} \times  r_{3}) \sim _{m}^{r_{3}r} (r_{2} \times  r_{3})}}}}.
\end{lemma}
\begin{proof}
  By unfolding definitions and simple arithmetic.
\end{proof}

\begin{lemma}[Triangle]\label{thm:triangle}
  If   {{\color{\colorMATH}\ensuremath{\mathit{r_{1} \sim ^{r_{A}}_{m_{A}} r_{2}}}}}
  and  {{\color{\colorMATH}\ensuremath{\mathit{r_{2} \sim ^{r_{B}}_{m_{B}} r_{3}}}}}
  then {{\color{\colorMATH}\ensuremath{\mathit{r_{1} \sim ^{r_{A}+r_{B}}_{m_{A}\sqcup m_{B}} r_{3}}}}}.
\end{lemma}
\begin{proof}
  By unfolding definitions, simple arithmetic, and the standard
  triangle inequality property for real numbers.
\end{proof}

\begin{lemma}[Real Metric Weakening]\label{thm:weakening}
  If {{\color{\colorMATH}\ensuremath{\mathit{r_{1} \sim ^{r}m r_{2}}}}},
  {{\color{\colorMATH}\ensuremath{\mathit{r \leq  r^{\prime}}}}}
  and  {{\color{\colorMATH}\ensuremath{\mathit{m \sqsubseteq  m^{\prime}}}}}
  then {{\color{\colorMATH}\ensuremath{\mathit{r_{1} \sim ^{r^{\prime}}_{m^{\prime}} r_{2}}}}}.
\end{lemma}
\begin{proof}
  By unfolding definitions and simple arithmetic.
\end{proof}

\begin{lemma}[Step-index Weakening]\label{thm:step-index-weakening}
  For {{\color{\colorMATH}\ensuremath{\mathit{n^{\prime} \leq  n}}}}:
  (1): If {{\color{\colorMATH}\ensuremath{\mathit{\rho _{1} \sim _{n}^{\Sigma } \rho _{2}}}}}             then {{\color{\colorMATH}\ensuremath{\mathit{\rho _{1} \sim _{n^{\prime}}^{\Sigma } \rho _{2}}}}};
  (2): If {{\color{\colorMATH}\ensuremath{\mathit{\sigma _{1} \sim _{n}^{\Sigma } \sigma _{2}}}}}             then {{\color{\colorMATH}\ensuremath{\mathit{\sigma _{1} \sim _{n^{\prime}}^{\Sigma } \sigma _{2}}}}}; and
  (3): If {{\color{\colorMATH}\ensuremath{\mathit{\rho _{1},\sigma _{1},e_{1} \sim _{n}^{\Sigma } \rho _{2},\sigma _{2},e_{2}}}}} then {{\color{\colorMATH}\ensuremath{\mathit{\rho _{1},\sigma _{1},e_{1} \sim _{n^{\prime}}^{\Sigma } \rho _{2},\sigma _{2},e_{2}}}}}.
\end{lemma}
\begin{proof}
  By mutual induction on {{\color{\colorMATH}\ensuremath{\mathit{n}}}} for all three properties, and
  additionally case analysis on {{\color{\colorMATH}\ensuremath{\mathit{e_{1}}}}} and {{\color{\colorMATH}\ensuremath{\mathit{e_{2}}}}} for property {\textit{(3)}}.
\end{proof}

\begin{theorem}[Metric Preservation]\label{thm:metric-preservation}\ \\
  \begin{itemize}[label={},leftmargin=0pt]\item  \begin{tabular}{r@{\hspace*{1.00em}}l@{\hspace*{1.00em}}c
     } If:      & {{\color{\colorMATH}\ensuremath{\mathit{\rho _{1} \sim _{n}^{\Sigma } \rho _{2}}}}}                & {\textit{(H1)}}
     \cr  And:     & {{\color{\colorMATH}\ensuremath{\mathit{\sigma _{1} \sim _{n}^{\Sigma } \sigma _{2}}}}}                & {\textit{(H2)}}
     \cr  Then:    & {{\color{\colorMATH}\ensuremath{\mathit{\rho _{1},\sigma _{1},e \sim _{n}^{\Sigma } \rho _{2},\sigma _{2},e}}}}      &
     \end{tabular}
  \item
  \item  That is, either {{\color{\colorMATH}\ensuremath{\mathit{n = 0}}}} or {{\color{\colorMATH}\ensuremath{\mathit{n = n^{\prime}+1}}}} and\ldots
  \item
  \item  \begin{tabular}{r@{\hspace*{1.00em}}l@{\hspace*{1.00em}}c
     } If:      & {{\color{\colorMATH}\ensuremath{\mathit{n_{1} \leq  n^{\prime}}}}}                    & {\textit{(H3)}}
     \cr  And:     & {{\color{\colorMATH}\ensuremath{\mathit{\rho _{1} \vdash  \mathrlap{\raisebox{-5pt}{\color{black!50}$\llcorner$}}\mathrlap{\raisebox{6pt}{\color{black!50}$\ulcorner$}}\hspace{2pt}\vphantom{\underset x{\overset XX}}\sigma _{1},e\hspace{2pt}\mathllap{\raisebox{-5pt}{\color{black!50}$\lrcorner$}}\mathllap{\raisebox{6pt}{\color{black!50}$\urcorner$}}\vphantom{\underset x{\overset XX}} \Downarrow _{n_{1}} \mathrlap{\raisebox{-5pt}{\color{black!50}$\llcorner$}}\mathrlap{\raisebox{6pt}{\color{black!50}$\ulcorner$}}\hspace{2pt}\vphantom{\underset x{\overset XX}}\sigma _{1}^{\prime},v_{1}\hspace{2pt}\mathllap{\raisebox{-5pt}{\color{black!50}$\lrcorner$}}\mathllap{\raisebox{6pt}{\color{black!50}$\urcorner$}}\vphantom{\underset x{\overset XX}}}}}} & {\textit{(H4)}}
     \cr  And:     & {{\color{\colorMATH}\ensuremath{\mathit{\rho _{2} \vdash  \mathrlap{\raisebox{-5pt}{\color{black!50}$\llcorner$}}\mathrlap{\raisebox{6pt}{\color{black!50}$\ulcorner$}}\hspace{2pt}\vphantom{\underset x{\overset XX}}\sigma _{2},e\hspace{2pt}\mathllap{\raisebox{-5pt}{\color{black!50}$\lrcorner$}}\mathllap{\raisebox{6pt}{\color{black!50}$\urcorner$}}\vphantom{\underset x{\overset XX}} \Downarrow _{n_{2}} \mathrlap{\raisebox{-5pt}{\color{black!50}$\llcorner$}}\mathrlap{\raisebox{6pt}{\color{black!50}$\ulcorner$}}\hspace{2pt}\vphantom{\underset x{\overset XX}}\sigma _{2}^{\prime},v_{2}\hspace{2pt}\mathllap{\raisebox{-5pt}{\color{black!50}$\lrcorner$}}\mathllap{\raisebox{6pt}{\color{black!50}$\urcorner$}}\vphantom{\underset x{\overset XX}}}}}} & {\textit{(H5)}}
     \cr  Then:    & {{\color{\colorMATH}\ensuremath{\mathit{n_{1} = n_{2}}}}}                    & {\textit{(C1)}}
     \cr  And:     & {{\color{\colorMATH}\ensuremath{\mathit{\sigma _{1}^{\prime} \sim ^{\Sigma }_{n^{\prime}-n_{1}} \sigma _{2}^{\prime}}}}}        & {\textit{(C2)}}
     \cr  And:     & {{\color{\colorMATH}\ensuremath{\mathit{v_{1} \sim ^{\Sigma }_{n^{\prime}-n_{1}} v_{2}}}}}          & {\textit{(C3)}}
     \end{tabular}
  \end{itemize}
\end{theorem}
\begin{proof}
  See detailed proof later in this section.
\end{proof}

\begin{corollary}[Sound Dynamic Analysis for Sensitivity]\ \\
   \begin{tabular}{r@{\hspace*{1.00em}}l@{\hspace*{1.00em}}l
   } If:   & {{\color{\colorMATH}\ensuremath{\mathit{n_{1} < n}}}}, {{\color{\colorMATH}\ensuremath{\mathit{n_{2} < n}}}} and {{\color{\colorMATH}\ensuremath{\mathit{n_{3} < n}}}}     & {\textit{(H1)}}
   \cr  And:  & {{\color{\colorMATH}\ensuremath{\mathit{\rho  \sim _{n}^{\{ o\mapsto \infty \} } \rho _{1}}}}}                        & {\textit{(H2)}}
   \cr  And:  & {{\color{\colorMATH}\ensuremath{\mathit{\rho  \vdash  \mathrlap{\raisebox{-5pt}{\color{black!50}$\llcorner$}}\mathrlap{\raisebox{6pt}{\color{black!50}$\ulcorner$}}\hspace{2pt}\vphantom{\underset x{\overset XX}}\varnothing ,e\hspace{2pt}\mathllap{\raisebox{-5pt}{\color{black!50}$\lrcorner$}}\mathllap{\raisebox{6pt}{\color{black!50}$\urcorner$}}\vphantom{\underset x{\overset XX}} \Downarrow _{n_{1}} \mathrlap{\raisebox{-5pt}{\color{black!50}$\llcorner$}}\mathrlap{\raisebox{6pt}{\color{black!50}$\ulcorner$}}\hspace{2pt}\vphantom{\underset x{\overset XX}}\sigma ,r@_{m}^{\Sigma }\hspace{2pt}\mathllap{\raisebox{-5pt}{\color{black!50}$\lrcorner$}}\mathllap{\raisebox{6pt}{\color{black!50}$\urcorner$}}\vphantom{\underset x{\overset XX}}}}}}        & {\textit{(H3)}}
   \cr  And:  & {{\color{\colorMATH}\ensuremath{\mathit{\rho _{1} \sim ^{\Sigma ^{\prime}}_{n} \rho _{2}}}}}                      & {\textit{(H4)}}
   \cr  And:  & {{\color{\colorMATH}\ensuremath{\mathit{\rho _{1} \vdash  \mathrlap{\raisebox{-5pt}{\color{black!50}$\llcorner$}}\mathrlap{\raisebox{6pt}{\color{black!50}$\ulcorner$}}\hspace{2pt}\vphantom{\underset x{\overset XX}}\varnothing ,e\hspace{2pt}\mathllap{\raisebox{-5pt}{\color{black!50}$\lrcorner$}}\mathllap{\raisebox{6pt}{\color{black!50}$\urcorner$}}\vphantom{\underset x{\overset XX}} \Downarrow _{n_{2}} \mathrlap{\raisebox{-5pt}{\color{black!50}$\llcorner$}}\mathrlap{\raisebox{6pt}{\color{black!50}$\ulcorner$}}\hspace{2pt}\vphantom{\underset x{\overset XX}}\sigma _{1},r_{1}@^{\Sigma _{1}}_{m_{1}}\hspace{2pt}\mathllap{\raisebox{-5pt}{\color{black!50}$\lrcorner$}}\mathllap{\raisebox{6pt}{\color{black!50}$\urcorner$}}\vphantom{\underset x{\overset XX}}}}}} & {\textit{(H5)}}
   \cr  And:  & {{\color{\colorMATH}\ensuremath{\mathit{\rho _{2} \vdash  \mathrlap{\raisebox{-5pt}{\color{black!50}$\llcorner$}}\mathrlap{\raisebox{6pt}{\color{black!50}$\ulcorner$}}\hspace{2pt}\vphantom{\underset x{\overset XX}}\varnothing ,e\hspace{2pt}\mathllap{\raisebox{-5pt}{\color{black!50}$\lrcorner$}}\mathllap{\raisebox{6pt}{\color{black!50}$\urcorner$}}\vphantom{\underset x{\overset XX}} \Downarrow _{n_{3}} \mathrlap{\raisebox{-5pt}{\color{black!50}$\llcorner$}}\mathrlap{\raisebox{6pt}{\color{black!50}$\ulcorner$}}\hspace{2pt}\vphantom{\underset x{\overset XX}}\sigma _{2},r_{2}@^{\Sigma _{2}}_{m_{2}}\hspace{2pt}\mathllap{\raisebox{-5pt}{\color{black!50}$\lrcorner$}}\mathllap{\raisebox{6pt}{\color{black!50}$\urcorner$}}\vphantom{\underset x{\overset XX}}}}}} & {\textit{(H6)}}
   \cr  Then: & {{\color{\colorMATH}\ensuremath{\mathit{r_{1} \sim ^{\Sigma ^{\prime}\mathord{\cdotp }\Sigma }_{m_{1}} r_{2}}}}}                 & {\textit{(C1)}}
   \end{tabular}
\end{corollary}
\begin{proof}\ \\
  \begin{itemize}[label={},leftmargin=0pt]\item  By \nameref{thm:metric-preservation},
     {\textit{(H2)}}, {\textit{(H1)}}, {\textit{(H3)}} and {\textit{(H5)}} we have {{\color{\colorMATH}\ensuremath{\mathit{\Sigma _{1} = \Sigma }}}} and {{\color{\colorMATH}\ensuremath{\mathit{m_{1} = m}}}}.
  \item  By \nameref{thm:metric-preservation},
     {\textit{(H4)}}, {\textit{(H1)}}, {\textit{(H5)}} and {\textit{(H6)}} we have proved the goal {\textit{(C1)}}.
  \end{itemize}
\end{proof}

\begin{proof}[Proof of \nameref{thm:metric-preservation}]\ \\
  By strong induction on {{\color{\colorMATH}\ensuremath{\mathit{n}}}} and case analysis on {{\color{\colorMATH}\ensuremath{\mathit{e}}}}:
  \begin{itemize}[label=\textbf{-},leftmargin=1em]\item  \begin{itemize}[label={},leftmargin=0pt]\item  {\textbf{Case}} {{\color{\colorMATH}\ensuremath{\mathit{n=0}}}}: 
        Trivial by definition.
     \end{itemize} 
  \item  \begin{itemize}[label={},leftmargin=0pt]\item  {\textbf{Case}} {{\color{\colorMATH}\ensuremath{\mathit{n=n+1}}}} {\textbf{and}} {{\color{\colorMATH}\ensuremath{\mathit{e=x}}}}: 
     \item  By inversion on {\textit{(H4)}} and {\textit{(H5)}} we have:
        {{\color{\colorMATH}\ensuremath{\mathit{n_{1} = n_{2} = 0}}}},
        {{\color{\colorMATH}\ensuremath{\mathit{\sigma _{1}^{\prime} = \sigma _{1}}}}},
        {{\color{\colorMATH}\ensuremath{\mathit{\sigma _{2}^{\prime} = \sigma _{2}}}}},
        {{\color{\colorMATH}\ensuremath{\mathit{v_{1} = \rho _{1}(x)}}}} and
        {{\color{\colorMATH}\ensuremath{\mathit{v_{2} = \rho _{2}(x)}}}}.
        To show:
        {\textit{(C1)}}: {{\color{\colorMATH}\ensuremath{\mathit{0 = 0}}}};
        {\textit{(C2)}}: {{\color{\colorMATH}\ensuremath{\mathit{\sigma _{1} \sim _{n}^{\Sigma } \sigma _{2}}}}}; and
        {\textit{(C3)}}: {{\color{\colorMATH}\ensuremath{\mathit{\rho _{1}(x) \sim _{n}^{\Sigma } \rho _{2}(x)}}}}.
        {\textit{(C1)}} is trivial.
        {\textit{(C2)}} is immediate by {\textit{(H2)}}.
        {\textit{(C3)}} is immediate by {\textit{(H1)}}.
     \end{itemize} 
  \item  \begin{itemize}[label={},leftmargin=0pt]\item  {\textbf{Case}} {{\color{\colorMATH}\ensuremath{\mathit{n=n+1}}}} {\textbf{and}} {{\color{\colorMATH}\ensuremath{\mathit{e=r}}}}: 
     \item  By inversion on {\textit{(H4)}} and {\textit{(H5)}} we have
        {{\color{\colorMATH}\ensuremath{\mathit{n_{1} = n_{2} = 0}}}},
        {{\color{\colorMATH}\ensuremath{\mathit{v_{1} = v_{2} = r@_{{{\color{\colorSYNTAX}\texttt{disc}}}}^{{\textbf{Z}}}}}}},
        {{\color{\colorMATH}\ensuremath{\mathit{\sigma _{1}^{\prime} = \sigma _{1}}}}} and
        {{\color{\colorMATH}\ensuremath{\mathit{\sigma _{2}^{\prime} = \sigma _{2}}}}}.
        To show:
        {\textit{(C1)}}: {{\color{\colorMATH}\ensuremath{\mathit{0 = 0}}}};
        {\textit{(C2)}}: {{\color{\colorMATH}\ensuremath{\mathit{\sigma _{1} \sim _{n}^{\Sigma } \sigma _{2}}}}}; and
        {\textit{(C3)}}: {{\color{\colorMATH}\ensuremath{\mathit{r \sim _{{{\color{\colorSYNTAX}\texttt{disc}}}}^{0} r}}}}.
        {\textit{(C1)}} is trivial.
        {\textit{(C2)}} is immediate by {\textit{(H2)}}.
        {\textit{(C3)}} is immediate by definition of relation {{\color{\colorMATH}\ensuremath{\mathit{\sim _{m}^{r}}}}}.
     \end{itemize} 
  \item  \begin{itemize}[label={},leftmargin=0pt]\item  {\textbf{Case}} {{\color{\colorMATH}\ensuremath{\mathit{n=n+1}}}} {\textbf{and}} {{\color{\colorMATH}\ensuremath{\mathit{e=e_{1}+e_{2}}}}}: 
     \item  By inversion on {\textit{(H4)}} and {\textit{(H5)}} we have:
     \item  \ {{\color{\colorMATH}\ensuremath{\mathit{\begin{array}{rclclclcl@{\hspace*{1.00em}}l
           } \rho _{1} &{}\vdash {}& \mathrlap{\raisebox{-5pt}{\color{black!50}$\llcorner$}}\mathrlap{\raisebox{6pt}{\color{black!50}$\ulcorner$}}\hspace{2pt}\vphantom{\underset x{\overset XX}}\sigma _{1} &{},{}&e_{1}\hspace{2pt}\mathllap{\raisebox{-5pt}{\color{black!50}$\lrcorner$}}\mathllap{\raisebox{6pt}{\color{black!50}$\urcorner$}}\vphantom{\underset x{\overset XX}} &{}\Downarrow _{n_{1 1}}{}& \mathrlap{\raisebox{-5pt}{\color{black!50}$\llcorner$}}\mathrlap{\raisebox{6pt}{\color{black!50}$\ulcorner$}}\hspace{2pt}\vphantom{\underset x{\overset XX}}\sigma _{1}^{\prime}&{},{}&r_{1 1}@^{\Sigma _{1 1}}_{m_{1 1}}\hspace{2pt}\mathllap{\raisebox{-5pt}{\color{black!50}$\lrcorner$}}\mathllap{\raisebox{6pt}{\color{black!50}$\urcorner$}}\vphantom{\underset x{\overset XX}} & {{\color{\colorTEXT}\textnormal{{\textit{(H4.1)}}}}}
           \cr  \rho _{1} &{}\vdash {}& \mathrlap{\raisebox{-5pt}{\color{black!50}$\llcorner$}}\mathrlap{\raisebox{6pt}{\color{black!50}$\ulcorner$}}\hspace{2pt}\vphantom{\underset x{\overset XX}}\sigma _{1}^{\prime}&{},{}&e_{2}\hspace{2pt}\mathllap{\raisebox{-5pt}{\color{black!50}$\lrcorner$}}\mathllap{\raisebox{6pt}{\color{black!50}$\urcorner$}}\vphantom{\underset x{\overset XX}} &{}\Downarrow _{n_{1 2}}{}& \mathrlap{\raisebox{-5pt}{\color{black!50}$\llcorner$}}\mathrlap{\raisebox{6pt}{\color{black!50}$\ulcorner$}}\hspace{2pt}\vphantom{\underset x{\overset XX}}\sigma _{1}^{\prime \prime}&{},{}&r_{1 2}@^{\Sigma _{1 2}}_{m_{1 2}}\hspace{2pt}\mathllap{\raisebox{-5pt}{\color{black!50}$\lrcorner$}}\mathllap{\raisebox{6pt}{\color{black!50}$\urcorner$}}\vphantom{\underset x{\overset XX}} & {{\color{\colorTEXT}\textnormal{{\textit{(H4.2)}}}}}
           \cr  \rho _{2} &{}\vdash {}& \mathrlap{\raisebox{-5pt}{\color{black!50}$\llcorner$}}\mathrlap{\raisebox{6pt}{\color{black!50}$\ulcorner$}}\hspace{2pt}\vphantom{\underset x{\overset XX}}\sigma _{2} &{},{}&e_{1}\hspace{2pt}\mathllap{\raisebox{-5pt}{\color{black!50}$\lrcorner$}}\mathllap{\raisebox{6pt}{\color{black!50}$\urcorner$}}\vphantom{\underset x{\overset XX}} &{}\Downarrow _{n_{2 1}}{}& \mathrlap{\raisebox{-5pt}{\color{black!50}$\llcorner$}}\mathrlap{\raisebox{6pt}{\color{black!50}$\ulcorner$}}\hspace{2pt}\vphantom{\underset x{\overset XX}}\sigma _{2}^{\prime}&{},{}&r_{2 1}@^{\Sigma _{2 1}}_{m_{2 1}}\hspace{2pt}\mathllap{\raisebox{-5pt}{\color{black!50}$\lrcorner$}}\mathllap{\raisebox{6pt}{\color{black!50}$\urcorner$}}\vphantom{\underset x{\overset XX}} & {{\color{\colorTEXT}\textnormal{{\textit{(H5.1)}}}}}
           \cr  \rho _{2} &{}\vdash {}& \mathrlap{\raisebox{-5pt}{\color{black!50}$\llcorner$}}\mathrlap{\raisebox{6pt}{\color{black!50}$\ulcorner$}}\hspace{2pt}\vphantom{\underset x{\overset XX}}\sigma _{2}^{\prime}&{},{}&e_{2}\hspace{2pt}\mathllap{\raisebox{-5pt}{\color{black!50}$\lrcorner$}}\mathllap{\raisebox{6pt}{\color{black!50}$\urcorner$}}\vphantom{\underset x{\overset XX}} &{}\Downarrow _{n_{2 2}}{}& \mathrlap{\raisebox{-5pt}{\color{black!50}$\llcorner$}}\mathrlap{\raisebox{6pt}{\color{black!50}$\ulcorner$}}\hspace{2pt}\vphantom{\underset x{\overset XX}}\sigma _{2}^{\prime \prime}&{},{}&r_{2 2}@^{\Sigma _{2 2}}_{m_{2 2}}\hspace{2pt}\mathllap{\raisebox{-5pt}{\color{black!50}$\lrcorner$}}\mathllap{\raisebox{6pt}{\color{black!50}$\urcorner$}}\vphantom{\underset x{\overset XX}} & {{\color{\colorTEXT}\textnormal{{\textit{(H5.2)}}}}}
           \end{array}}}}}
     \item  and we also have:
        {{\color{\colorMATH}\ensuremath{\mathit{n_{1} = n_{1 1} + n_{1 2}}}}},
        {{\color{\colorMATH}\ensuremath{\mathit{n_{2} = n_{2 1} + n_{2 2}}}}},
        {{\color{\colorMATH}\ensuremath{\mathit{\sigma _{1}^{\prime} = \sigma _{1}^{\prime \prime}}}}},
        {{\color{\colorMATH}\ensuremath{\mathit{\sigma _{2}^{\prime} = \sigma _{2}^{\prime \prime}}}}},
        {{\color{\colorMATH}\ensuremath{\mathit{v_{1} = (r_{1 1}+r_{1 2})@^{\Sigma _{1 1}+\Sigma _{1 2}}_{m_{1 1}\sqcup m_{1 2}}}}}} and
        {{\color{\colorMATH}\ensuremath{\mathit{v_{2} = (r_{2 1}+r_{2 2})@^{\Sigma _{2 1}+\Sigma _{2 2}}_{m_{2 1}\sqcup m_{2 2}}}}}}.
        By IH ({{\color{\colorMATH}\ensuremath{\mathit{n = n}}}} decreasing), {\textit{(H1)}}, {\textit{(H2)}}, {\textit{(H3)}}, {\textit{(H4.1)}} and {\textit{(H5.1)}} we have:
        {{\color{\colorMATH}\ensuremath{\mathit{n_{1 1} = n_{2 1}}}}}                                 {\textit{(IH.1.C1)}};
        {{\color{\colorMATH}\ensuremath{\mathit{\sigma _{1}^{\prime} \sim ^{\Sigma }_{n-n_{1 1}} \sigma _{2}^{\prime}}}}}                       {\textit{(IH.1.C2)}}; and
        {{\color{\colorMATH}\ensuremath{\mathit{r_{1 1}@^{\Sigma _{1 1}}_{m_{1 1}} \sim ^{\Sigma }_{n-n_{1 1}} r_{2 1}@^{\Sigma _{2 1}}_{m_{2 1}}}}}} {\textit{(IH.1.C3)}}.
        By unfolding the definition in {\textit{(IH.1.C3)}}, we have:
        {{\color{\colorMATH}\ensuremath{\mathit{\Sigma _{1 1} = \Sigma _{2 1}}}}}             {\textit{(IH.1.C3.1)}};
        {{\color{\colorMATH}\ensuremath{\mathit{m_{1 1} = m_{2 1}}}}}             {\textit{(IH.1.C3.2)}}; and
        {{\color{\colorMATH}\ensuremath{\mathit{r_{1 1} \sim ^{\Sigma \mathord{\cdotp }\Sigma _{1 1}}_{m_{1 1}} r_{2 1}}}}} {\textit{(IH.1.C3.3)}}.
        Note the following facts:
        {{\color{\colorMATH}\ensuremath{\mathit{n_{1 2} \leq  n-n_{1 1}}}}}    {\textit{(F1)}}; and
        {{\color{\colorMATH}\ensuremath{\mathit{\rho _{1} \sim ^{\Sigma }_{n-n_{1 1}}}}}} {\textit{(F2)}}.
        {\textit{(F1)}} follows from {\textit{(H3)}} and {{\color{\colorMATH}\ensuremath{\mathit{n_{1} = n_{1 1} + n_{1 2}}}}}.
        {\textit{(F2)}} follows from {\textit{(H1)}} and {\nameref{thm:step-index-weakening}.1}.
        By IH ({{\color{\colorMATH}\ensuremath{\mathit{n = n-n_{1 1}}}}} decreasing), {\textit{(F2)}}, {\textit{(IH.1.C2)}}, {\textit{(F1)}}, {\textit{(H4.2)}} and {\textit{(H5.2)}} we have:
        {{\color{\colorMATH}\ensuremath{\mathit{n_{1 2} = n_{2 2}}}}}                                   {\textit{(IH.2.C1)}};
        {{\color{\colorMATH}\ensuremath{\mathit{\sigma _{1}^{\prime \prime} \sim ^{\Sigma }_{n-n_{1 1}-n_{1 2}} \sigma _{2}^{\prime \prime}}}}}                     {\textit{(IH.2.C2)}}; and
        {{\color{\colorMATH}\ensuremath{\mathit{r_{1 2}@^{\Sigma _{1 2}}_{m_{1 2}} \sim ^{\Sigma }_{n_{1 1}-n_{1 2}} r_{2 2}@^{\Sigma _{2 2}}_{m_{2 2}}}}}} {\textit{(IH.2.C3)}}.
        By unfolding the definition in {\textit{(IH.2.C3)}}, we have:
        {{\color{\colorMATH}\ensuremath{\mathit{\Sigma _{1 2} = \Sigma _{2 2}}}}}             {\textit{(IH.2.C3.1)}};
        {{\color{\colorMATH}\ensuremath{\mathit{m_{1 2} = m_{2 2}}}}}             {\textit{(IH.2.C3.2)}}; and
        {{\color{\colorMATH}\ensuremath{\mathit{r_{1 2} \sim ^{\Sigma \mathord{\cdotp }\Sigma _{1 2}}_{m_{1 2}} r_{2 2}}}}} {\textit{(IH.2.C3.3)}}.
        To show:
        {\textit{(C1)}}: {{\color{\colorMATH}\ensuremath{\mathit{n_{1 1} + n_{1 2} = n_{2 1} + n_{2 2}}}}};
        {\textit{(C2)}}: {{\color{\colorMATH}\ensuremath{\mathit{\sigma _{1}^{\prime \prime} \sim ^{\Sigma }_{n-n_{1 1}-n_{1 2}} \sigma _{2}^{\prime \prime}}}}}; and
        {\textit{(C3)}}:
        {{\color{\colorMATH}\ensuremath{\mathit{(r_{1 1}+r_{1 2})@^{\Sigma _{1 1}+\Sigma _{1 2}}_{m_{1 1}\sqcup m_{1 2}}
         \sim ^{\Sigma }_{n-n_{1 1}-n_{1 2}}
         (r_{2 1}+r_{2 2})@^{\Sigma _{2 1}+\Sigma _{2 2}}_{m_{2 1}\sqcup m_{2 2}}}}}}.
        {\textit{(C1)}} is immediate from {\textit{(IH.1.C1)}} and {\textit{(IH.2.C1)}}.
        {\textit{(C2)}} is immediate from {\textit{(IH.2.C2)}}.
        To show {\textit{(C3)}} we must show:
        {\textit{(C3.1)}}: {{\color{\colorMATH}\ensuremath{\mathit{\Sigma _{1 1} + \Sigma _{1 2} = \Sigma _{2 1} + \Sigma _{2 2}}}}};
        {\textit{(C3.2)}}: {{\color{\colorMATH}\ensuremath{\mathit{m_{1 1} \sqcup  m_{1 2} = m_{2 1} \sqcup  m_{2 2}}}}}; and
        {\textit{(C3.3)}}: {{\color{\colorMATH}\ensuremath{\mathit{(r_{1 1}+r_{1 2}) \sim ^{\Sigma \mathord{\cdotp }(\Sigma _{1 1}+\Sigma _{1 2})}_{m_{1 1}\sqcup m_{1 2}} (r_{2 1} + r_{2 2})}}}}.
        {\textit{(C3.1)}} is immediate from {\textit{(IH.1.C3.1)}} and {\textit{(IH.2.C3.1)}}.
        {\textit{(C3.2)}} is immediate from {\textit{(IH.1.C3.2)}} and {\textit{(IH.2.C3.2)}}.
        {\textit{(C3.3)}} holds as follows:
        By \nameref{thm:plus-respects}, {\textit{(IH.1.C3.3)}} and {\textit{(IH.2.C3.3)}}:
        {{\color{\colorMATH}\ensuremath{\mathit{(r_{1 1}+r_{1 2}) \sim ^{\Sigma \mathord{\cdotp }\Sigma _{1 1}}_{m_{1 1}} (r_{2 1} + r_{1 2}) \sim ^{\Sigma \mathord{\cdotp }\Sigma _{1 2}}_{m_{1 2}} (r_{2 1} + r_{2 2})}}}}.
        By \nameref{thm:triangle}:
        {{\color{\colorMATH}\ensuremath{\mathit{(r_{1 1}+r_{1 2}) \sim ^{\Sigma \mathord{\cdotp }\Sigma _{1 1} + \Sigma \mathord{\cdotp }\Sigma _{1 2}}_{m_{1 1}\sqcup m_{1 2}} (r_{2 1} + r_{2 2})}}}}.
        By basic algebra:
        {{\color{\colorMATH}\ensuremath{\mathit{(r_{1 1}+r_{1 2})
         \sim ^{\Sigma \mathord{\cdotp }(\Sigma _{1 1}+\Sigma _{1 2})}_{m_{1 1}\sqcup m_{1 2}}
         (r_{2 1} + r_{2 2})}}}}
        We have shown the goal.
     \end{itemize} 
  \item  \begin{itemize}[label={},leftmargin=0pt]\item  {\textbf{Case}} {{\color{\colorMATH}\ensuremath{\mathit{n=n+1}}}} {\textbf{and}} {{\color{\colorMATH}\ensuremath{\mathit{e=e_{1}\ltimes e_{2}}}}}: 
     \item  By inversion on {\textit{(H4)}} and {\textit{(H5)}} we have:
     \item  \ {{\color{\colorMATH}\ensuremath{\mathit{\begin{array}{rclclclcl@{\hspace*{1.00em}}l
           } \rho _{1} &{}\vdash {}& \mathrlap{\raisebox{-5pt}{\color{black!50}$\llcorner$}}\mathrlap{\raisebox{6pt}{\color{black!50}$\ulcorner$}}\hspace{2pt}\vphantom{\underset x{\overset XX}}\sigma _{1} &{},{}&e_{1}\hspace{2pt}\mathllap{\raisebox{-5pt}{\color{black!50}$\lrcorner$}}\mathllap{\raisebox{6pt}{\color{black!50}$\urcorner$}}\vphantom{\underset x{\overset XX}} &{}\Downarrow _{n_{1 1}}{}& \mathrlap{\raisebox{-5pt}{\color{black!50}$\llcorner$}}\mathrlap{\raisebox{6pt}{\color{black!50}$\ulcorner$}}\hspace{2pt}\vphantom{\underset x{\overset XX}}\sigma _{1}^{\prime}&{},{}&r_{1 1}@^{{\textbf{Z}}}_{m_{1 1}}\hspace{2pt}\mathllap{\raisebox{-5pt}{\color{black!50}$\lrcorner$}}\mathllap{\raisebox{6pt}{\color{black!50}$\urcorner$}}\vphantom{\underset x{\overset XX}} & {{\color{\colorTEXT}\textnormal{{\textit{(H4.1)}}}}}
           \cr  \rho _{1} &{}\vdash {}& \mathrlap{\raisebox{-5pt}{\color{black!50}$\llcorner$}}\mathrlap{\raisebox{6pt}{\color{black!50}$\ulcorner$}}\hspace{2pt}\vphantom{\underset x{\overset XX}}\sigma _{1}^{\prime}&{},{}&e_{2}\hspace{2pt}\mathllap{\raisebox{-5pt}{\color{black!50}$\lrcorner$}}\mathllap{\raisebox{6pt}{\color{black!50}$\urcorner$}}\vphantom{\underset x{\overset XX}} &{}\Downarrow _{n_{1 2}}{}& \mathrlap{\raisebox{-5pt}{\color{black!50}$\llcorner$}}\mathrlap{\raisebox{6pt}{\color{black!50}$\ulcorner$}}\hspace{2pt}\vphantom{\underset x{\overset XX}}\sigma _{1}^{\prime \prime}&{},{}&r_{1 2}@^{{\textbf{Z}}}_{m_{1 2}}\hspace{2pt}\mathllap{\raisebox{-5pt}{\color{black!50}$\lrcorner$}}\mathllap{\raisebox{6pt}{\color{black!50}$\urcorner$}}\vphantom{\underset x{\overset XX}} & {{\color{\colorTEXT}\textnormal{{\textit{(H4.2)}}}}}
           \cr  \rho _{2} &{}\vdash {}& \mathrlap{\raisebox{-5pt}{\color{black!50}$\llcorner$}}\mathrlap{\raisebox{6pt}{\color{black!50}$\ulcorner$}}\hspace{2pt}\vphantom{\underset x{\overset XX}}\sigma _{2} &{},{}&e_{1}\hspace{2pt}\mathllap{\raisebox{-5pt}{\color{black!50}$\lrcorner$}}\mathllap{\raisebox{6pt}{\color{black!50}$\urcorner$}}\vphantom{\underset x{\overset XX}} &{}\Downarrow _{n_{2 1}}{}& \mathrlap{\raisebox{-5pt}{\color{black!50}$\llcorner$}}\mathrlap{\raisebox{6pt}{\color{black!50}$\ulcorner$}}\hspace{2pt}\vphantom{\underset x{\overset XX}}\sigma _{2}^{\prime}&{},{}&r_{2 1}@^{\Sigma _{1}}_{m_{2 1}}\hspace{2pt}\mathllap{\raisebox{-5pt}{\color{black!50}$\lrcorner$}}\mathllap{\raisebox{6pt}{\color{black!50}$\urcorner$}}\vphantom{\underset x{\overset XX}} & {{\color{\colorTEXT}\textnormal{{\textit{(H5.1)}}}}}
           \cr  \rho _{2} &{}\vdash {}& \mathrlap{\raisebox{-5pt}{\color{black!50}$\llcorner$}}\mathrlap{\raisebox{6pt}{\color{black!50}$\ulcorner$}}\hspace{2pt}\vphantom{\underset x{\overset XX}}\sigma _{2}^{\prime}&{},{}&e_{2}\hspace{2pt}\mathllap{\raisebox{-5pt}{\color{black!50}$\lrcorner$}}\mathllap{\raisebox{6pt}{\color{black!50}$\urcorner$}}\vphantom{\underset x{\overset XX}} &{}\Downarrow _{n_{2 2}}{}& \mathrlap{\raisebox{-5pt}{\color{black!50}$\llcorner$}}\mathrlap{\raisebox{6pt}{\color{black!50}$\ulcorner$}}\hspace{2pt}\vphantom{\underset x{\overset XX}}\sigma _{2}^{\prime \prime}&{},{}&r_{2 2}@^{\Sigma _{2}}_{m_{2 2}}\hspace{2pt}\mathllap{\raisebox{-5pt}{\color{black!50}$\lrcorner$}}\mathllap{\raisebox{6pt}{\color{black!50}$\urcorner$}}\vphantom{\underset x{\overset XX}} & {{\color{\colorTEXT}\textnormal{{\textit{(H5.2)}}}}}
           \end{array}}}}}
     \item  and we also have:
        {{\color{\colorMATH}\ensuremath{\mathit{n_{1} = n_{1 1} + n_{1 2}}}}},
        {{\color{\colorMATH}\ensuremath{\mathit{n_{2} = n_{2 1} + n_{2 2}}}}},
        {{\color{\colorMATH}\ensuremath{\mathit{\sigma _{1}^{\prime} = \sigma _{1}^{\prime \prime}}}}},
        {{\color{\colorMATH}\ensuremath{\mathit{\sigma _{2}^{\prime} = \sigma _{2}^{\prime \prime}}}}},
        {{\color{\colorMATH}\ensuremath{\mathit{v_{1} = (r_{1 1}\times r_{1 2})@^{r_{1 1}\Sigma _{1 2}}_{m_{1 2}}}}}} and
        {{\color{\colorMATH}\ensuremath{\mathit{v_{2} = (r_{2 1}\times r_{2 2})@^{r_{2 1}\Sigma _{2 2}}_{m_{2 2}}}}}}.
        By IH ({{\color{\colorMATH}\ensuremath{\mathit{n = n}}}} decreasing), {\textit{(H1)}}, {\textit{(H2)}}, {\textit{(H3)}}, {\textit{(H4.1)}}
        and {\textit{(H5.1)}} we have:
        {{\color{\colorMATH}\ensuremath{\mathit{n_{1 1} = n_{2 1}}}}}                                     {\textit{(IH.1.C1)}};
        {{\color{\colorMATH}\ensuremath{\mathit{\sigma _{1}^{\prime} \sim ^{\Sigma }_{n-n_{1 1}} \sigma _{2}^{\prime}}}}}                           {\textit{(IH.1.C2)}}; and
        {{\color{\colorMATH}\ensuremath{\mathit{r_{1 1}@^{\Sigma _{1 1}}_{m_{1 1}} \sim ^{\Sigma }_{n-n_{1 1}-n_{1 2}} r_{2 1}@^{\Sigma _{2 1}}_{m_{2 1}}}}}} {\textit{(IH.1.C3)}}.
        By unfolding the definition in {\textit{(IH.1.C3)}}, we have:
        {{\color{\colorMATH}\ensuremath{\mathit{\Sigma _{1 1} = \Sigma _{2 1}}}}}       {\textit{(IH.1.C3.1)}};
        {{\color{\colorMATH}\ensuremath{\mathit{m_{1 1} = m_{2 1}}}}}       {\textit{(IH.1.C3.2)}}; and
        {{\color{\colorMATH}\ensuremath{\mathit{r_{1 1} \sim ^{0}_{m_{1 1}} r_{2 1}}}}} {\textit{(IH.1.C3.3)}}.
        As a consequence of {\textit{(IH.1.C3)}}, we have: {{\color{\colorMATH}\ensuremath{\mathit{r_{1 1} = r_{2 1}}}}}.
        Note the following facts:
        {{\color{\colorMATH}\ensuremath{\mathit{n_{1 2} \leq  n-n_{1 1}}}}}    {\textit{(F1)}}; and
        {{\color{\colorMATH}\ensuremath{\mathit{\rho _{1} \sim ^{\Sigma }_{n-n_{1 1}}}}}} {\textit{(F2)}}.
        {\textit{(F1)}} follows from {\textit{(H3)}} and {{\color{\colorMATH}\ensuremath{\mathit{n_{1} = n_{1 1} + n_{1 2}}}}}.
        {\textit{(F2)}} follows from {\textit{(H1)}} and {\nameref{thm:step-index-weakening}.1}.
        By IH ({{\color{\colorMATH}\ensuremath{\mathit{n = n-n_{1 1}}}}} decreasing), {\textit{(F2)}}, {\textit{(IH.1.C2)}}, {\textit{(F1)}}, {\textit{(H4.2)}} and {\textit{(H5.2)}} we have:
        {{\color{\colorMATH}\ensuremath{\mathit{n_{1 2} = n_{2 2}}}}}                                   {\textit{(IH.2.C1)}};
        {{\color{\colorMATH}\ensuremath{\mathit{\sigma _{1}^{\prime \prime} \sim ^{\Sigma }_{n-n_{1 1}-n_{1 2}} \sigma _{2}^{\prime \prime}}}}}                     {\textit{(IH.2.C2)}}; and
        {{\color{\colorMATH}\ensuremath{\mathit{r_{1 2}@^{\Sigma _{1 2}}_{m_{1 2}} \sim ^{\Sigma }_{n_{1 1}-n_{1 2}} r_{2 2}@^{\Sigma _{2 2}}_{m_{2 2}}}}}} {\textit{(IH.2.C3)}}.
        By unfolding the definition in {\textit{(IH.2.C3)}}, we have:
        {{\color{\colorMATH}\ensuremath{\mathit{\Sigma _{1 2} = \Sigma _{2 2}}}}}             {\textit{(IH.2.C3.1)}};
        {{\color{\colorMATH}\ensuremath{\mathit{m_{1 2} = m_{2 2}}}}}             {\textit{(IH.2.C3.2)}}; and
        {{\color{\colorMATH}\ensuremath{\mathit{r_{1 2} \sim ^{\Sigma \mathord{\cdotp }\Sigma _{1 2}}_{m_{1 2}} r_{2 2}}}}} {\textit{(IH.2.C3.3)}}.
        To show:
        {\textit{(C1)}}: {{\color{\colorMATH}\ensuremath{\mathit{n_{1 1} + n_{1 2} = n_{2 1} + n_{2 2}}}}};
        {\textit{(C2)}}: {{\color{\colorMATH}\ensuremath{\mathit{\sigma _{1}^{\prime \prime} \sim ^{\Sigma }_{n-n_{1 1}-n_{1 2}} \sigma _{2}^{\prime \prime}}}}}; and
        {\textit{(C3)}}:
        {{\color{\colorMATH}\ensuremath{\mathit{(r_{1 1}\times r_{1 2})@^{r_{1 1}\Sigma _{1 2}}_{m_{1 2}}
         \sim ^{\Sigma }_{n-n_{1 1}-n_{1 2}}
         (r_{2 1}\times r_{2 2})@^{(r_{2 1}\Sigma _{2 2})}_{m_{2 2}}}}}}.
        {\textit{(C1)}} is immediate from {\textit{(IH.1.C1)}} and {\textit{(IH.2.C1)}}.
        {\textit{(C2)}} is immediate from {\textit{(IH.2.C2)}}.
        To show {\textit{(C3)}} we must show:
        {\textit{(C3.1)}}: {{\color{\colorMATH}\ensuremath{\mathit{r_{1 1}\Sigma _{1 2} = r_{2 1}\Sigma _{2 2}}}}};
        {\textit{(C3.2)}}: {{\color{\colorMATH}\ensuremath{\mathit{m_{1 2} = m_{2 2}}}}}; and
        {\textit{(C3.3)}}: {{\color{\colorMATH}\ensuremath{\mathit{(r_{1 1}\times r_{1 2}) \sim ^{\Sigma \mathord{\cdotp }r_{1 1}\Sigma _{1 2}}_{m_{1 2}} (r_{2 1}\times r_{2 2})}}}}.
        {\textit{(C3)}} holds as follows:
        {\textit{(C3.1)}} is immediate from {{\color{\colorMATH}\ensuremath{\mathit{r_{1 1} = r_{2 1}}}}} and {\textit{(IH.2.C3.1)}}.
        {\textit{(C3.2)}} is immediate from {\textit{(IH.2.C3.2)}}.
        {\textit{(C3.3)}} holds as follows:
        By \nameref{thm:times-respects}, {{\color{\colorMATH}\ensuremath{\mathit{r_{1 1} = r_{2 1}}}}} and {\textit{(IH.2.C3.3)}}:
        {{\color{\colorMATH}\ensuremath{\mathit{(r_{1 1}\times r_{1 2}) = (r_{2 1}\times r_{1 2}) \sim ^{r_{1 1}(\Sigma \mathord{\cdotp }\Sigma _{1 1})}_{m_{2 1}} (r_{2 1} \times  r_{2 2})}}}}.
        By basic algebra:
        {{\color{\colorMATH}\ensuremath{\mathit{(r_{1 1}\times r_{1 2}) = (r_{2 1}\times r_{1 2}) \sim ^{\Sigma \mathord{\cdotp }r_{1 1}\Sigma _{1 1}}_{m_{2 1}} (r_{2 1} \times  r_{2 2})}}}}.
        We have shown the goal.
     \end{itemize} 
  \item  \begin{itemize}[label={},leftmargin=0pt]\item  {\textbf{Case}} {{\color{\colorMATH}\ensuremath{\mathit{n=n+1}}}} {\textbf{and}} {{\color{\colorMATH}\ensuremath{\mathit{e=e_{1}\rtimes e_{2}}}}}: 
        {\textit{Analogous to previous case.}}
     \end{itemize} 
  \item  \begin{itemize}[label={},leftmargin=0pt]\item  {\textbf{Case}} {{\color{\colorMATH}\ensuremath{\mathit{n=n+1}}}} {\textbf{and}} {{\color{\colorMATH}\ensuremath{\mathit{e = {{\color{\colorSYNTAX}\texttt{if0}}}(e_{1})\{ e_{2}\} \{ e_{3}\} }}}}: 
     \item  By inversion on {\textit{(H4)}} and {\textit{(H5)}} we have 4 subcases, each which induce:
     \item  \ {{\color{\colorMATH}\ensuremath{\mathit{\begin{array}{rclclclcl@{\hspace*{1.00em}}l
           } \rho _{1} &{}\vdash {}& \mathrlap{\raisebox{-5pt}{\color{black!50}$\llcorner$}}\mathrlap{\raisebox{6pt}{\color{black!50}$\ulcorner$}}\hspace{2pt}\vphantom{\underset x{\overset XX}}\sigma _{1} &{},{}&e_{1}\hspace{2pt}\mathllap{\raisebox{-5pt}{\color{black!50}$\lrcorner$}}\mathllap{\raisebox{6pt}{\color{black!50}$\urcorner$}}\vphantom{\underset x{\overset XX}} &{}\Downarrow _{n_{1 1}}{}& \mathrlap{\raisebox{-5pt}{\color{black!50}$\llcorner$}}\mathrlap{\raisebox{6pt}{\color{black!50}$\ulcorner$}}\hspace{2pt}\vphantom{\underset x{\overset XX}}\sigma _{1}^{\prime}&{},{}&r_{1}@^{{\textbf{Z}}}_{m_{1}}\hspace{2pt}\mathllap{\raisebox{-5pt}{\color{black!50}$\lrcorner$}}\mathllap{\raisebox{6pt}{\color{black!50}$\urcorner$}}\vphantom{\underset x{\overset XX}} & {{\color{\colorTEXT}\textnormal{{\textit{(H4.1)}}}}}
           \cr  \rho _{2} &{}\vdash {}& \mathrlap{\raisebox{-5pt}{\color{black!50}$\llcorner$}}\mathrlap{\raisebox{6pt}{\color{black!50}$\ulcorner$}}\hspace{2pt}\vphantom{\underset x{\overset XX}}\sigma _{2} &{},{}&e_{1}\hspace{2pt}\mathllap{\raisebox{-5pt}{\color{black!50}$\lrcorner$}}\mathllap{\raisebox{6pt}{\color{black!50}$\urcorner$}}\vphantom{\underset x{\overset XX}} &{}\Downarrow _{n_{2 1}}{}& \mathrlap{\raisebox{-5pt}{\color{black!50}$\llcorner$}}\mathrlap{\raisebox{6pt}{\color{black!50}$\ulcorner$}}\hspace{2pt}\vphantom{\underset x{\overset XX}}\sigma _{2}^{\prime}&{},{}&r_{2}@^{{\textbf{Z}}}_{m_{2}}\hspace{2pt}\mathllap{\raisebox{-5pt}{\color{black!50}$\lrcorner$}}\mathllap{\raisebox{6pt}{\color{black!50}$\urcorner$}}\vphantom{\underset x{\overset XX}} & {{\color{\colorTEXT}\textnormal{{\textit{(H5.1)}}}}}
           \end{array}}}}}
     \item  By IH ({{\color{\colorMATH}\ensuremath{\mathit{n = n}}}} decreasing), {\textit{(H1)}}, {\textit{(H2)}}, {\textit{(H3)}}, {\textit{(H4.1)}} and {\textit{(H5.1)}} we have:
        {{\color{\colorMATH}\ensuremath{\mathit{n_{1 1} = n_{2 1}}}}}           {\textit{(IH.1.C1)}};
        {{\color{\colorMATH}\ensuremath{\mathit{\sigma _{1}^{\prime} \sim ^{\Sigma }_{n-n_{1 1}} \sigma _{2}^{\prime}}}}} {\textit{(IH.1.C2)}}; and
        {{\color{\colorMATH}\ensuremath{\mathit{r_{1} \sim ^{0}_{m_{1}\sqcap m_{2}} r_{2}}}}}     {\textit{(IH.1.C3)}}.
        As a consequence of {\textit{(IH.1.C3)}}, we have: {{\color{\colorMATH}\ensuremath{\mathit{r_{1} = r_{2}}}}}.
        The 4 subcases initially are for:
        {\textit{(1)}}: {{\color{\colorMATH}\ensuremath{\mathit{r_{1} = 0}}}} and {{\color{\colorMATH}\ensuremath{\mathit{r_{2} = 0}}}};
        {\textit{(2)}}: {{\color{\colorMATH}\ensuremath{\mathit{r_{1} = 0}}}} and {{\color{\colorMATH}\ensuremath{\mathit{r_{2} \neq  0}}}};
        {\textit{(3)}}: {{\color{\colorMATH}\ensuremath{\mathit{r_{1} \neq  0}}}} and {{\color{\colorMATH}\ensuremath{\mathit{r_{2} = 0}}}}; and
        {\textit{(4)}}: {{\color{\colorMATH}\ensuremath{\mathit{r_{1} \neq  0}}}} and {{\color{\colorMATH}\ensuremath{\mathit{r_{2} \neq  0}}}}.
        However 2 are absurd given {{\color{\colorMATH}\ensuremath{\mathit{r_{1} = r_{2}}}}}, so these 4 subcases
        collapse to 2:
     \item  \begin{itemize}[label=\textbf{-},leftmargin=1em]\item  \begin{itemize}[label={},leftmargin=0pt]\item  {\textbf{Subcase}} {{\color{\colorMATH}\ensuremath{\mathit{r_{1} = r_{2} = 0}}}}:
           \item  From prior inversion on {\textit{(H4)}} and {\textit{(H5)}} and fact {{\color{\colorMATH}\ensuremath{\mathit{r_{1} = r_{2} = 0}}}}, we also have:
           \item  \ {{\color{\colorMATH}\ensuremath{\mathit{\begin{array}{rclclclcl@{\hspace*{1.00em}}l
                 } \rho _{1} &{}\vdash {}& \mathrlap{\raisebox{-5pt}{\color{black!50}$\llcorner$}}\mathrlap{\raisebox{6pt}{\color{black!50}$\ulcorner$}}\hspace{2pt}\vphantom{\underset x{\overset XX}}\sigma _{1}^{\prime}&{},{}&e_{2}\hspace{2pt}\mathllap{\raisebox{-5pt}{\color{black!50}$\lrcorner$}}\mathllap{\raisebox{6pt}{\color{black!50}$\urcorner$}}\vphantom{\underset x{\overset XX}} &{}\Downarrow _{n_{1 2}}{}& \mathrlap{\raisebox{-5pt}{\color{black!50}$\llcorner$}}\mathrlap{\raisebox{6pt}{\color{black!50}$\ulcorner$}}\hspace{2pt}\vphantom{\underset x{\overset XX}}\sigma _{1}^{\prime \prime}&{},{}&v_{1}\hspace{2pt}\mathllap{\raisebox{-5pt}{\color{black!50}$\lrcorner$}}\mathllap{\raisebox{6pt}{\color{black!50}$\urcorner$}}\vphantom{\underset x{\overset XX}}             & {{\color{\colorTEXT}\textnormal{{\textit{(H4.2)}}}}}
                 \cr  \rho _{2} &{}\vdash {}& \mathrlap{\raisebox{-5pt}{\color{black!50}$\llcorner$}}\mathrlap{\raisebox{6pt}{\color{black!50}$\ulcorner$}}\hspace{2pt}\vphantom{\underset x{\overset XX}}\sigma _{2}^{\prime}&{},{}&e_{2}\hspace{2pt}\mathllap{\raisebox{-5pt}{\color{black!50}$\lrcorner$}}\mathllap{\raisebox{6pt}{\color{black!50}$\urcorner$}}\vphantom{\underset x{\overset XX}} &{}\Downarrow _{n_{2 2}}{}& \mathrlap{\raisebox{-5pt}{\color{black!50}$\llcorner$}}\mathrlap{\raisebox{6pt}{\color{black!50}$\ulcorner$}}\hspace{2pt}\vphantom{\underset x{\overset XX}}\sigma _{2}^{\prime \prime}&{},{}&v_{2}\hspace{2pt}\mathllap{\raisebox{-5pt}{\color{black!50}$\lrcorner$}}\mathllap{\raisebox{6pt}{\color{black!50}$\urcorner$}}\vphantom{\underset x{\overset XX}}             & {{\color{\colorTEXT}\textnormal{{\textit{(H5.2)}}}}}
                 \end{array}}}}}
           \item  and we also have:
              {{\color{\colorMATH}\ensuremath{\mathit{n_{1} = n_{1 1} + n_{1 2}}}}},
              {{\color{\colorMATH}\ensuremath{\mathit{n_{2} = n_{2 1} + n_{2 2}}}}},
              {{\color{\colorMATH}\ensuremath{\mathit{\sigma _{1}^{\prime} = \sigma _{1}^{\prime \prime}}}}},
              {{\color{\colorMATH}\ensuremath{\mathit{\sigma _{2}^{\prime} = \sigma _{2}^{\prime \prime}}}}},
              {{\color{\colorMATH}\ensuremath{\mathit{v_{1} = v_{1}}}}} and
              {{\color{\colorMATH}\ensuremath{\mathit{v_{2} = v_{2}}}}}.
              We continue reasoning in a generic way outside this subcase\ldots
           \end{itemize}
        \item  \begin{itemize}[label={},leftmargin=0pt]\item  {\textbf{Subcase}} {{\color{\colorMATH}\ensuremath{\mathit{r_{1} \neq  0}}}} and {{\color{\colorMATH}\ensuremath{\mathit{r_{2} \neq  0}}}}:
           \item  From prior inversion on {\textit{(H4)}} and {\textit{(H5)}} and fact {{\color{\colorMATH}\ensuremath{\mathit{r_{1} = r_{2} \neq  0}}}}, we also have:
           \item  \ {{\color{\colorMATH}\ensuremath{\mathit{\begin{array}{rclclclcl@{\hspace*{1.00em}}l
                 } \rho _{1} &{}\vdash {}& \mathrlap{\raisebox{-5pt}{\color{black!50}$\llcorner$}}\mathrlap{\raisebox{6pt}{\color{black!50}$\ulcorner$}}\hspace{2pt}\vphantom{\underset x{\overset XX}}\sigma _{1}^{\prime}&{},{}&e_{3}\hspace{2pt}\mathllap{\raisebox{-5pt}{\color{black!50}$\lrcorner$}}\mathllap{\raisebox{6pt}{\color{black!50}$\urcorner$}}\vphantom{\underset x{\overset XX}} &{}\Downarrow _{n_{1 2}}{}& \mathrlap{\raisebox{-5pt}{\color{black!50}$\llcorner$}}\mathrlap{\raisebox{6pt}{\color{black!50}$\ulcorner$}}\hspace{2pt}\vphantom{\underset x{\overset XX}}\sigma _{1}^{\prime \prime}&{},{}&v_{1}\hspace{2pt}\mathllap{\raisebox{-5pt}{\color{black!50}$\lrcorner$}}\mathllap{\raisebox{6pt}{\color{black!50}$\urcorner$}}\vphantom{\underset x{\overset XX}}             & {{\color{\colorTEXT}\textnormal{{\textit{(H4.2)}}}}}
                 \cr  \rho _{2} &{}\vdash {}& \mathrlap{\raisebox{-5pt}{\color{black!50}$\llcorner$}}\mathrlap{\raisebox{6pt}{\color{black!50}$\ulcorner$}}\hspace{2pt}\vphantom{\underset x{\overset XX}}\sigma _{2}^{\prime}&{},{}&e_{3}\hspace{2pt}\mathllap{\raisebox{-5pt}{\color{black!50}$\lrcorner$}}\mathllap{\raisebox{6pt}{\color{black!50}$\urcorner$}}\vphantom{\underset x{\overset XX}} &{}\Downarrow _{n_{2 2}}{}& \mathrlap{\raisebox{-5pt}{\color{black!50}$\llcorner$}}\mathrlap{\raisebox{6pt}{\color{black!50}$\ulcorner$}}\hspace{2pt}\vphantom{\underset x{\overset XX}}\sigma _{2}^{\prime \prime}&{},{}&v_{2}\hspace{2pt}\mathllap{\raisebox{-5pt}{\color{black!50}$\lrcorner$}}\mathllap{\raisebox{6pt}{\color{black!50}$\urcorner$}}\vphantom{\underset x{\overset XX}}             & {{\color{\colorTEXT}\textnormal{{\textit{(H5.2)}}}}}
                 \end{array}}}}}
           \item  and we also have:
              {{\color{\colorMATH}\ensuremath{\mathit{n_{1} = n_{1 1} + n_{1 2}}}}}
              {{\color{\colorMATH}\ensuremath{\mathit{n_{2} = n_{2 1} + n_{2 2}}}}}
              {{\color{\colorMATH}\ensuremath{\mathit{\sigma _{1}^{\prime} = \sigma _{1}^{\prime \prime}}}}}
              {{\color{\colorMATH}\ensuremath{\mathit{\sigma _{2}^{\prime} = \sigma _{2}^{\prime \prime}}}}}
              {{\color{\colorMATH}\ensuremath{\mathit{v_{1} = v_{1}}}}}
              {{\color{\colorMATH}\ensuremath{\mathit{v_{2} = v_{2}}}}}
              We continue reasoning in a generic way outside this subcase\ldots
           \end{itemize}
        \end{itemize}
     \item  Note the following facts:
        {{\color{\colorMATH}\ensuremath{\mathit{n_{1 2} \leq  n-n_{1 1}}}}}    {\textit{(F1)}}; and
        {{\color{\colorMATH}\ensuremath{\mathit{\rho _{1} \sim ^{\Sigma }_{n-n_{1 1}}}}}} {\textit{(F2)}}.
        {\textit{(F1)}} follows from {\textit{(H3)}} and {{\color{\colorMATH}\ensuremath{\mathit{n_{1} = n_{1 1} + n_{1 2}}}}}.
        {\textit{(F2)}} follows from {\textit{(H1)}} and {\nameref{thm:step-index-weakening}.1}.
        By IH ({{\color{\colorMATH}\ensuremath{\mathit{n = n-n_{1 1}}}}} decreasing), {\textit{(F2)}}, {\textit{(IH.1.C2)}}, {\textit{(F1)}}, {\textit{(H4.2)}} and {\textit{(H5.2)}} we have:
        {{\color{\colorMATH}\ensuremath{\mathit{n_{2 1} = n_{2 2}}}}}               {\textit{(IH.2.C1)}};
        {{\color{\colorMATH}\ensuremath{\mathit{\sigma _{1}^{\prime \prime} \sim ^{\Sigma }_{n-n_{1 1}-n_{2 1}} \sigma _{2}^{\prime \prime}}}}} {\textit{(IH.2.C2)}}; and
        {{\color{\colorMATH}\ensuremath{\mathit{v_{1} \sim ^{\Sigma }_{n-n_{1 1}-n_{2 1}} v_{2}}}}}   {\textit{(IH.2.C3)}}.
        To show:
        {\textit{(C1)}}: {{\color{\colorMATH}\ensuremath{\mathit{n_{1 1} + n_{1 2} = n_{2 1} + n_{2 2}}}}};
        {\textit{(C2)}}: {{\color{\colorMATH}\ensuremath{\mathit{\sigma _{1}^{\prime \prime} \sim ^{\Sigma }_{n-n_{1 1}-n_{1 2}} \sigma _{2}^{\prime \prime}}}}}; and
        {\textit{(C3)}}: {{\color{\colorMATH}\ensuremath{\mathit{v_{1} \sim ^{\Sigma }_{n-n_{1 1}-n_{1 2}} v_{2}}}}}.
        {\textit{(C1)}} is immediate from {\textit{(IH.1.C1)}} and {\textit{(IH.2.C1)}}.
        {\textit{(C2)}} is immediate from {\textit{(IH.2.C2)}}.
        {\textit{(C3)}} is immediate from {\textit{(IH.2.C3)}}
     \end{itemize} 
  \item  \begin{itemize}[label={},leftmargin=0pt]\item  {\textbf{Case}} {{\color{\colorMATH}\ensuremath{\mathit{n=n+1}}}} {\textbf{and}} {{\color{\colorMATH}\ensuremath{\mathit{e=\langle e_{1},e_{2}\rangle }}}}: 
     \item  By inversion on {\textit{(H4)}} and {\textit{(H5)}} we have:
     \item  \ {{\color{\colorMATH}\ensuremath{\mathit{\begin{array}{rclclclcl@{\hspace*{1.00em}}l
           } \rho _{1} &{}\vdash {}& \mathrlap{\raisebox{-5pt}{\color{black!50}$\llcorner$}}\mathrlap{\raisebox{6pt}{\color{black!50}$\ulcorner$}}\hspace{2pt}\vphantom{\underset x{\overset XX}}\sigma _{1} &{},{}&e_{1}\hspace{2pt}\mathllap{\raisebox{-5pt}{\color{black!50}$\lrcorner$}}\mathllap{\raisebox{6pt}{\color{black!50}$\urcorner$}}\vphantom{\underset x{\overset XX}} &{}\Downarrow _{n_{1 1}}{}& \mathrlap{\raisebox{-5pt}{\color{black!50}$\llcorner$}}\mathrlap{\raisebox{6pt}{\color{black!50}$\ulcorner$}}\hspace{2pt}\vphantom{\underset x{\overset XX}}\sigma _{1}^{\prime}&{},{}&v_{1 1}\hspace{2pt}\mathllap{\raisebox{-5pt}{\color{black!50}$\lrcorner$}}\mathllap{\raisebox{6pt}{\color{black!50}$\urcorner$}}\vphantom{\underset x{\overset XX}} & {{\color{\colorTEXT}\textnormal{{\textit{(H4.1)}}}}}
           \cr  \rho _{1} &{}\vdash {}& \mathrlap{\raisebox{-5pt}{\color{black!50}$\llcorner$}}\mathrlap{\raisebox{6pt}{\color{black!50}$\ulcorner$}}\hspace{2pt}\vphantom{\underset x{\overset XX}}\sigma _{1}^{\prime}&{},{}&e_{2}\hspace{2pt}\mathllap{\raisebox{-5pt}{\color{black!50}$\lrcorner$}}\mathllap{\raisebox{6pt}{\color{black!50}$\urcorner$}}\vphantom{\underset x{\overset XX}} &{}\Downarrow _{n_{1 2}}{}& \mathrlap{\raisebox{-5pt}{\color{black!50}$\llcorner$}}\mathrlap{\raisebox{6pt}{\color{black!50}$\ulcorner$}}\hspace{2pt}\vphantom{\underset x{\overset XX}}\sigma _{1}^{\prime \prime}&{},{}&v_{1 2}\hspace{2pt}\mathllap{\raisebox{-5pt}{\color{black!50}$\lrcorner$}}\mathllap{\raisebox{6pt}{\color{black!50}$\urcorner$}}\vphantom{\underset x{\overset XX}} & {{\color{\colorTEXT}\textnormal{{\textit{(H4.2)}}}}}
           \cr  \rho _{2} &{}\vdash {}& \mathrlap{\raisebox{-5pt}{\color{black!50}$\llcorner$}}\mathrlap{\raisebox{6pt}{\color{black!50}$\ulcorner$}}\hspace{2pt}\vphantom{\underset x{\overset XX}}\sigma _{2} &{},{}&e_{1}\hspace{2pt}\mathllap{\raisebox{-5pt}{\color{black!50}$\lrcorner$}}\mathllap{\raisebox{6pt}{\color{black!50}$\urcorner$}}\vphantom{\underset x{\overset XX}} &{}\Downarrow _{n_{2 1}}{}& \mathrlap{\raisebox{-5pt}{\color{black!50}$\llcorner$}}\mathrlap{\raisebox{6pt}{\color{black!50}$\ulcorner$}}\hspace{2pt}\vphantom{\underset x{\overset XX}}\sigma _{2}^{\prime}&{},{}&v_{2 1}\hspace{2pt}\mathllap{\raisebox{-5pt}{\color{black!50}$\lrcorner$}}\mathllap{\raisebox{6pt}{\color{black!50}$\urcorner$}}\vphantom{\underset x{\overset XX}} & {{\color{\colorTEXT}\textnormal{{\textit{(H5.1)}}}}}
           \cr  \rho _{2} &{}\vdash {}& \mathrlap{\raisebox{-5pt}{\color{black!50}$\llcorner$}}\mathrlap{\raisebox{6pt}{\color{black!50}$\ulcorner$}}\hspace{2pt}\vphantom{\underset x{\overset XX}}\sigma _{2}^{\prime}&{},{}&e_{2}\hspace{2pt}\mathllap{\raisebox{-5pt}{\color{black!50}$\lrcorner$}}\mathllap{\raisebox{6pt}{\color{black!50}$\urcorner$}}\vphantom{\underset x{\overset XX}} &{}\Downarrow _{n_{2 2}}{}& \mathrlap{\raisebox{-5pt}{\color{black!50}$\llcorner$}}\mathrlap{\raisebox{6pt}{\color{black!50}$\ulcorner$}}\hspace{2pt}\vphantom{\underset x{\overset XX}}\sigma _{2}^{\prime \prime}&{},{}&v_{2 2}\hspace{2pt}\mathllap{\raisebox{-5pt}{\color{black!50}$\lrcorner$}}\mathllap{\raisebox{6pt}{\color{black!50}$\urcorner$}}\vphantom{\underset x{\overset XX}} & {{\color{\colorTEXT}\textnormal{{\textit{(H5.2)}}}}}
           \end{array}}}}}
     \item  and we also have:
        {{\color{\colorMATH}\ensuremath{\mathit{n_{1} = n_{1 1} + n_{1 2}}}}},
        {{\color{\colorMATH}\ensuremath{\mathit{n_{2} = n_{2 1} + n_{2 2}}}}},
        {{\color{\colorMATH}\ensuremath{\mathit{\sigma _{1}^{\prime} = \sigma _{1}^{\prime \prime}}}}},
        {{\color{\colorMATH}\ensuremath{\mathit{\sigma _{2}^{\prime} = \sigma _{2}^{\prime \prime}}}}},
        {{\color{\colorMATH}\ensuremath{\mathit{v_{1} = \langle v_{1 1},v_{1 2}\rangle }}}} and
        {{\color{\colorMATH}\ensuremath{\mathit{v_{2} = \langle v_{2 1},v_{2 2}\rangle }}}}.
        By IH ({{\color{\colorMATH}\ensuremath{\mathit{n = n}}}} decreasing), {\textit{(H1)}}, {\textit{(H2)}}, {\textit{(H3)}}, {\textit{(H4.1)}} and {\textit{(H5.1)}} we have:
        {{\color{\colorMATH}\ensuremath{\mathit{n_{1 1} = n_{2 1}}}}}            {\textit{(IH.1.C1)}};
        {{\color{\colorMATH}\ensuremath{\mathit{\sigma _{1}^{\prime} \sim ^{\Sigma }_{n-n_{1 1}} \sigma _{2}^{\prime}}}}}  {\textit{(IH.1.C2)}}; and
        {{\color{\colorMATH}\ensuremath{\mathit{v_{1 1} \sim ^{\Sigma }_{n-n_{1 1}} v_{2 1}}}}}  {\textit{(IH.1.C3)}}.
        Note the following facts:
        {{\color{\colorMATH}\ensuremath{\mathit{n_{1 2} \leq  n-n_{1 1}}}}}    {\textit{(F1)}}; and
        {{\color{\colorMATH}\ensuremath{\mathit{\rho _{1} \sim ^{\Sigma }_{n-n_{1 1}}}}}} {\textit{(F2)}}.
        {\textit{(F1)}} follows from {\textit{(H3)}} and {{\color{\colorMATH}\ensuremath{\mathit{n_{1} = n_{1 1} + n_{1 2}}}}}.
        {\textit{(F2)}} follows from {\textit{(H1)}} and {\nameref{thm:step-index-weakening}.1}.
        By IH ({{\color{\colorMATH}\ensuremath{\mathit{n = n-n_{1 1}}}}} decreasing), {\textit{(F2)}}, {\textit{(IH.1.C2)}}, {\textit{(F1)}}, {\textit{(H4.2)}} and {\textit{(H5.2)}} we have:
        {{\color{\colorMATH}\ensuremath{\mathit{n_{1 2} = n_{2 2}}}}}               {\textit{(IH.2.C1)}};
        {{\color{\colorMATH}\ensuremath{\mathit{\sigma _{1}^{\prime \prime} \sim ^{\Sigma }_{n-n_{1 1}-n_{1 2}} \sigma _{2}^{\prime \prime}}}}} {\textit{(IH.2.C2)}}; and
        {{\color{\colorMATH}\ensuremath{\mathit{v_{1 2} \sim ^{\Sigma }_{n-n_{1 1}-n_{1 2}} v_{2 2}}}}} {\textit{(IH.2.C3)}}.
        To show:
        {\textit{(C1)}}: {{\color{\colorMATH}\ensuremath{\mathit{n_{1 1} + n_{1 2} = n_{2 1} + n_{2 2}}}}};
        {\textit{(C2)}}: {{\color{\colorMATH}\ensuremath{\mathit{\sigma _{1}^{\prime \prime} \sim ^{\Sigma }_{n-n_{1 1}-n_{1 2}} \sigma _{2}^{\prime \prime}}}}}; and
        {\textit{(C3)}}: {{\color{\colorMATH}\ensuremath{\mathit{\langle v_{1 1},v_{1 2}\rangle  \sim ^{\Sigma }_{n-n_{1 1}-n_{1 2}} \langle v_{2 1},v_{2 2}\rangle }}}}.
        {\textit{(C1)}} is immediate from {\textit{(IH.1.C1)}} and {\textit{(IH.2.C1)}}.
        {\textit{(C2)}} is immediate from {\textit{(IH.2.C2)}}.
        {\textit{(C3)}} is immediate from {\textit{(IH.1.C3)}} and {\textit{(IH.2.C3)}}.
     \end{itemize} 
  \item  \begin{itemize}[label={},leftmargin=0pt]\item  {\textbf{Case}} {{\color{\colorMATH}\ensuremath{\mathit{n=n+1}}}} {\textbf{and}} {{\color{\colorMATH}\ensuremath{\mathit{e=\pi _{i}(e)}}}}: 
     \item  By inversion on {\textit{(H4)}} and {\textit{(H5)}} we have:
     \item  \ {{\color{\colorMATH}\ensuremath{\mathit{\begin{array}{rclclclcl@{\hspace*{1.00em}}l
           } \rho _{1} &{}\vdash {}& \mathrlap{\raisebox{-5pt}{\color{black!50}$\llcorner$}}\mathrlap{\raisebox{6pt}{\color{black!50}$\ulcorner$}}\hspace{2pt}\vphantom{\underset x{\overset XX}}\sigma _{1} &{},{}&e\hspace{2pt}\mathllap{\raisebox{-5pt}{\color{black!50}$\lrcorner$}}\mathllap{\raisebox{6pt}{\color{black!50}$\urcorner$}}\vphantom{\underset x{\overset XX}} &{}\Downarrow _{n_{1}}{}& \mathrlap{\raisebox{-5pt}{\color{black!50}$\llcorner$}}\mathrlap{\raisebox{6pt}{\color{black!50}$\ulcorner$}}\hspace{2pt}\vphantom{\underset x{\overset XX}}\sigma _{1}^{\prime}&{},{}&\langle v_{1 1},v_{1 2}\rangle \hspace{2pt}\mathllap{\raisebox{-5pt}{\color{black!50}$\lrcorner$}}\mathllap{\raisebox{6pt}{\color{black!50}$\urcorner$}}\vphantom{\underset x{\overset XX}} & {{\color{\colorTEXT}\textnormal{{\textit{(H4.1)}}}}}
           \cr  \rho _{2} &{}\vdash {}& \mathrlap{\raisebox{-5pt}{\color{black!50}$\llcorner$}}\mathrlap{\raisebox{6pt}{\color{black!50}$\ulcorner$}}\hspace{2pt}\vphantom{\underset x{\overset XX}}\sigma _{2} &{},{}&e\hspace{2pt}\mathllap{\raisebox{-5pt}{\color{black!50}$\lrcorner$}}\mathllap{\raisebox{6pt}{\color{black!50}$\urcorner$}}\vphantom{\underset x{\overset XX}} &{}\Downarrow _{n_{2}}{}& \mathrlap{\raisebox{-5pt}{\color{black!50}$\llcorner$}}\mathrlap{\raisebox{6pt}{\color{black!50}$\ulcorner$}}\hspace{2pt}\vphantom{\underset x{\overset XX}}\sigma _{2}^{\prime}&{},{}&\langle v_{2 1},v_{2 2}\rangle \hspace{2pt}\mathllap{\raisebox{-5pt}{\color{black!50}$\lrcorner$}}\mathllap{\raisebox{6pt}{\color{black!50}$\urcorner$}}\vphantom{\underset x{\overset XX}} & {{\color{\colorTEXT}\textnormal{{\textit{(H5.1)}}}}}
           \end{array}}}}}
     \item  and we also have:
        {{\color{\colorMATH}\ensuremath{\mathit{n_{1} = n_{1}}}}},
        {{\color{\colorMATH}\ensuremath{\mathit{n_{2} = n_{2}}}}},
        {{\color{\colorMATH}\ensuremath{\mathit{\sigma _{1}^{\prime} = \sigma _{1}^{\prime}}}}},
        {{\color{\colorMATH}\ensuremath{\mathit{\sigma _{2}^{\prime} = \sigma _{2}^{\prime}}}}},
        {{\color{\colorMATH}\ensuremath{\mathit{v_{1} = v_{1 i}}}}} and
        {{\color{\colorMATH}\ensuremath{\mathit{v_{2} = v_{2 i}}}}}.
        By IH ({{\color{\colorMATH}\ensuremath{\mathit{n = n}}}} decreasing), {\textit{(H1)}}, {\textit{(H2)}}, {\textit{(H3)}}, {\textit{(H4.1)}} and {\textit{(H5.1)}} we have:
        {{\color{\colorMATH}\ensuremath{\mathit{n_{1} = n_{2}}}}}                         {\textit{(IH.1.C1)}};
        {{\color{\colorMATH}\ensuremath{\mathit{\sigma _{1}^{\prime} \sim ^{\Sigma }_{n-n_{1}} \sigma _{2}^{\prime}}}}}              {\textit{(IH.1.C2)}}; and
        {{\color{\colorMATH}\ensuremath{\mathit{\langle v_{1 1},v_{1 2}\rangle  \sim ^{\Sigma }_{n-n_{1}} \langle v_{2 1},v_{2 2}\rangle }}}}  {\textit{(IH.1.C3)}}.
        To show:
        {\textit{(C1)}}: {{\color{\colorMATH}\ensuremath{\mathit{n_{1} = n_{2}}}}};
        {\textit{(C2)}}: {{\color{\colorMATH}\ensuremath{\mathit{\sigma _{1}^{\prime} \sim ^{\Sigma }_{n-n_{1}} \sigma _{2}^{\prime}}}}}; and
        {\textit{(C3)}}: {{\color{\colorMATH}\ensuremath{\mathit{v_{1 n^{\prime}} \sim ^{\Sigma }_{n-n_{1}} v_{2 n^{\prime}}}}}}.
        {\textit{(C1)}} is immediate from {\textit{(IH.1.C1)}}.
        {\textit{(C2)}} is immediate from {\textit{(IH.1.C2)}}.
        {\textit{(C3)}} is immediate from {\textit{(IH.1.C3)}}.
     \end{itemize} 
  \item  \begin{itemize}[label={},leftmargin=0pt]\item  {\textbf{Case}} {{\color{\colorMATH}\ensuremath{\mathit{n=n+1}}}} {\textbf{and}} {{\color{\colorMATH}\ensuremath{\mathit{e=\lambda x.\hspace*{0.33em}e}}}}: 
        By inversion on {\textit{(H4)}} and {\textit{(H5)}} we have:
        {{\color{\colorMATH}\ensuremath{\mathit{n_{1} = 0}}}},
        {{\color{\colorMATH}\ensuremath{\mathit{n_{2} = 0}}}},
        {{\color{\colorMATH}\ensuremath{\mathit{\sigma _{1}^{\prime} = \sigma _{1}}}}},
        {{\color{\colorMATH}\ensuremath{\mathit{\sigma _{2}^{\prime} = \sigma _{2}}}}},
        {{\color{\colorMATH}\ensuremath{\mathit{v_{1} = \langle \lambda x.\hspace*{0.33em}e\mathrel{|}\rho _{1}\rangle }}}} and
        {{\color{\colorMATH}\ensuremath{\mathit{v_{2} = \langle \lambda x.\hspace*{0.33em}e\mathrel{|}\rho _{2}\rangle }}}}.
        To show:
        {\textit{(C1)}}: {{\color{\colorMATH}\ensuremath{\mathit{0 = 0}}}};
        {\textit{(C2)}}: {{\color{\colorMATH}\ensuremath{\mathit{\sigma _{1} \sim _{n}^{\Sigma } \sigma _{2}}}}}; and
        {\textit{(C3)}}: {{\color{\colorMATH}\ensuremath{\mathit{\langle \lambda x.\hspace*{0.33em}e\mathrel{|}\rho _{1}\rangle  \sim _{n}^{\Sigma } \langle \lambda x.\hspace*{0.33em}e\mathrel{|}\rho _{2}\rangle }}}}.
        {\textit{(C1)}} is trivial.
        {\textit{(C2)}} is immediate from {\textit{(H2)}}.
        {\textit{(C3)}} holds as follows:
        Unfolding the definition, we must show:
        {\textit{(C3)}}: {{\color{\colorMATH}\ensuremath{\mathit{\begin{array}[t]{lrl
                 } &{} {}&  \forall  n^{\prime}\leq n,v_{1},v_{2},\sigma _{1}^{\prime},\sigma _{2}^{\prime}.\hspace*{0.33em} \sigma _{1}^{\prime} \sim _{n^{\prime}}^{\Sigma } \sigma _{2}^{\prime} \hspace*{0.33em}\wedge \hspace*{0.33em} v_{1} \sim _{n^{\prime}}^{\Sigma } v_{2}
                 \cr  &{}\Rightarrow {}&  \sigma _{1}^{\prime},\{ x\mapsto v_{1}\} \uplus \rho _{1},e \sim _{n^{\prime}}^{\Sigma } \sigma _{2}^{\prime},\{ x\mapsto v_{2}\} \uplus \rho _{2},e
                 \end{array}}}}}.
        To show {\textit{(C3)}}, we assume:
        {{\color{\colorMATH}\ensuremath{\mathit{\sigma _{1}^{\prime} \sim _{n^{\prime}}^{\Sigma } \sigma _{2}^{\prime}}}}}  {\textit{(C3.H1)}}; and
        {{\color{\colorMATH}\ensuremath{\mathit{v_{1} \sim _{n^{\prime}}^{\Sigma } v_{2}}}}}    {\textit{(C3.H2)}}.
        And we must show:
        {\textit{(C3.1)}}: {{\color{\colorMATH}\ensuremath{\mathit{\sigma _{1}^{\prime},\{ x\mapsto v_{1}\} \uplus \rho _{1},e \sim _{n^{\prime}}^{\Sigma } \sigma _{2}^{\prime},\{ x\mapsto v_{2}\} \uplus \rho _{2},e}}}}.
        Note the following facts:
        {{\color{\colorMATH}\ensuremath{\mathit{\rho _{1} \sim _{n^{\prime}}^{\Sigma } \rho _{1}}}}}               {\textit{(F1)}}; and
        {{\color{\colorMATH}\ensuremath{\mathit{\{ x\mapsto v_{1}\} \uplus \rho _{1} \sim _{n^{\prime}}^{\Sigma } \{ x\mapsto v_{2}\} \uplus \rho _{2}}}}} {\textit{(F2)}}.
        {\textit{(F1)}} holds from {\textit{H1}} and {\nameref{thm:step-index-weakening}.1}.
        {\textit{(F2)}} holds from {\textit{(F1)}}, {\textit{(C3.H2)}} and the definition of {{\color{\colorMATH}\ensuremath{\mathit{\rho  \sim _{n^{\prime}}^{\Sigma } \rho }}}}.
        {\textit{(C3.1)}} then holds by IH ({{\color{\colorMATH}\ensuremath{\mathit{n = n^{\prime}}}}} decreasing), {{\color{\colorMATH}\ensuremath{\mathit{F2}}}} and {{\color{\colorMATH}\ensuremath{\mathit{C3.H1}}}}.
     \end{itemize} 
  \item  \begin{itemize}[label={},leftmargin=0pt]\item  {\textbf{Case}} {{\color{\colorMATH}\ensuremath{\mathit{n=n+1}}}} {\textbf{and}} {{\color{\colorMATH}\ensuremath{\mathit{e=e_{1}(e_{2})}}}}: 
     \item  By inversion on {\textit{(H4)}} and {\textit{(H5)}} we have:
     \item  \ {{\color{\colorMATH}\ensuremath{\mathit{\begin{array}{rclclclcl@{\hspace*{1.00em}}l
           } \rho _{1}         &{}\vdash {}& \mathrlap{\raisebox{-5pt}{\color{black!50}$\llcorner$}}\mathrlap{\raisebox{6pt}{\color{black!50}$\ulcorner$}}\hspace{2pt}\vphantom{\underset x{\overset XX}}\sigma _{1} &{},{}&e_{1}\hspace{2pt}\mathllap{\raisebox{-5pt}{\color{black!50}$\lrcorner$}}\mathllap{\raisebox{6pt}{\color{black!50}$\urcorner$}}\vphantom{\underset x{\overset XX}}  &{}\Downarrow _{n_{1 1}}{}& \mathrlap{\raisebox{-5pt}{\color{black!50}$\llcorner$}}\mathrlap{\raisebox{6pt}{\color{black!50}$\ulcorner$}}\hspace{2pt}\vphantom{\underset x{\overset XX}}\sigma _{1}^{\prime}&{},{}&\langle \lambda x.\hspace*{0.33em}e_{1}^{\prime}\mathrel{|}\rho _{1}^{\prime}\rangle \hspace{2pt}\mathllap{\raisebox{-5pt}{\color{black!50}$\lrcorner$}}\mathllap{\raisebox{6pt}{\color{black!50}$\urcorner$}}\vphantom{\underset x{\overset XX}} & {{\color{\colorTEXT}\textnormal{{\textit{(H4.1)}}}}}
           \cr  \rho _{1}         &{}\vdash {}& \mathrlap{\raisebox{-5pt}{\color{black!50}$\llcorner$}}\mathrlap{\raisebox{6pt}{\color{black!50}$\ulcorner$}}\hspace{2pt}\vphantom{\underset x{\overset XX}}\sigma _{1}^{\prime}&{},{}&e_{2}\hspace{2pt}\mathllap{\raisebox{-5pt}{\color{black!50}$\lrcorner$}}\mathllap{\raisebox{6pt}{\color{black!50}$\urcorner$}}\vphantom{\underset x{\overset XX}}  &{}\Downarrow _{n_{1 2}}{}& \mathrlap{\raisebox{-5pt}{\color{black!50}$\llcorner$}}\mathrlap{\raisebox{6pt}{\color{black!50}$\ulcorner$}}\hspace{2pt}\vphantom{\underset x{\overset XX}}\sigma _{1}^{\prime \prime}&{},{}&v_{1}\hspace{2pt}\mathllap{\raisebox{-5pt}{\color{black!50}$\lrcorner$}}\mathllap{\raisebox{6pt}{\color{black!50}$\urcorner$}}\vphantom{\underset x{\overset XX}}           & {{\color{\colorTEXT}\textnormal{{\textit{(H4.2)}}}}}
           \cr  \{ x\mapsto v_{1}\} \uplus \rho _{1}^{\prime} &{}\vdash {}& \mathrlap{\raisebox{-5pt}{\color{black!50}$\llcorner$}}\mathrlap{\raisebox{6pt}{\color{black!50}$\ulcorner$}}\hspace{2pt}\vphantom{\underset x{\overset XX}}\sigma _{2}^{\prime \prime}&{},{}&e_{1}^{\prime}\hspace{2pt}\mathllap{\raisebox{-5pt}{\color{black!50}$\lrcorner$}}\mathllap{\raisebox{6pt}{\color{black!50}$\urcorner$}}\vphantom{\underset x{\overset XX}} &{}\Downarrow _{n_{1 3}}{}& \mathrlap{\raisebox{-5pt}{\color{black!50}$\llcorner$}}\mathrlap{\raisebox{6pt}{\color{black!50}$\ulcorner$}}\hspace{2pt}\vphantom{\underset x{\overset XX}}\sigma _{1}^{\prime \prime \prime}&{},{}&v_{1}^{\prime}\hspace{2pt}\mathllap{\raisebox{-5pt}{\color{black!50}$\lrcorner$}}\mathllap{\raisebox{6pt}{\color{black!50}$\urcorner$}}\vphantom{\underset x{\overset XX}}          & {{\color{\colorTEXT}\textnormal{{\textit{(H4.3)}}}}}
           \cr  \rho _{2}         &{}\vdash {}& \mathrlap{\raisebox{-5pt}{\color{black!50}$\llcorner$}}\mathrlap{\raisebox{6pt}{\color{black!50}$\ulcorner$}}\hspace{2pt}\vphantom{\underset x{\overset XX}}\sigma _{2} &{},{}&e_{1}\hspace{2pt}\mathllap{\raisebox{-5pt}{\color{black!50}$\lrcorner$}}\mathllap{\raisebox{6pt}{\color{black!50}$\urcorner$}}\vphantom{\underset x{\overset XX}}  &{}\Downarrow _{n_{2 1}}{}& \mathrlap{\raisebox{-5pt}{\color{black!50}$\llcorner$}}\mathrlap{\raisebox{6pt}{\color{black!50}$\ulcorner$}}\hspace{2pt}\vphantom{\underset x{\overset XX}}\sigma _{2}^{\prime}&{},{}&\langle \lambda x.\hspace*{0.33em}e_{2}^{\prime}\mathrel{|}\rho _{2}^{\prime}\rangle \hspace{2pt}\mathllap{\raisebox{-5pt}{\color{black!50}$\lrcorner$}}\mathllap{\raisebox{6pt}{\color{black!50}$\urcorner$}}\vphantom{\underset x{\overset XX}} & {{\color{\colorTEXT}\textnormal{{\textit{(H5.1)}}}}}
           \cr  \rho _{2}         &{}\vdash {}& \mathrlap{\raisebox{-5pt}{\color{black!50}$\llcorner$}}\mathrlap{\raisebox{6pt}{\color{black!50}$\ulcorner$}}\hspace{2pt}\vphantom{\underset x{\overset XX}}\sigma _{2}^{\prime}&{},{}&e_{2}\hspace{2pt}\mathllap{\raisebox{-5pt}{\color{black!50}$\lrcorner$}}\mathllap{\raisebox{6pt}{\color{black!50}$\urcorner$}}\vphantom{\underset x{\overset XX}}  &{}\Downarrow _{n_{2 2}}{}& \mathrlap{\raisebox{-5pt}{\color{black!50}$\llcorner$}}\mathrlap{\raisebox{6pt}{\color{black!50}$\ulcorner$}}\hspace{2pt}\vphantom{\underset x{\overset XX}}\sigma _{2}^{\prime \prime}&{},{}&v_{2}\hspace{2pt}\mathllap{\raisebox{-5pt}{\color{black!50}$\lrcorner$}}\mathllap{\raisebox{6pt}{\color{black!50}$\urcorner$}}\vphantom{\underset x{\overset XX}}           & {{\color{\colorTEXT}\textnormal{{\textit{(H5.2)}}}}}
           \cr  \{ x\mapsto v_{2}\} \uplus \rho _{2}^{\prime} &{}\vdash {}& \mathrlap{\raisebox{-5pt}{\color{black!50}$\llcorner$}}\mathrlap{\raisebox{6pt}{\color{black!50}$\ulcorner$}}\hspace{2pt}\vphantom{\underset x{\overset XX}}\sigma _{2}^{\prime \prime}&{},{}&e_{2}^{\prime}\hspace{2pt}\mathllap{\raisebox{-5pt}{\color{black!50}$\lrcorner$}}\mathllap{\raisebox{6pt}{\color{black!50}$\urcorner$}}\vphantom{\underset x{\overset XX}} &{}\Downarrow _{n_{2 3}}{}& \mathrlap{\raisebox{-5pt}{\color{black!50}$\llcorner$}}\mathrlap{\raisebox{6pt}{\color{black!50}$\ulcorner$}}\hspace{2pt}\vphantom{\underset x{\overset XX}}\sigma _{2}^{\prime \prime \prime}&{},{}&v_{2}^{\prime}\hspace{2pt}\mathllap{\raisebox{-5pt}{\color{black!50}$\lrcorner$}}\mathllap{\raisebox{6pt}{\color{black!50}$\urcorner$}}\vphantom{\underset x{\overset XX}}          & {{\color{\colorTEXT}\textnormal{{\textit{(H5.3)}}}}}
           \end{array}}}}}
     \item  and we also have:
        {{\color{\colorMATH}\ensuremath{\mathit{n_{1} = n_{1 1} + n_{1 2} + n_{1 3} + 1}}}},
        {{\color{\colorMATH}\ensuremath{\mathit{n_{2} = n_{2 1} + n_{2 2} + n_{2 3} + 1}}}},
        {{\color{\colorMATH}\ensuremath{\mathit{\sigma _{1}^{\prime} = \sigma _{1}^{\prime \prime \prime}}}}},
        {{\color{\colorMATH}\ensuremath{\mathit{\sigma _{2}^{\prime} = \sigma _{2}^{\prime \prime \prime}}}}},
        {{\color{\colorMATH}\ensuremath{\mathit{v_{1} = v_{1}^{\prime}}}}} and
        {{\color{\colorMATH}\ensuremath{\mathit{v_{2} = v_{2}^{\prime}}}}}.
        By IH ({{\color{\colorMATH}\ensuremath{\mathit{n = n}}}} decreasing), {\textit{(H1)}}, {\textit{(H2)}}, {\textit{(H3)}}, {\textit{(H4.1)}} and {\textit{(H5.1)}} we have:
        {{\color{\colorMATH}\ensuremath{\mathit{n_{1 1} = n_{2 1}}}}}                    {\textit{(IH.1.C1)}};
        {{\color{\colorMATH}\ensuremath{\mathit{\sigma _{1}^{\prime} \sim ^{\Sigma }_{n-n_{1 1}} \sigma _{2}^{\prime}}}}}          {\textit{(IH.1.C2)}}; and
        {{\color{\colorMATH}\ensuremath{\mathit{\langle \lambda x.\hspace*{0.33em}e_{1}^{\prime}\mathrel{|}\rho _{1}^{\prime} \sim ^{\Sigma }_{n-n_{1 1}} v_{2 1}}}}}  {\textit{(IH.1.C3)}}.
        Note the following facts:
        {{\color{\colorMATH}\ensuremath{\mathit{n_{1 2} \leq  n-n_{1 1}}}}}    {\textit{(F1)}}; and
        {{\color{\colorMATH}\ensuremath{\mathit{\rho _{1} \sim ^{\Sigma }_{n-n_{1 1}}}}}} {\textit{(F2)}}.
        {\textit{(F1)}} follows from {\textit{(H3)}} and {{\color{\colorMATH}\ensuremath{\mathit{n_{1} = n_{1 1} + n_{1 2} + n_{1 3} + 1}}}}.
        {\textit{(F2)}} follows from {\textit{(H1)}} and {\nameref{thm:step-index-weakening}.1}.
        By IH ({{\color{\colorMATH}\ensuremath{\mathit{n = n-n_{1 1}}}}} decreasing), {\textit{(H2)}}, {\textit{(IH.1.C2)}}, {\textit{(F1)}}, {\textit{(H4.2)}} and {\textit{(H5.2)}} we have:
        {{\color{\colorMATH}\ensuremath{\mathit{n_{1 2} = n_{2 2}}}}}               {\textit{(IH.2.C1)}};
        {{\color{\colorMATH}\ensuremath{\mathit{\sigma _{1}^{\prime \prime} \sim ^{\Sigma }_{n-n_{1 1}-n_{1 2}} \sigma _{2}^{\prime \prime}}}}} {\textit{(IH.2.C2)}}; and
        {{\color{\colorMATH}\ensuremath{\mathit{v_{1} \sim ^{\Sigma }_{n-n_{1 1}-n_{1 2}} v_{2}}}}}   {\textit{(IH.2.C3)}}.
        Note the following facts, each of which follow from {\textit{(H3)}}
        and {{\color{\colorMATH}\ensuremath{\mathit{n_{1} = n_{1 1} + n_{1 2} + n_{1 3} + 1}}}}:
        {{\color{\colorMATH}\ensuremath{\mathit{n_{1 3} \leq  n-n_{1 1}-n_{1 2}}}}}   {\textit{(F3)}}; and
        {{\color{\colorMATH}\ensuremath{\mathit{n-n_{1 1}-n_{1 2}-n_{1 3} > 0}}}} {\textit{(F4)}}.
        Also note the following facts which follow from {\textit{(IH.1.C3)}},
        {\textit{(IH.2.C2)}}, {\textit{(IH.2.C3)}}, {\textit{(F3)}} and {\textit{(F4)}}:
        {{\color{\colorMATH}\ensuremath{\mathit{n_{1 3} = n_{2 3}}}}}                       {\textit{(F4.C1)}};
        {{\color{\colorMATH}\ensuremath{\mathit{\sigma _{1}^{\prime \prime \prime} \sim ^{\Sigma }_{n-n_{1 1}-n_{1 2}-n_{1 3}-1} \sigma _{2}^{\prime \prime \prime}}}}}   {\textit{(F4.C2)}}; and
        {{\color{\colorMATH}\ensuremath{\mathit{v_{1}^{\prime} \sim ^{\sigma }_{n-n_{1 1}-n_{1 2}-n_{1 3}-1} v_{2}^{\prime}}}}}   {\textit{(F4.C3)}}.
        To show:
        {\textit{(C1)}}: {{\color{\colorMATH}\ensuremath{\mathit{n_{1 1} + n_{1 2} + n_{1 3} + 1 = n_{2 1} + n_{2 2} + n_{2 3} + 1}}}};
        {\textit{(C2)}}: {{\color{\colorMATH}\ensuremath{\mathit{\sigma _{1}^{\prime \prime \prime} \sim ^{\Sigma }_{n-n_{1 1}-n_{1 2}-n_{1 3}-1} \sigma _{2}^{\prime \prime \prime}}}}}; and
        {\textit{(C3)}}: {{\color{\colorMATH}\ensuremath{\mathit{v_{1}^{\prime} \sim ^{\Sigma }_{n-n_{1 1}-n_{1 2}-n_{1 3}-1} v_{2}^{\prime}}}}}.
        {\textit{(C1)}} is immediate from {\textit{(IH.1.C1)}}, {\textit{(IH.2.C1)}} and {\textit{(F4.C1)}}.
        {\textit{(C2)}} is immediate from {\textit{(F4.C2)}} and {\nameref{thm:step-index-weakening}.2}.
        {\textit{(C3)}} is immediate from {\textit{(F4.C3)}} and {\nameref{thm:step-index-weakening}.3}.
     \end{itemize} 
  \item  \begin{itemize}[label={},leftmargin=0pt]\item  {\textbf{Case}} {{\color{\colorMATH}\ensuremath{\mathit{n=n+1}}}} {\textbf{and}} {{\color{\colorMATH}\ensuremath{\mathit{e={{\color{\colorSYNTAX}\texttt{ref}}}(e)}}}}: 
     \item  By inversion on {\textit{(H4)}} and {\textit{(H5)}} we have:
     \item  \ {{\color{\colorMATH}\ensuremath{\mathit{\begin{array}{rclclclcl@{\hspace*{1.00em}}l
           } \rho _{1} &{}\vdash {}& \mathrlap{\raisebox{-5pt}{\color{black!50}$\llcorner$}}\mathrlap{\raisebox{6pt}{\color{black!50}$\ulcorner$}}\hspace{2pt}\vphantom{\underset x{\overset XX}}\sigma _{1} &{},{}&e\hspace{2pt}\mathllap{\raisebox{-5pt}{\color{black!50}$\lrcorner$}}\mathllap{\raisebox{6pt}{\color{black!50}$\urcorner$}}\vphantom{\underset x{\overset XX}} &{}\Downarrow _{n_{1}}{}& \mathrlap{\raisebox{-5pt}{\color{black!50}$\llcorner$}}\mathrlap{\raisebox{6pt}{\color{black!50}$\ulcorner$}}\hspace{2pt}\vphantom{\underset x{\overset XX}}\sigma _{1}^{\prime}&{},{}&v_{1}\hspace{2pt}\mathllap{\raisebox{-5pt}{\color{black!50}$\lrcorner$}}\mathllap{\raisebox{6pt}{\color{black!50}$\urcorner$}}\vphantom{\underset x{\overset XX}}    & {{\color{\colorTEXT}\textnormal{{\textit{(H4.1)}}}}}
           \cr  \multicolumn{5}{r}{\ell _{1}} &{}={}& \multicolumn{3}{l}{{\text{alloc}}({\text{dom}}(\sigma _{1}^{\prime}))} & {{\color{\colorTEXT}\textnormal{{\textit{(H4.2)}}}}}
           \cr  \rho _{2} &{}\vdash {}& \mathrlap{\raisebox{-5pt}{\color{black!50}$\llcorner$}}\mathrlap{\raisebox{6pt}{\color{black!50}$\ulcorner$}}\hspace{2pt}\vphantom{\underset x{\overset XX}}\sigma _{2} &{},{}&e\hspace{2pt}\mathllap{\raisebox{-5pt}{\color{black!50}$\lrcorner$}}\mathllap{\raisebox{6pt}{\color{black!50}$\urcorner$}}\vphantom{\underset x{\overset XX}} &{}\Downarrow _{n_{2}}{}& \mathrlap{\raisebox{-5pt}{\color{black!50}$\llcorner$}}\mathrlap{\raisebox{6pt}{\color{black!50}$\ulcorner$}}\hspace{2pt}\vphantom{\underset x{\overset XX}}\sigma _{2}^{\prime}&{},{}&v_{2}\hspace{2pt}\mathllap{\raisebox{-5pt}{\color{black!50}$\lrcorner$}}\mathllap{\raisebox{6pt}{\color{black!50}$\urcorner$}}\vphantom{\underset x{\overset XX}}    & {{\color{\colorTEXT}\textnormal{{\textit{(H5.1)}}}}}
           \cr  \multicolumn{5}{r}{\ell _{2}} &{}={}& \multicolumn{3}{l}{{\text{alloc}}({\text{dom}}(\sigma _{2}^{\prime}))} & {{\color{\colorTEXT}\textnormal{{\textit{(H4.2)}}}}}
           \end{array}}}}}
     \item  and we also have:
        {{\color{\colorMATH}\ensuremath{\mathit{n_{1} = n_{1}}}}},
        {{\color{\colorMATH}\ensuremath{\mathit{n_{2} = n_{2}}}}},
        {{\color{\colorMATH}\ensuremath{\mathit{\sigma _{1}^{\prime} = \{ \ell \mapsto v_{1}\} \uplus \sigma _{1}^{\prime}}}}},
        {{\color{\colorMATH}\ensuremath{\mathit{\sigma _{2}^{\prime} = \{ \ell \mapsto v_{2}\} \uplus \sigma _{2}^{\prime}}}}},
        {{\color{\colorMATH}\ensuremath{\mathit{v_{1} = \ell _{1}}}}} and
        {{\color{\colorMATH}\ensuremath{\mathit{v_{2} = \ell _{2}}}}}.
        By IH ({{\color{\colorMATH}\ensuremath{\mathit{n = n}}}} decreasing), {\textit{(H1)}}, {\textit{(H2)}}, {\textit{(H3)}}, {\textit{(H4.1)}} and {\textit{(H5.1)}} we have:
        {{\color{\colorMATH}\ensuremath{\mathit{n_{1} = n_{2}}}}}            {\textit{(IH.1.C1)}};
        {{\color{\colorMATH}\ensuremath{\mathit{\sigma _{1}^{\prime} \sim ^{\Sigma }_{n-n_{1}} \sigma _{2}^{\prime}}}}} {\textit{(IH.1.C2)}}; and
        {{\color{\colorMATH}\ensuremath{\mathit{v_{1} \sim ^{\Sigma }_{n-n_{1}} v_{2}}}}}   {\textit{(IH.1.C3)}}.
        Because {{\color{\colorMATH}\ensuremath{\mathit{\sigma _{1} \sim ^{\Sigma }_{n-n_{1}} \sigma _{2}^{\prime}}}}}, we know {{\color{\colorMATH}\ensuremath{\mathit{{\text{dom}}(\sigma _{1}) = {\text{dom}}(\sigma _{2})}}}} and therefore {{\color{\colorMATH}\ensuremath{\mathit{\ell _{1} = \ell _{2}}}}}.
        To show:
        {\textit{(C1)}}: {{\color{\colorMATH}\ensuremath{\mathit{n_{1} = n_{2}}}}};
        {\textit{(C2)}}: {{\color{\colorMATH}\ensuremath{\mathit{\{ \ell \mapsto v_{1}\} \uplus \sigma _{1}^{\prime} \sim ^{\Sigma }_{n-n_{1}} \{ \ell \mapsto v_{2}\} \uplus \sigma _{2}^{\prime}}}}}; and
        {\textit{(C3)}}: {{\color{\colorMATH}\ensuremath{\mathit{\ell _{1} \sim ^{\Sigma }_{n-n_{1}} \ell _{2}}}}}.
        {\textit{(C1)}} is immediate from {\textit{(IH.1.C1)}}.
        {\textit{(C2)}} is immediate from {\textit{(IH.1.C2)}}, {\textit{(IH.1.C3)}} and the definition of {{\color{\colorMATH}\ensuremath{\mathit{\sigma  \sim _{n}^{\Sigma } \sigma }}}}.
        {\textit{(C3)}} is immediate by definition of {{\color{\colorMATH}\ensuremath{\mathit{v \sim _{n}^{\Sigma } v}}}} and {{\color{\colorMATH}\ensuremath{\mathit{\ell _{1} = \ell _{2}}}}}.
     \end{itemize} 
  \item  \begin{itemize}[label={},leftmargin=0pt]\item  {\textbf{Case}} {{\color{\colorMATH}\ensuremath{\mathit{n=n+1}}}} {\textbf{and}} {{\color{\colorMATH}\ensuremath{\mathit{e={!}e}}}}: 
     \item  By inversion on {\textit{(H4)}} and {\textit{(H5)}} we have:
     \item  \ {{\color{\colorMATH}\ensuremath{\mathit{\begin{array}{rclclclcl@{\hspace*{1.00em}}l
           } \rho _{1} &{}\vdash {}& \mathrlap{\raisebox{-5pt}{\color{black!50}$\llcorner$}}\mathrlap{\raisebox{6pt}{\color{black!50}$\ulcorner$}}\hspace{2pt}\vphantom{\underset x{\overset XX}}\sigma _{1} &{},{}&e\hspace{2pt}\mathllap{\raisebox{-5pt}{\color{black!50}$\lrcorner$}}\mathllap{\raisebox{6pt}{\color{black!50}$\urcorner$}}\vphantom{\underset x{\overset XX}} &{}\Downarrow _{n_{1}}{}& \mathrlap{\raisebox{-5pt}{\color{black!50}$\llcorner$}}\mathrlap{\raisebox{6pt}{\color{black!50}$\ulcorner$}}\hspace{2pt}\vphantom{\underset x{\overset XX}}\sigma _{1}^{\prime}&{},{}&\ell _{1}\hspace{2pt}\mathllap{\raisebox{-5pt}{\color{black!50}$\lrcorner$}}\mathllap{\raisebox{6pt}{\color{black!50}$\urcorner$}}\vphantom{\underset x{\overset XX}} & {{\color{\colorTEXT}\textnormal{{\textit{(H4.1)}}}}}
           \cr  \rho _{2} &{}\vdash {}& \mathrlap{\raisebox{-5pt}{\color{black!50}$\llcorner$}}\mathrlap{\raisebox{6pt}{\color{black!50}$\ulcorner$}}\hspace{2pt}\vphantom{\underset x{\overset XX}}\sigma _{2} &{},{}&e\hspace{2pt}\mathllap{\raisebox{-5pt}{\color{black!50}$\lrcorner$}}\mathllap{\raisebox{6pt}{\color{black!50}$\urcorner$}}\vphantom{\underset x{\overset XX}} &{}\Downarrow _{n_{2}}{}& \mathrlap{\raisebox{-5pt}{\color{black!50}$\llcorner$}}\mathrlap{\raisebox{6pt}{\color{black!50}$\ulcorner$}}\hspace{2pt}\vphantom{\underset x{\overset XX}}\sigma _{2}^{\prime}&{},{}&\ell _{2}\hspace{2pt}\mathllap{\raisebox{-5pt}{\color{black!50}$\lrcorner$}}\mathllap{\raisebox{6pt}{\color{black!50}$\urcorner$}}\vphantom{\underset x{\overset XX}} & {{\color{\colorTEXT}\textnormal{{\textit{(H5.1)}}}}}
           \end{array}}}}}
     \item  and we also have:
        {{\color{\colorMATH}\ensuremath{\mathit{n_{1} = n_{1}}}}},
        {{\color{\colorMATH}\ensuremath{\mathit{n_{2} = n_{2}}}}},
        {{\color{\colorMATH}\ensuremath{\mathit{\sigma _{1}^{\prime} = \sigma _{1}^{\prime}}}}},
        {{\color{\colorMATH}\ensuremath{\mathit{\sigma _{2}^{\prime} = \sigma _{2}^{\prime}}}}},
        {{\color{\colorMATH}\ensuremath{\mathit{v_{1} = \sigma _{1}^{\prime}(\ell _{1})}}}} and
        {{\color{\colorMATH}\ensuremath{\mathit{v_{2} = \sigma _{2}^{\prime}(\ell _{2})}}}}.
        By IH ({{\color{\colorMATH}\ensuremath{\mathit{n = n}}}} decreasing), {\textit{(H1)}}, {\textit{(H2)}}, {\textit{(H3)}}, {\textit{(H4.1)}} and {\textit{(H5.1)}} we have:
        {{\color{\colorMATH}\ensuremath{\mathit{n_{1} = n_{2}}}}}            {\textit{(IH.1.C1)}};
        {{\color{\colorMATH}\ensuremath{\mathit{\sigma _{1}^{\prime} \sim ^{\Sigma }_{n-n_{1}} \sigma _{2}^{\prime}}}}} {\textit{(IH.1.C2)}}; and
        {{\color{\colorMATH}\ensuremath{\mathit{\ell _{1} \sim ^{\Sigma }_{n-n_{1}} \ell _{2}}}}}   {\textit{(IH.1.C3)}}.
        Because {{\color{\colorMATH}\ensuremath{\mathit{\ell _{1} \sim ^{\Sigma }_{n-n_{1}} \ell _{2}}}}}, we know {{\color{\colorMATH}\ensuremath{\mathit{\ell _{1} = \ell _{2}}}}} by definition of {{\color{\colorMATH}\ensuremath{\mathit{v \sim _{n}^{\Sigma } v}}}}.
        To show:
        {\textit{(C1)}}: {{\color{\colorMATH}\ensuremath{\mathit{n_{1} = n_{2}}}}};
        {\textit{(C2)}}: {{\color{\colorMATH}\ensuremath{\mathit{\sigma _{1}^{\prime} \sim ^{\Sigma }_{n-n_{1}} \sigma _{2}^{\prime}}}}}; and
        {\textit{(C3)}}: {{\color{\colorMATH}\ensuremath{\mathit{\sigma _{1}^{\prime}(\ell _{1}) \sim ^{\Sigma }_{n-n_{1}} \sigma _{2}^{\prime}(\ell _{2})}}}}.
        {\textit{(C1)}} is immediate from {\textit{(IH.1.C1)}}.
        {\textit{(C2)}} is immediate from {\textit{(IH.1.C2)}}.
        {\textit{(C3)}} is immediate from {\textit{(IH.1.C2)}} and {{\color{\colorMATH}\ensuremath{\mathit{\ell _{1} = \ell _{2}}}}}.
     \end{itemize} 
  \item  \begin{itemize}[label={},leftmargin=0pt]\item  {\textbf{Case}} {{\color{\colorMATH}\ensuremath{\mathit{n=n+1}}}} {\textbf{and}} {{\color{\colorMATH}\ensuremath{\mathit{e=e_{1} \leftarrow  e_{2}}}}}: 
     \item  By inversion on {\textit{(H4)}} and {\textit{(H5)}} we have:
     \item  \ {{\color{\colorMATH}\ensuremath{\mathit{\begin{array}{rclclclcl@{\hspace*{1.00em}}l
           } \rho _{1} &{}\vdash {}& \mathrlap{\raisebox{-5pt}{\color{black!50}$\llcorner$}}\mathrlap{\raisebox{6pt}{\color{black!50}$\ulcorner$}}\hspace{2pt}\vphantom{\underset x{\overset XX}}\sigma _{1} &{},{}&e_{1}\hspace{2pt}\mathllap{\raisebox{-5pt}{\color{black!50}$\lrcorner$}}\mathllap{\raisebox{6pt}{\color{black!50}$\urcorner$}}\vphantom{\underset x{\overset XX}} &{}\Downarrow _{n_{1 1}}{}& \mathrlap{\raisebox{-5pt}{\color{black!50}$\llcorner$}}\mathrlap{\raisebox{6pt}{\color{black!50}$\ulcorner$}}\hspace{2pt}\vphantom{\underset x{\overset XX}}\sigma _{1}^{\prime}&{},{}&\ell _{1}\hspace{2pt}\mathllap{\raisebox{-5pt}{\color{black!50}$\lrcorner$}}\mathllap{\raisebox{6pt}{\color{black!50}$\urcorner$}}\vphantom{\underset x{\overset XX}} & {{\color{\colorTEXT}\textnormal{{\textit{(H4.1)}}}}}
           \cr  \rho _{1} &{}\vdash {}& \mathrlap{\raisebox{-5pt}{\color{black!50}$\llcorner$}}\mathrlap{\raisebox{6pt}{\color{black!50}$\ulcorner$}}\hspace{2pt}\vphantom{\underset x{\overset XX}}\sigma _{1}^{\prime}&{},{}&e_{2}\hspace{2pt}\mathllap{\raisebox{-5pt}{\color{black!50}$\lrcorner$}}\mathllap{\raisebox{6pt}{\color{black!50}$\urcorner$}}\vphantom{\underset x{\overset XX}} &{}\Downarrow _{n_{1 2}}{}& \mathrlap{\raisebox{-5pt}{\color{black!50}$\llcorner$}}\mathrlap{\raisebox{6pt}{\color{black!50}$\ulcorner$}}\hspace{2pt}\vphantom{\underset x{\overset XX}}\sigma _{1}^{\prime \prime}&{},{}&\ell _{2}\hspace{2pt}\mathllap{\raisebox{-5pt}{\color{black!50}$\lrcorner$}}\mathllap{\raisebox{6pt}{\color{black!50}$\urcorner$}}\vphantom{\underset x{\overset XX}} & {{\color{\colorTEXT}\textnormal{{\textit{(H4.2)}}}}}
           \cr  \rho _{2} &{}\vdash {}& \mathrlap{\raisebox{-5pt}{\color{black!50}$\llcorner$}}\mathrlap{\raisebox{6pt}{\color{black!50}$\ulcorner$}}\hspace{2pt}\vphantom{\underset x{\overset XX}}\sigma _{2} &{},{}&e_{1}\hspace{2pt}\mathllap{\raisebox{-5pt}{\color{black!50}$\lrcorner$}}\mathllap{\raisebox{6pt}{\color{black!50}$\urcorner$}}\vphantom{\underset x{\overset XX}} &{}\Downarrow _{n_{2 1}}{}& \mathrlap{\raisebox{-5pt}{\color{black!50}$\llcorner$}}\mathrlap{\raisebox{6pt}{\color{black!50}$\ulcorner$}}\hspace{2pt}\vphantom{\underset x{\overset XX}}\sigma _{2}^{\prime}&{},{}&v_{1}\hspace{2pt}\mathllap{\raisebox{-5pt}{\color{black!50}$\lrcorner$}}\mathllap{\raisebox{6pt}{\color{black!50}$\urcorner$}}\vphantom{\underset x{\overset XX}} & {{\color{\colorTEXT}\textnormal{{\textit{(H5.1)}}}}}
           \cr  \rho _{2} &{}\vdash {}& \mathrlap{\raisebox{-5pt}{\color{black!50}$\llcorner$}}\mathrlap{\raisebox{6pt}{\color{black!50}$\ulcorner$}}\hspace{2pt}\vphantom{\underset x{\overset XX}}\sigma _{2}^{\prime}&{},{}&e_{2}\hspace{2pt}\mathllap{\raisebox{-5pt}{\color{black!50}$\lrcorner$}}\mathllap{\raisebox{6pt}{\color{black!50}$\urcorner$}}\vphantom{\underset x{\overset XX}} &{}\Downarrow _{n_{2 2}}{}& \mathrlap{\raisebox{-5pt}{\color{black!50}$\llcorner$}}\mathrlap{\raisebox{6pt}{\color{black!50}$\ulcorner$}}\hspace{2pt}\vphantom{\underset x{\overset XX}}\sigma _{2}^{\prime \prime}&{},{}&v_{2}\hspace{2pt}\mathllap{\raisebox{-5pt}{\color{black!50}$\lrcorner$}}\mathllap{\raisebox{6pt}{\color{black!50}$\urcorner$}}\vphantom{\underset x{\overset XX}} & {{\color{\colorTEXT}\textnormal{{\textit{(H5.2)}}}}}
           \end{array}}}}}
     \item  and we also have:
        {{\color{\colorMATH}\ensuremath{\mathit{n_{1} = n_{1 1} + n_{1 2}}}}},
        {{\color{\colorMATH}\ensuremath{\mathit{n_{2} = n_{2 1} + n_{2 2}}}}},
        {{\color{\colorMATH}\ensuremath{\mathit{\sigma _{1}^{\prime} = \sigma _{1}^{\prime \prime}}}}},
        {{\color{\colorMATH}\ensuremath{\mathit{\sigma _{2}^{\prime} = \sigma _{2}^{\prime \prime}}}}},
        {{\color{\colorMATH}\ensuremath{\mathit{v_{1} = v_{1}}}}} and
        {{\color{\colorMATH}\ensuremath{\mathit{v_{2} = v_{2}}}}}.
        By IH ({{\color{\colorMATH}\ensuremath{\mathit{n = n}}}} decreasing), {\textit{(H1)}}, {\textit{(H2)}}, {\textit{(H3)}}, {\textit{(H4.1)}} and {\textit{(H5.1)}} we have:
        {{\color{\colorMATH}\ensuremath{\mathit{n_{1 1} = n_{2 1}}}}}           {\textit{(IH.1.C1)}};
        {{\color{\colorMATH}\ensuremath{\mathit{\sigma _{1}^{\prime} \sim ^{\Sigma }_{n-n_{1 1}} \sigma _{2}^{\prime}}}}} {\textit{(IH.1.C2)}}; and
        {{\color{\colorMATH}\ensuremath{\mathit{\ell _{1} \sim ^{\Sigma }_{n-n_{1 1}} \ell _{2}}}}}   {\textit{(IH.1.C3)}}.
        Because {{\color{\colorMATH}\ensuremath{\mathit{\ell _{1} \sim ^{\Sigma }_{n-n_{1 1}} \ell _{2}}}}} we know {{\color{\colorMATH}\ensuremath{\mathit{\ell _{1} = \ell _{2}}}}}.
        Note the following facts:
        {{\color{\colorMATH}\ensuremath{\mathit{n_{1 2} \leq  n-n_{1 1}}}}}    {\textit{(F1)}}; and
        {{\color{\colorMATH}\ensuremath{\mathit{\rho _{1} \sim ^{\Sigma }_{n-n_{1 1}}}}}} {\textit{(F2)}}.
        {\textit{(F1)}} follows from {\textit{(H3)}} and {{\color{\colorMATH}\ensuremath{\mathit{n_{1} = n_{1 1} + n_{1 2}}}}}.
        {\textit{(F2)}} follows from {\textit{(H1)}} and {\nameref{thm:step-index-weakening}.1}.
        By IH ({{\color{\colorMATH}\ensuremath{\mathit{n = n-n_{1 1}}}}} decreasing), {\textit{(F2)}}, {\textit{(IH.1.C2)}}, {\textit{(F1)}}, {\textit{(H4.2)}} and {\textit{(H5.2)}} we have:
        {{\color{\colorMATH}\ensuremath{\mathit{n_{1 2} = n_{2 2}}}}}               {\textit{(IH.2.C1)}};
        {{\color{\colorMATH}\ensuremath{\mathit{\sigma _{1}^{\prime \prime} \sim ^{\Sigma }_{n-n_{1 1}-n_{1 2}} \sigma _{2}^{\prime \prime}}}}} {\textit{(IH.2.C2)}}; and
        {{\color{\colorMATH}\ensuremath{\mathit{v_{1} \sim ^{\Sigma }_{n-n_{1 1}-n_{1 2}} v_{2}}}}}   {\textit{(IH.2.C3)}}.
        To show:
        {\textit{(C1)}}: {{\color{\colorMATH}\ensuremath{\mathit{n_{1 1} + n_{1 2} = n_{2 1} + n_{2 2}}}}};
        {\textit{(C2)}}: {{\color{\colorMATH}\ensuremath{\mathit{\sigma _{1}^{\prime \prime}[\ell \mapsto v_{1}] \sim ^{\Sigma }_{n-n_{1 1}-n_{1 2}} \sigma _{2}^{\prime \prime}[\ell \mapsto v_{2}]}}}}; and
        {\textit{(C3)}}: {{\color{\colorMATH}\ensuremath{\mathit{v_{1} \sim ^{\Sigma }_{n-n_{1 1}-n_{1 2}} v_{2}}}}}.
        {\textit{(C1)}} is immediate from {\textit{(IH.1.C1)}} and {\textit{(IH.2.C1)}}.
        {\textit{(C2)}} is immediate from {\textit{(IH.2.C2)}} and {\textit{(IH.2.C3)}}.
        {\textit{(C3)}} is immediate from {\textit{(IH.2.C3)}}.
     \end{itemize} 
  \end{itemize}
\end{proof}

\section{Additional Case Studies}
\subsection{MWEM with Pandas}

The MWEM algorithm \cite{hardt2012simple} constructs a differentially private synthetic dataset that approximates a real dataset. MWEM produces competitive privacy bounds by utilizing a combination of the exponential mechanism, Laplacian/Gaussian noise, and the multiplicative weights update rule.
The algorithm uses these mechanisms iteratively, providing a tight analysis of privacy leakage
via composition.

The inputs to the MWEM are as follows: some uniform or random distribution over a domain ($syn\_data$), some sensitive dataset ($age\_counts$), a query workload, a number of iterations {{\color{\colorMATH}\ensuremath{\mathit{i}}}}, and a privacy budget {{\color{\colorMATH}\ensuremath{\mathit{\epsilon }}}}.

The algorithm works by, at each iteration:
\begin{itemize}
  \item privately selecting a query from the query workload (using the exponential mechanism) whose result on the synthetic dataset greatly differs from the real dataset

\vspace{1em}
\begin{minted}{python}
for t in range(i):
  q = exponential(q_workload,score_fn,eps/(2*i))
  ...
\end{minted}
\vspace{-2ex}
  \item and then privately using the query result on the real dataset to adjust the synthetic dataset towards the truth using the multiplicative weights update rule

\vspace{1em}
\begin{minted}{python}
for t in range(i):
  ...
  syn_data = mwem_step(q, age_counts, syn_data)
\end{minted}
\vspace{-2ex}
\end{itemize}

We present a modified, adaptive MWEM algorithm (Figure~\ref{fig:mwem}) which privately halts execution if the error of the synthetic dataset reaches an acceptably low level before the entire privacy budget is exceeded, conserving the remainder of the budget for other private analyses.

\begin{figure}

\begin{minted}{python}
def mwem_step(query, real_data, syn_data):
  lower, upper = query
  sm = [v for k, v in syn_data.items()]
  total = np.sum(sm)
  q_ans = range_query(real_data, lower, upper)
  real = dduo.renyi_gauss(|$\alpha $|=alpha,|$\epsilon $|=eps,q_ans,sens)
  syn = range_query(syn_data, lower, upper)
  l = [(k, mwem_update(k, x, lower, upper,
        real, syn, total))
      for k, x in syn_data.items()]
  return dict(l)

with dduo.RenyiFilter(alpha,20.0):
  with dduo.RenyiOdometer((alpha,2.0)) as odo:
    while stable < stability_thresh:
      e = err(age_counts,curr_syn)
      curr_noisy_err=dduo.renyi_gauss(|$\alpha $|=alpha,|$\epsilon $|=1.0,e)
      if (curr_noisy_err < thresh):
        stable += step
      else:
        stable = 0
      for t in range(iterations):
        q = exponential(q_workload,score_fn,eps/(2*i))
        curr_syn = mwem_step(q,age_counts,curr_syn)

acc = dduo.renyi_gauss(alpha, eps_acc,
  accuracy(age_counts,curr_syn))
print(f"final accuracy: {acc}")
\end{minted}

\vspace{-1em}
\begin{minted}[frame=lines,bgcolor=mygray]{text}
Odometer_|$(\alpha ,\epsilon )$|(|$\{ \textit{data.csv} \mapsto  (10.0, 0.5)\} $|)
final accuracy: 0.703
\end{minted}

\vspace{-2ex}
\caption{MWEM with Pandas}
\label{fig:mwem}

\end{figure}


\subsection{DiffPrivLib}

DiffPrivLib is library for experimenting with analytics and machine learning with differential privacy in Python by IBM. It provides a comprehensive suite of differentially private mechanisms, tools, and machine learning models.

While DiffPrivLib provides several mechanisms, models and tools for developing private applications, as well as a basic privacy accountant, it lacks the ability to perform
a tight privacy analysis in the context of more sophisticated forms of composition
with dynamic and adaptive privacy tracking. Via integration with
\dduo we are able to gain these abilities with minimal changes to library and program code.

Figure~\ref{fig:dpl1} shows an example of a modified DiffPrivLib program: a private naive Bayes classifier run on the standard iris dataset. The original program has been modified with \dduo hooks to detect sensitivity violations and track privacy cost.

We also present a \dduo instrumented example of differentially private logistic regression with DiffPrivLib (Figure~\ref{fig:dpl2}).

Both of these programs have been modified to perform \emph{adaptively} private clipping. Over
several iterations, clipping parameters are gradually modified to optimize model accuracy.
This form of control flow on probabilistic values is only sound following the adaptive composition strategies that \dduo provides. In order to preserve the privacy budget, such hyperparameter optimization procedures should normally be run on artificial datasets based on domain knowledge.

The changes required for the integration with the DiffPrivLib library consist of 15 lines of \dduo instrumentation code.

\begin{figure}
\begin{minted}{python}
from dduo import sklearn as sk
from dduo import DiffPrivLib as dpl
with dduo.AdvEdOdometer() as odo:
  while noisy_acc < thresh or iters < max_iters:
    prev_bounds = bounds
    bounds = update_bounds(bounds)
    clf = dpl.GNB(bounds=bounds, epsilon=epsilon)
    clf.fit(X_train, y_train)
    prev_acc = noisy_acc
    accuracy = dpl.score(y_test, clf.predict(X_test))
    noisy_acc = dduo.gauss(epsilon_acc,delta,accuracy)
    if noisy_acc < prev_acc:
      bounds = prev_bounds
    iters += 1
dduo.print_privacy_cost()
\end{minted}

\vspace{-1em}
\begin{minted}[frame=lines,bgcolor=mygray]{text}
Odometer_|$(\epsilon ,\delta )$|(|$\{ \textit{data.csv} \mapsto  (0.82, 0.0035)\} $|)
\end{minted}

\vspace{-2ex}
\caption{DiffPrivLib: Naive Bayes Classification}
\label{fig:dpl1}
\end{figure}

\begin{figure}
\begin{minted}{python}
from dduo import sklearn as sk
from dduo import DiffPrivLib as dpl
with dduo.AdvEdOdometer() as odo:
  while noisy_acc < thresh or norm > 0.0:
    dp_clf = dpl.LogisticRegression(epsilon=epsilon,
      data_norm = norm)
    dp_clf.fit(X_train, y_train)
    accuracy = dp_clf.score(X_test, y_test)
    noisy_acc = dduo.gauss(epsilon_acc,delta,accuracy)
    norm -= step
dduo.print_privacy_cost()
\end{minted}

\vspace{-1em}
\begin{minted}[frame=lines,bgcolor=mygray]{text}
Odometer_|$(\epsilon ,\delta )$|(|$\{ \textit{data.csv} \mapsto  (0.53, 0.0015)\} $|)
\end{minted}
\vspace{-2ex}
\caption{DiffPrivLib: Logistic Regression}
\label{fig:dpl2}
\end{figure}

\end{document}